\documentclass[11pt]{article}
\usepackage{amssymb,amsthm,amsmath}
\usepackage{color}
\usepackage[usenames,dvipsnames,svgnames,table]{xcolor}
\usepackage{bbm}

\usepackage[colorlinks=true, linkcolor=red, urlcolor=blue, citecolor=gray]{hyperref}
\usepackage{graphicx}
\usepackage{float}
\usepackage{caption}
\usepackage{subcaption}
\usepackage[margin=1in,footskip=0.25in]{geometry}
\usepackage[bottom]{footmisc}
\usepackage{booktabs}
\usepackage{enumitem}
\usepackage{algorithm}
\usepackage{algpseudocode}
\newcommand{\algoname}[1]{\textnormal{\textsc{#1}}}

\newcommand{\spara}[1]{\smallskip\noindent{\bf #1}}
\newif\ifdraft
\draftfalse
\newcommand{\Haim}[1]{\textcolor{orange}{Haim: #1}}

\newcommand{\Cam}[1]{\textcolor{blue}{Cam: #1}}

\newcommand{\eqdef}{\mathbin{\stackrel{\rm def}{=}}}
\makeatletter
\def\hlinewd#1{%
	\noalign{\ifnum0=`}\fi\hrule \@height #1 \futurelet
	\reserved@a\@xhline}
\makeatother

\newtheorem{theorem}{Theorem}

\newtheorem{corollary}[theorem]{Corollary}
\newtheorem{lemma}[theorem]{Lemma}

\newtheorem{claim}[theorem]{Claim}

\newtheorem{defn}{Definition}
\newtheorem{problem}[defn]{Problem}
\newtheorem{example}[theorem]{Example}

\newtheorem*{rep@theorem}{\rep@title}
\newcommand{\newreptheorem}[2]{%
	\newenvironment{rep#1}[1]{%
		\def\rep@title{#2 \ref{##1}}%
		\begin{rep@theorem}}%
		{\end{rep@theorem}}}
\makeatother
\newreptheorem{theorem}{Theorem}
\newreptheorem{claim}{Claim}

\newcommand{\RR}{\mathbb{R}}
\newcommand{\CC}{\mathbb{C}}

\newcommand{\HH}{\mathcal{H}}
\newcommand{\bs}[1]{\boldsymbol{#1}}
\newcommand{\bv}[1]{\mathbf{#1}}
\newcommand{\wh}{\widehat}
\newcommand{\norm}[1]{\|#1\|}
\newcommand{\opnorm}[1]{\|#1\|_\mathrm{op}}

\DeclareMathOperator{\poly}{poly}

\DeclareMathOperator{\sinc}{sinc}
\DeclareMathOperator{\rank}{rank}
\DeclareMathOperator*{\argmin}{arg\,min}
\newcommand{\Kmu}{\mathcal{K}_{\mu}}
\newcommand{\Gmu}{\mathcal{G}_\mu}
\newcommand{\Fmu}{\mathcal{F}_\mu}

\newcommand{\Imu}{\mathcal{I}_\mu}

\newcommand{\smu}{s_{\mu,\epsilon}}
\newcommand{\tmu}{\tau_{\mu,\epsilon}}
\newcommand{\ttmu}{\tilde{\tau}_{\mu,\epsilon}}
\newcommand{\btmu}{\bar{\tau}_{\mu,\epsilon}}
\newcommand{\tsmu}{\tilde{s}_{\mu,\epsilon}}

\newcommand{\BTC}{\mathbb{B}_{TC}}
\newcommand{\BHS}{\mathbb{B}_{HS}}

\newcommand{\Fmut}{\mathcal{F}_{\mu,t}}

\newcommand{\bFmu}{\mathcal{\bar F}_\mu}
\newcommand{\bPmu}{\mathcal{\bar P}_\mu}

\newcommand{\E}{\mathbb{E}}

\newcommand{\ce}{72}
\newcommand{\cef}{12}
\newcommand{\ceu}{c_p}

\DeclareMathOperator{\tr}{tr}

\DeclareMathOperator{\erf}{erf}
\DeclareMathOperator{\range}{range}

  \usepackage{nth}
  \usepackage{intcalc}

\title{A Universal Sampling Method \\ for Reconstructing Signals with Simple Fourier Transforms}

\author{
	\and Haim Avron\\ \small Tel Aviv University\\ \small \texttt{haimav@post.tau.ac.il}
	\and\and
	Michael Kapralov\\ \small EPFL\\ \small \texttt{michael.kapralov@epfl.ch}
	\and
	Cameron Musco\\ \small Microsoft Research\\ \small \texttt{camusco@microsoft.com}
	\and\and
	Christopher Musco\\ \small Princeton University\\ \small \texttt{cmusco@cs.princeton.edu}
	\and\and
	Ameya Velingker\\ \small Google Research\\ \small \texttt{ameyav@google.com}
	\and
	Amir Zandieh\\ \small EPFL\\ \small \texttt{amir.zandieh@epfl.ch}
}

\begin{document}

\maketitle
\begin{abstract}
	Reconstructing continuous signals based on a small number of discrete samples is a fundamental problem across science and engineering. In practice, we are often interested in signals with ``simple'' Fourier structure -- e.g., those involving frequencies within a bounded range, a small number of frequencies, or a few blocks of frequencies.\footnote{I.e. bandlimited, sparse, and multiband signals, respectively.} More broadly, any prior knowledge about a signal's Fourier power spectrum can constrain its complexity.  Intuitively, signals with more highly constrained Fourier structure require fewer samples to reconstruct.
	
	We formalize this intuition by showing that, roughly, a continuous signal from a given class can be approximately reconstructed using a number of samples proportional to the \emph{statistical dimension} of the allowed power spectrum of that class. We prove that, in nearly all settings, this natural measure tightly characterizes the sample complexity of signal reconstruction.
	
	Surprisingly, we also show that, up to logarithmic factors, a universal non-uniform sampling strategy can achieve this optimal complexity for \emph{any class of signals}. We present a simple, efficient, and general algorithm for recovering a signal from the samples taken. For bandlimited and sparse signals, our method matches the state-of-the-art. At the same time, it gives the first computationally and sample efficient solution to a broad range of problems, including multiband signal reconstruction and kriging and Gaussian process regression tasks in one dimension. 
	%\Haim{Can we really say that we have the ``first computationally and sample efficient solution to a broad range of problems, including multiband signal reconstruction and kriging and Gaussian process regression tasks'' when we solve only the 1d case?} \Chris{I cut ``common'' since kriging is usually in 2d, but at least valid in 1d, so still worth mentioning.} \Haim{I would say that is might be better to just add ``in one dimension'', perhaps in parentheses. It is true that it does tone down our selling points in the abstract, but I feel that our results are strong and impressive enough for it to be OK.}
	
	Our work is based on a novel connection between {randomized linear algebra} and the problem of reconstructing signals with constrained Fourier structure. We extend tools based on {statistical leverage score sampling} and {column-based matrix reconstruction} to the approximation of continuous linear operators that arise in the signal reconstruction problem. We believe that these extensions are of independent interest and serve as a foundation for tackling a broad range of continuous time problems using randomized methods.
\end{abstract}

\thispagestyle{empty}
\clearpage
\setcounter{page}{1}

\section{Introduction}
\label{sec:intro}
 
Consider the following fundamental function fitting problem, pictured in Figure \ref{fig:basic_problem}. We can access a continuous signal $y(t)$ at any time $t \in [0,T]$. We wish to select a finite set of sample times $t_1, \ldots, t_q$ such that, by observing the signal values $y(t_1), \ldots, y(t_q)$ at those samples, we are able to find a good approximation $\tilde{y}$ to $y$ over the entire range $[0,T]$. We also study the problem in a noisy setting, where for each sample $t_i$, we only observe $y(t_i) + n(t_i)$ for some fixed noise function $n$.

\begin{figure}[h]
	\centering
	\captionsetup{width=.85\linewidth}
	\begin{subfigure}[t]{0.48\textwidth}
		\centering
		\includegraphics[width=.7\textwidth]{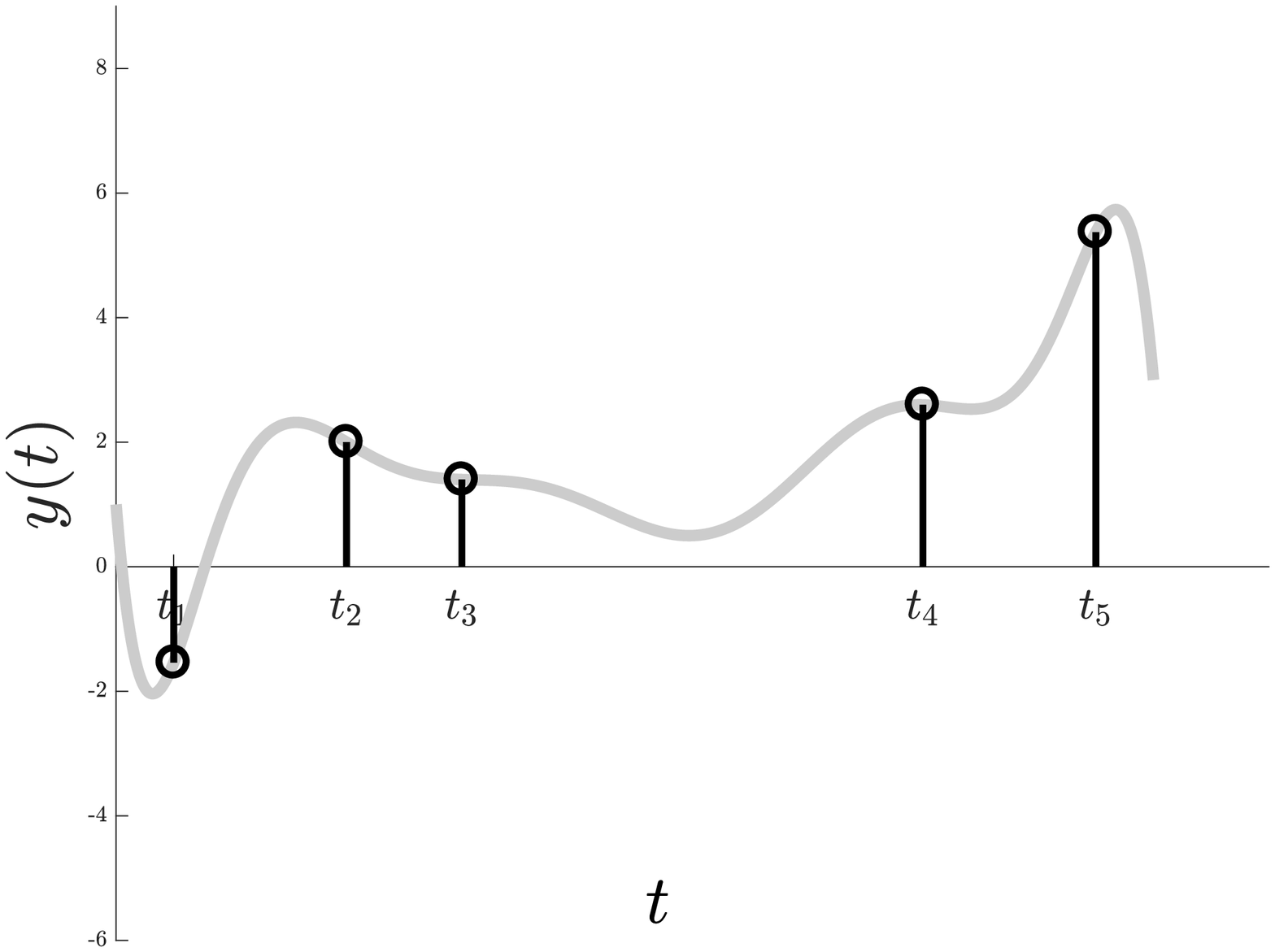}
		\caption{Observed signal $y$ sampled at times $t_1, \ldots, t_q$.}
	\end{subfigure}
	~
	\begin{subfigure}[t]{0.48\textwidth}
		\centering
		\includegraphics[width=.7\textwidth]{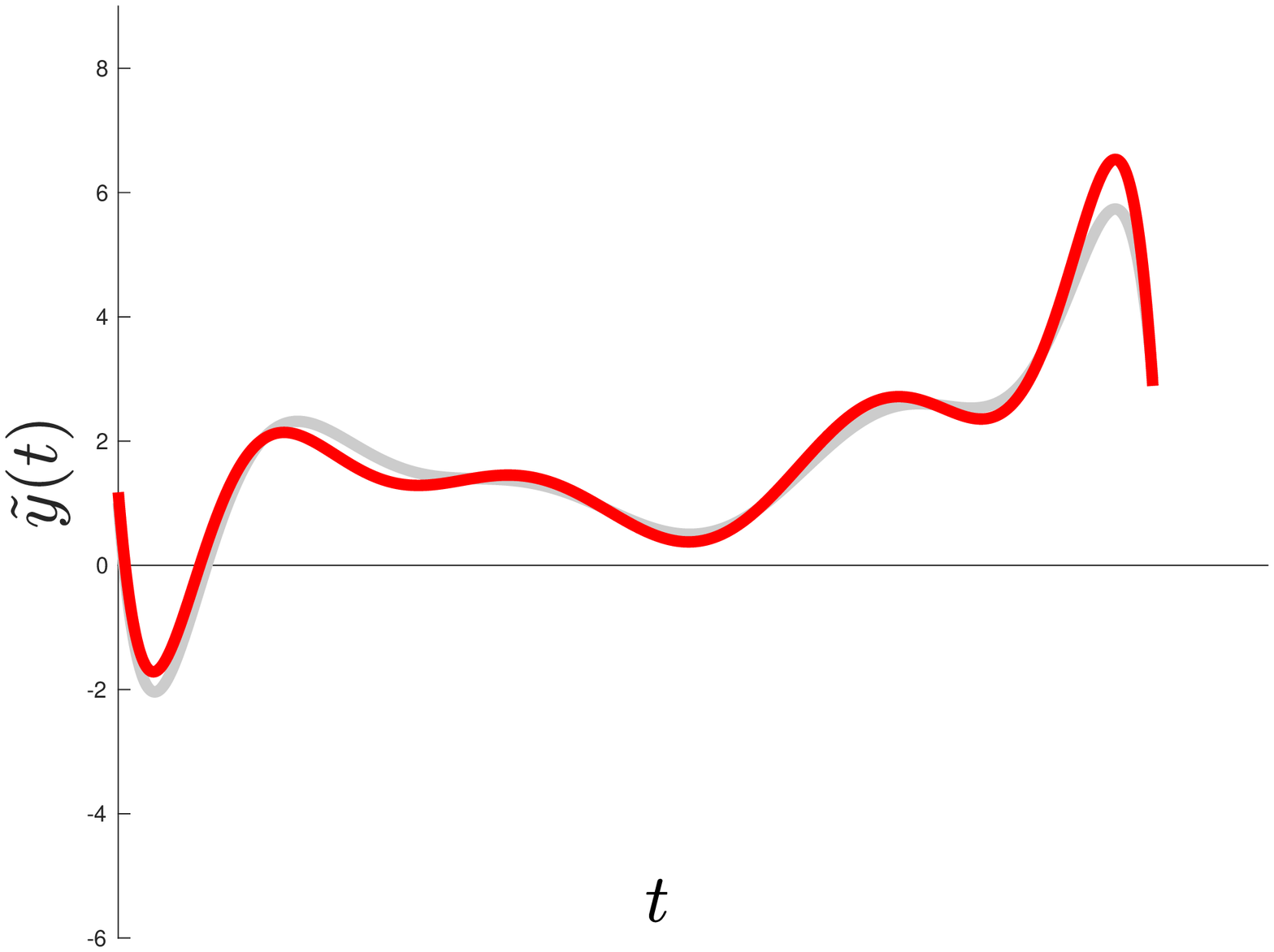}
		\caption{Reconstructed signal $\tilde{y}$ based on samples.}
	\end{subfigure}
	\vspace{-.5em}
	\caption{Our basic function fitting problem requires reconstructing a continuous signal based on a small number of (possibly noisy) discrete samples.}
	\label{fig:basic_problem}
\end{figure}

% In particular, we would like to obtain some representation of a function $\tilde{y}$ with small error:
%\begin{align*}
%	\frac{1}{T}\int_0^T |y(t) - \tilde{y}(t)|^2 \, dt
%\end{align*}

We seek to understand:
\begin{enumerate}
	\item How many samples $q$ are required to approximately reconstruct $y$ and how should we select these samples?
	\item After sampling, how can we find and represent $\tilde{y}$ in a computationally efficient way?
\end{enumerate}
Answering these questions requires assumptions about the underlying signal $y$. In particular, for the information at our samples $t_1, \ldots, t_q$ to be useful in reconstructing $y$ on the entirety of $[0,T]$, the signal must be smooth, structured, or otherwise ``simple'' in some way.

Across science and engineering, by far one of the most common ways in which structure arises is through various assumptions about $\hat{y}$, the \emph{Fourier transform} of $y$:
\begin{align*}
\hat{y}({\xi}) = \int_{-\infty}^\infty y({t}) e^{-2\pi i {t}{\xi}}\,d{t}.
\end{align*}
Our goal is to understand signal reconstruction under natural constraints on the complexity of $\hat{y}$.

\subsection{Classical sampling theory and bandlimited signals}
Classically, the most standard example of such a constraint is requiring $y$ to be \emph{bandlimited}, meaning that $\hat{y}$ is only non-zero for frequencies $\xi$ with $|\xi| \leq F$ for some bandlimit $F$.
In this case, we recall the famous sampling theory of Nyquist, Shannon, and others \cite{Whittaker:1915,Kotelnikov:1933,Nyquist:1928,Shannon:1949}. This theory shows that $y$ can be reconstructed exactly using sinc interpolation (i.e, Whittaker-Shannon interpolation) if $1/2F$ uniformly spaced samples of $y$ are taken per unit of time (the `Nyquist rate').

%Unfortunately, this theory is asymptotic: it requires samples over the entire real line to interpolate $y$, even at a single point in time. For signals available over a finite interval $[0,T]$, sinc interpolation is not a good reconstruction method, either in theory or in practice \cite{Xiao:2001}.
Unfortunately, this theory is asymptotic: it requires infinite samples over the entire real line to interpolate $y$, even at a single point. When a finite number of samples are taken over an interval $[0,T]$, sinc interpolation is not a good reconstruction method, either in theory or in practice \cite{Xiao:2001}.\footnote{Approximation bounds can be obtained by truncating the Whittaker-Shannon method; however, they are weak,  depending \emph{polynomially}, rather than \emph{logarithmically}, on the desired error $\epsilon$ (see Appendix \ref{app:prior_work}, Example \ref{shannonLB}).}

This well-known issue was resolved through a seminal line of work by Slepian, Landau, and Pollak \cite{SlepianPollak:1961,LandauPollak:1961,LandauPollak:1962}, who presented a set of explicit basis functions for interpolating bandlimited functions when a finite number of samples are taken from a finite interval. Their so-called ``prolate spheroidal wave functions'' can be combined with numerical quadrature methods \cite{XiaoRokhlinYarvin:2001,ShkolniskyTygertRokhlin:2006,KarnikZhuWakin:2017} to obtain sample efficient (and computationally efficient) methods for bandlimited reconstruction. Overall, this work shows that roughly $O(FT + \log(1/\epsilon))$ samples from $[0,T]$ are required to interpolate a signal with bandlimit $F$ to accuracy $\epsilon$ on that same interval.\footnote{We formalize our notion of accuracy in Section \ref{sec:prob_statement}.}

\subsection{More general Fourier structure}
\label{sec:first_fourier_structure}
While the aforementioned line of work is beautiful and powerful,
in today's world, we are interested in far more general constraints than bandlimits. For example, there is wide-spread interest in \emph{Fourier-sparse} signals \cite{Donoho:2006}, where $\hat{y}$ is only non-zero for a small number of frequencies, and \emph{multiband} signals, where the Fourier transform is confined to a small number of intervals. Methods for recovering signals in these classes have countless applications in communication, imaging, statistics, and a wide variety of other disciplines \cite{Eldar:2015}.

More generally, in statistical signal processing, a \emph{prior distribution}, specified by some probability measure $\mu$, is often assumed on the frequency content of $y$ \cite{EldarUnser:2006,RamaniVilleUnser:2005}. For signals with bandlimit $F$, $\mu$ would be the uniform probability measure on $[-F,F]$. Alternatively, instead of assuming a hard bandlimit, a zero-centered Gaussian prior on $\hat{y}$ can encode knowledge that higher frequencies are less likely to contribute significantly to $y$, although they may still be present. Such a prior naturally suits a Bayesian approach to signal reconstruction \cite{HandcockStein:1993} and, in fact, is essentially equivalent to assuming $y$ is a stationary stochastic process with a certain covariance function (see  Section \ref{sec:notation} and Appendix \ref{app:bayes}). Under various names, including ``Gaussian process regression'' and ``kriging,'' likelihood estimation under a covariance prior is the dominant statistical approach to fitting continuous signals in many scientific disciplines, from geostatistics to economics to medical imaging \cite{Ripley:2005,RasmussenWilliams06}.

\subsection{Our contributions}
Despite their clear importance, accurate methods for fitting continuous signals under most common Fourier transform priors are not well understood, even 50 years after the groundbreaking work of Slepian, Landau, and Pollak on the bandlimited problem. The only exception is Fourier sparse signals: the \emph{noiseless} interpolation problem can be solved using classical methods  \cite{Prony:1795,Pisarenko:1973,BreslerMacovski:1986}, and recent work has resolved the much more difficult noisy case \cite{ChenKanePrice:2016,ChenPrice:2018}.

In this paper, we address the problem far more generally. Our contributions are as follows:

\begin{enumerate}	
	\item We tightly characterize the information theoretic sample complexity of reconstructing $y$ under any Fourier transform prior, specified by probability measure $\mu$. In essentially all settings, we can prove that this complexity scales nearly linearly with a natural \emph{statistical dimension} parameter associated with $\mu$. See Theorem \ref{thm:informal_sample_complexity}.
	
	%\Haim{We say ``any Fourier transform prior $\mu$''. But: 1) we didn't say what $\mu$ is, so in a sense the sentence is out of context. 2) We do need to assume that $\mu$ is finite, so it is not really ``any'' $\mu$.} \Chris{I think there's enough context since we mentioned $\mu$ is a prior distribution, so it should be clear that it's finite. We also gave examples in 1.2 implying that we will formalize $\mu$ as a measure.} \Haim{A measure can be not-finite (i.e. $\mu(\RR) = \infty$). For example: the Lebesgue measure.}
	
	\item We present a method for sampling from $y$ that achieves the aforementioned statistical dimension bound to within a polylogarithmic factor. Our approach is randomized and \emph{universal}: we prove that it is possible to draw $t_1, \ldots, t_q$ from a fixed non-uniform distribution over $[0,T]$ that is \emph{independent of $\mu$}, i.e., ``spectrum-blind.'' In other words, the same sampling scheme works for bandlimited, sparse, or more general priors. See Theorem  \ref{thm:informal_sample_dist}.
	
	\item We show that $y$ can be recovered from $t_1, \ldots, t_q$   using a simple, efficient, and completely general interpolation method. In particular, we just need to solve a kernel ridge regression problem using  $y(t_1), \ldots, y(t_q)$, with an appropriately chosen kernel function for $\mu$. This method runs in $O(q^3)$ time and is already widely used for signal reconstruction in practice, albeit with suboptimal strategies for choosing $t_1, \ldots, t_q$. See Theorem \ref{thm:informal_main}.
\end{enumerate}

\begin{figure}[h]
	\centering
	\captionsetup{width=1\linewidth}
	\begin{subfigure}[t]{0.2\textwidth}
		\centering
		\includegraphics[width=\textwidth]{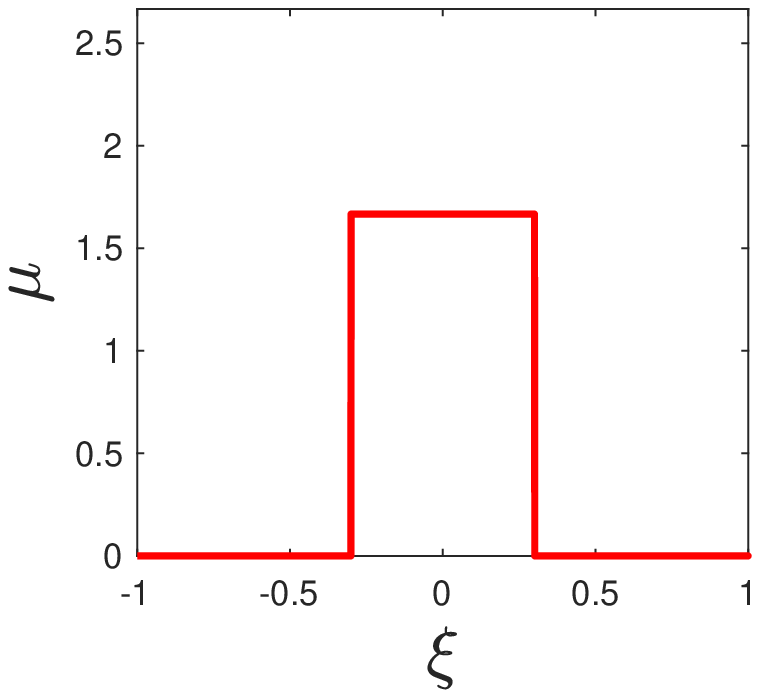}
		\caption{\centering Bandlimited.}
	\end{subfigure}
	~
	\begin{subfigure}[t]{0.2\textwidth}
		\centering
		\includegraphics[width=\textwidth]{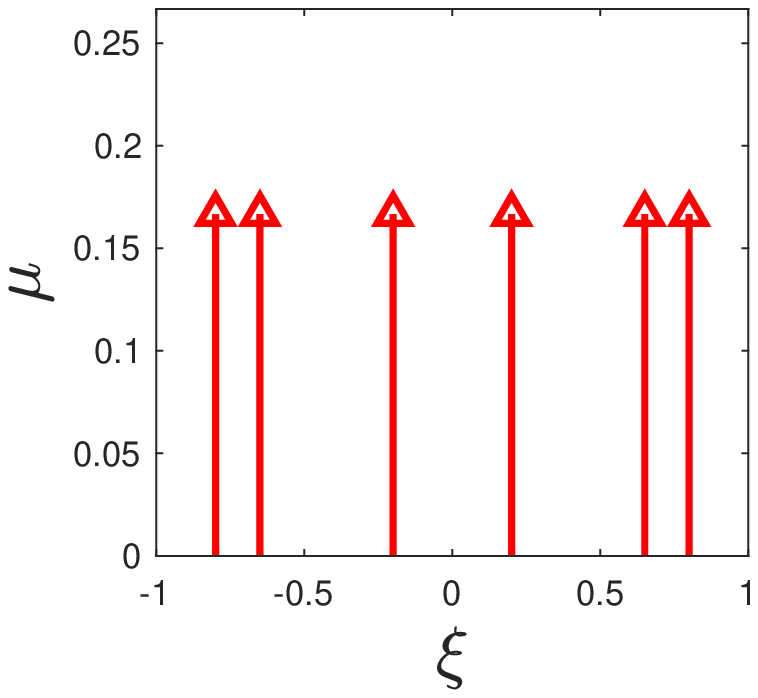}
		\caption{\centering Sparse.}
	\end{subfigure}
	~
	\begin{subfigure}[t]{0.2\textwidth}
		\centering
		\includegraphics[width=\textwidth]{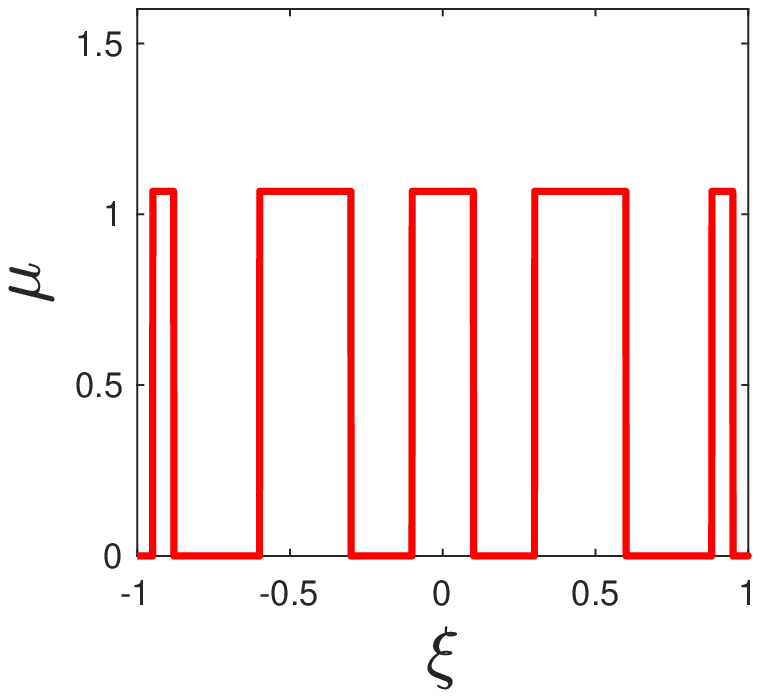}
		\caption{\centering Multiband.}
	\end{subfigure}
	~
	\begin{subfigure}[t]{0.2\textwidth}
		\centering
		\includegraphics[width=\textwidth]{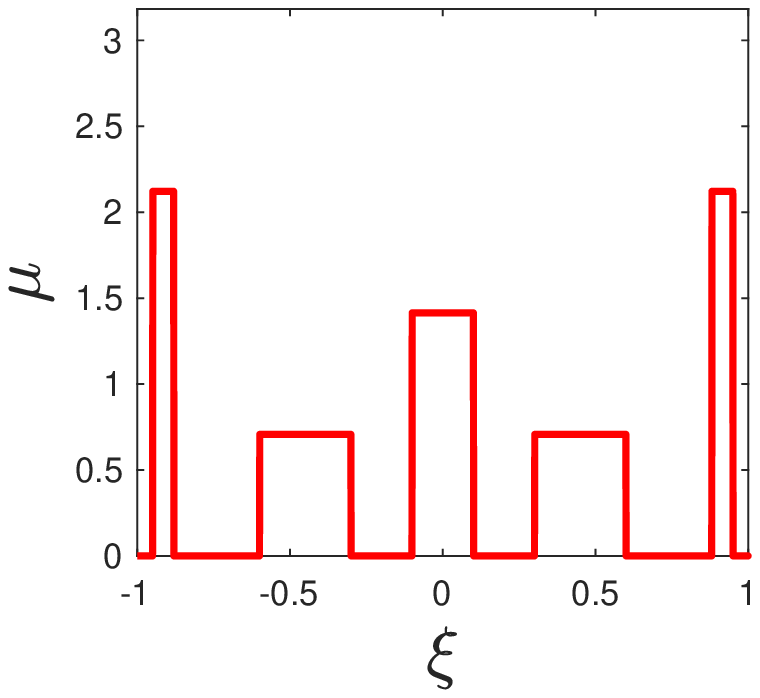}
		\caption{\centering Non-uniform multiband.}
		\label{fig:nonuni_multiband}
	\end{subfigure}
	~
	
	\hspace{-1.2em}
	\begin{subfigure}[t]{0.2\textwidth}
		\centering
		\includegraphics[width=\textwidth]{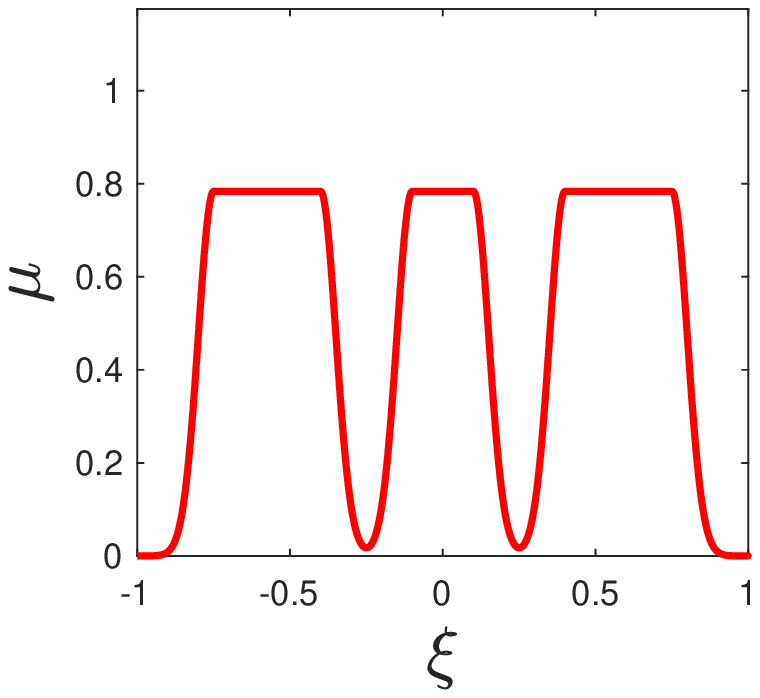}
		\caption{\centering Smooth multiband.}
                \label{fig:smooth_multiband}
	\end{subfigure}
	~
	\begin{subfigure}[t]{0.2\textwidth}
		\centering
		\includegraphics[width=\textwidth]{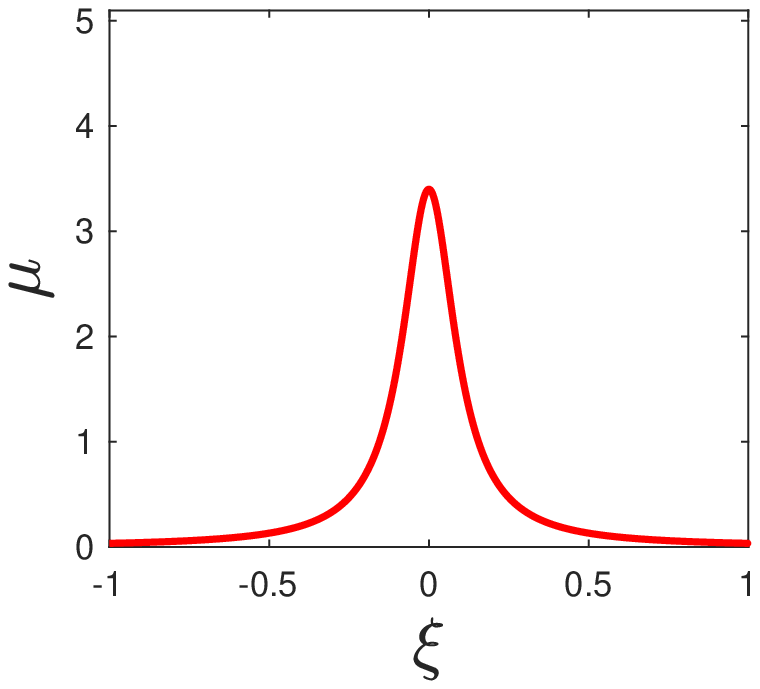}
		\caption{\centering Cauchy-Lorentz.}
	\end{subfigure}
	~
	\begin{subfigure}[t]{0.2\textwidth}
		\centering
		\includegraphics[width=\textwidth]{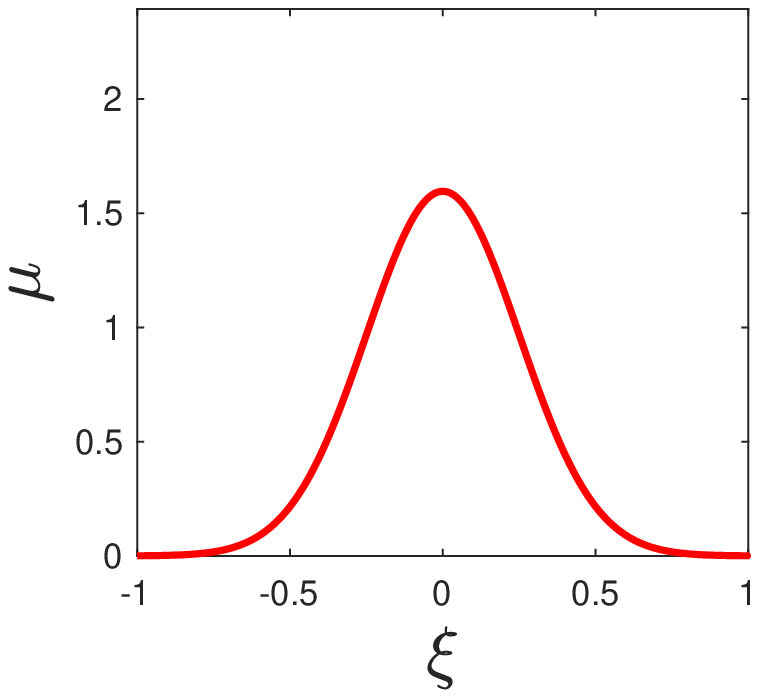}
		\caption{\centering Gaussian.}
	\end{subfigure}
	~
	\begin{subfigure}[t]{0.2\textwidth}
		\centering
		\includegraphics[width=\textwidth]{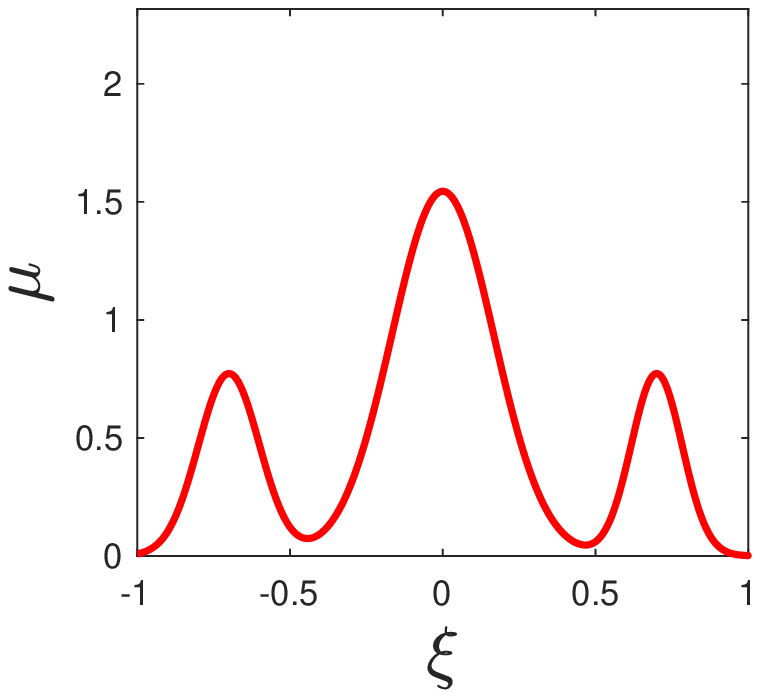}
		\caption{\centering Gaussian mixture.}
	\end{subfigure}
	
	\caption{Examples of Fourier transform ``priors'' induced by various measures $\mu$ (we plot the corresponding density). Our algorithm can reconstruct signals under any of these priors.}
	\label{fig:constraints}
\end{figure}

Overall, this approach gives the first finite sample, provable approximation bounds for all common Fourier-constrained signal reconstruction problems beyond bandlimited and sparse functions.

Our results are obtained by drawing on a rich set of tools from randomized numerical linear algebra, including sampling methods for approximate regression and deterministic column-based low-rank approximation methods \cite{BSS14,CohenNelsonWoodruff16}. Many of these methods view matrices as sums of rank-$1$ outer products and approximate them by sampling or deterministically selecting a subset of these outer products.
We adapt these tools to the approximation of continuous operators, which can be written as the (weak) integral of rank-$1$ operators. For example, our universal time domain sampling distribution is obtained using the notion of \emph{statistical leverage}  \cite{SpielmanSrivastava:2011,AlaouiMahoney:2015,DrineasMahoney:2016}, extended to a continuous Fourier transform operator that arises in the signal reconstruction problem.
%We believe our approach offers a starting point for further exploring the power of randomized methods in signal processing.
We hope that, by extending many of the fundamental contributions of randomized numerical linear algebra to build a toolkit for `randomized operator theory',
%We hope that
our work will offer a starting point for  progress on many signal processing problems using randomized methods.

\section{Formal statement of results}
\label{sec:prob_statement}

As suggested, we formally capture Fourier structure through any probability measure $\mu$ over the  reals.\footnote{Formally, we consider the measure space $(\RR, \mathcal{B},\mu)$ where $\mathcal{B}$ is the Borel $\sigma$-algebra on $\RR$.}
We often refer to $\mu$ as a ``prior'', although we will see that it can be understood beyond the context of Bayesian inference. The simplicity of a set of constraints will be quantified by a natural \emph{statistical dimension} parameter for $\mu$, defined in Section \ref{sec:stat_dimension_def}.

For signals with bandlimit $F$, $\mu$ is the uniform probability measure on $[-F,F]$. For multiband signals, it is uniform on the union of $k$ intervals, while for Fourier-sparse functions, $\mu$ is uniform on a union of $k$ Dirac measures. More general priors are visualized in Figure \ref{fig:constraints}. Those based on Gaussian or Cauchy-Lorentz distributions are especially common in scientific applications, and we will discuss examples shortly. For now, we begin with our main problem formulation.

\begin{problem}
	\label{prob:unformal_interp}
	Given a known probability measure $\mu$ on $\RR$, for any $t \in [0,T]$, define the inverse Fourier transform of a function $g(\xi)$ with respect to $\mu$ as
	\begin{align}
	\label{eq:mu_transform}
	\left[\Fmu^* \,g\right](t) \eqdef \int_{\RR} g(\xi) e^{2\pi i \xi t}\, d\mu(\xi).
	\end{align}
	Suppose our input $y$ can be written as $y = \mathcal{F}^*_{\mu} \, x$ for some frequency domain function $x(\xi)$ and, for any chosen $t$, we can observe $y(t) + n(t)$ for some fixed noise function $n(t)$. Then, for error parameter $\epsilon$, our goal is to recover an approximation $\tilde y$ satisfying
	\begin{align}
	\label{eq:main_guarantee}
	\| y - \tilde{y}\|_T^2 \leq \epsilon \|x\|_{\mu}^2 + C\|n\|_T^2,
	\end{align}
	where $\|x\|_{\mu}^2 \eqdef \int_{\RR} |x(\xi)|^2\, d\mu(\xi)$ is the energy of the function $x$ with respect to $\mu$, while  $\|z\|_T^2\eqdef \frac{1}{T}\int_0^T |z(t)|^2 dt$, so that $\| y - \tilde{y}\|_T^2$ is our mean squared error and $\|n\|_T^2$ is the mean squared noise level. $C \ge 1$ is a fixed positive constant.
\end{problem}

Unlike the  $\|x\|_{\mu}^2$ term in \eqref{eq:main_guarantee}, which we can control by adjusting $\epsilon$, we can never hope to recover $y$ to accuracy better than $\|n\|_T^2$. Accordingly, we consider $\norm{n}_T^2$  to be small and are happy with any solution of Problem \ref{prob:unformal_interp} that is within a constant factor of optimal -- i.e., where $C = O(1)$.

Problem \ref{prob:unformal_interp} captures signal reconstruction under all standard Fourier transform constraints, including bandlimited, multiband, and sparse signals.\footnote{For sparse or multiband signals, Problem \ref{prob:unformal_interp} assumes frequency or band locations are known \emph{a priori}. There has been significant work on algorithms that can recover $y$ when these locations are not known \cite{MishaliEldar:2009,Moitra:2015,PriceSong:2015,ChenKanePrice:2016}. Understanding this more complicated problem in the multiband case is an important future direction.
	% 	However, we do note that our main algorithm will employ ``spectrum blind sampling'' \cite{FengBresler:1996}, meaning that $t_1, \ldots, t_k$ can be chosen without knowledge of frequency locations. In many cases, this property allows our sampling methods to be used in fully blind reconstruction algorithms \cite{MishaliEldar:2009,ChenKanePrice:2016}.
} The error in \eqref{eq:main_guarantee} naturally scales with the average energy of the signal over the allowed frequencies. For more general priors, $\|x\|_{\mu}^2$ will be larger when $y$ contains a significant component of frequencies with low density in $\mu$.\footnote{Informally, decreasing $d\mu(\xi)$ by a factor of $c > 1$ requires increasing $x(\xi)$ by a factor of $c$ to give the same time domain signal. This increases $x(\xi)^2$ by  a factor of $c^2$ and so increases its contribution to $\|x\|_{\mu}^2$ by a factor of $c^2/c = c$.
	%	Accordingly, $x$ requires higher energy if it uses frequencies with low weights in $\mu$.
} For a given number of samples, we would thus incur larger error in \eqref{eq:main_guarantee} in comparison to a signal that uses more ``likely'' frequencies.

As an alternative to Problem \ref{prob:unformal_interp}, we can formulate signal fitting from a Bayesian perspective. We assume that $n$ is independent random noise, and $y$ is a stationary stochastic process with expected power spectral density $\mu$. %We are interested in finding the maximum a posteriori estimator for $y$. We formalize this goal as in Appendix \ref{app:bayes}.
This assumption on $y$'s power spectral density is equivalent to assuming that $y $ has covariance function (a.k.a. autocorrelation) $\hat{\mu}(t)$, which is the type of prior used in kriging and Gaussian process regression.
While we focus on the formulation of Problem \ref{prob:unformal_interp} in this work, we give an informal discussion of the Bayesian setup in Appendix \ref{app:bayes}.
%Ultimately, we show that, combined with our universal sampling scheme, kernel ridge regression provably solves both Problem \ref{prob:unformal_interp} and a natural estimation problem under this Bayesian setup, formalized in Appendix \ref{app:bayes}.

\subsubsection{Examples and applications}

As discussed in Section \ref{sec:first_fourier_structure}, ``hard constraint'' versions of Problem \ref{prob:unformal_interp}, such as bandlimited, sparse, and multiband signal reconstruction, have many applications in communications, imaging, audio, and other areas of engineering. Generalizations of the multiband problem to non-uniform measures (see Figure \ref{fig:nonuni_multiband}) are also useful in various communication problems \cite{MishaliEldar:2010}.

On the other hand, ``soft constraint'' versions of the problem are widely applied in scientific applications. In medical imaging, images are often denoised by setting $\mu$ to a heavy-tailed Cauchy-Lorentz measure on frequencies \cite{Fuderer:1989,LettingtonHong:1995,BourgeoisWajerOrmondt:2001}. This corresponds to assuming an exponential covariance function for spatial correlation. Exponential covariance and its generalization, Mat\'{e}rn covariance, are also common in the earth and geosciences \cite{Ripley:1989,Ripley:2005}, as well as in general image processing \cite{Pesquet-PopescuVehel:2002,RamaniVilleUnser:2006}.%\footnote{Many of our cited applications concern 2-dimensional interpolation problems. In this paper we focus on 1-dimensional interpolation, but believe extending to higher dimensions is an important future direction.}
%\Haim{Well, in general most of the aforementioned applications are 2d and 3d, while we do only 1d, so our statements might seem a bit too ambitious.} \Chris{Added a footnote.}

A Gaussian  prior $\mu$, which corresponds to Gaussian covariance, is also used to model both spatial and temporal correlation in medical imaging \cite{FristonJezzardTurner:1994,WorsleyMarrettNeelin:1996} and is very common in machine learning. Other choices for $\mu$ are practically unlimited. For example, the  popular ArcGIS kriging library also supports the following covariance functions: circular, spherical, tetraspherical, pentaspherical, rational quadratic, hole effect, k-bessel, and j-bessel, and stable \cite{ESRI:2018}.

\subsection{Sample complexity}
\label{sec:stat_dimension_def}

With Problem \ref{prob:unformal_interp} defined, our first goal is to characterize the number of samples required to reconstruct $y$, as a function of the \emph{accuracy parameter} $\epsilon$, the \emph{range} $T$, and the \emph{measure} $\mu$. We do so using what we refer to as the \emph{Fourier statistical dimension} of $\mu$, which corresponds to the standard notion of statistical or `effective dimension' for regularized function fitting problems \cite{HastieTibshiraniFriedman:2002,Zhang:2005}.
\begin{defn}[Fourier statistical dimension]\label{def:statDim}
	For a probability measure $\mu$ on $\RR$ and time length $T$, define the kernel operator $\Kmu: L_2(T)\rightarrow L_2(T)$\footnote{$L_2(T)$ denotes the complex-valued square integrable functions with respect to the uniform measure on $[0,T]$. } as:
	\begin{align}
	\label{eq:kernel_op}
	[\Kmu z](t) \eqdef \int_{\xi \in \RR} e^{2\pi i \xi t} \left[\frac{1}{T}\int_{s \in [0,T]} z(s)e^{-2\pi i \xi s} \, ds\right] d\mu(\xi).
	\end{align}
	Note that $\Kmu$ is self-adjoint, positive semidefinite and trace-class.\footnote{See Section \ref{sec:notation} for a formal explanation of these facts.}
		The Fourier statistical dimension for $\mu$, $T$, and error $\epsilon$ is denoted by $\smu$ and defined as:
	\begin{align}
	\label{eq:stat_dim_def}
	\smu \eqdef \tr(\Kmu (\Kmu + \epsilon \mathcal{I}_T)^{-1}),
	\end{align}	
	where $\mathcal{I}_T$ is the identity operator on $L_2(T)$. Letting $\lambda_i(\Kmu)$ denote the $i^{th}$ largest eigenvalue of $\Kmu$, we may also write
	\begin{align}
	\label{def:stat_dim_sum_version}
	\smu = \sum_{i=1}^\infty \frac{\lambda_i\left(\Kmu\right)}{\lambda_i\left(\Kmu\right) + \epsilon}.
	\end{align}	
\end{defn}
Note that $\Kmu$ and $s_{\mu,\epsilon}$ as defined above, and $\Fmu$ as defined in Problem \ref{prob:unformal_interp} all depend on $T$ and thus could naturally be denoted $\mathcal{F}_{\mu,T}$, $\mathcal{K}_{\mu,T}$, and $s_{\mu,\epsilon,T}$. However, since $T$ is fixed throughout our results, for conciseness we do not use $T$ in our notation for these and related notions.

It is not hard to see that $\smu$ increases as $\epsilon$ decreases, meaning that we will require more samples to obtain a more accurate solution to Problem \ref{prob:unformal_interp}.
The operator $\Kmu$ corresponds to taking the Fourier transform of a time domain input $z(t)$, scaling that transform by $\mu$, and then taking the inverse Fourier transform. Readers familiar with the literature on bandlimited signal reconstruction will recognize $\Kmu$ as the natural generalization of the frequency limiting operator studied in the work of Landau, Slepian, and Pollak on prolate spheroidal wave functions \cite{SlepianPollak:1961,LandauPollak:1961,LandauPollak:1962}. In that work, it is established that a quantity nearly identical to $\smu$ bounds the sample complexity of solving Problem \ref{prob:unformal_interp} for bandlimited functions.

Our first technical result is that this is actually true \emph{for any prior} $\mu$.

\begin{theorem}[Main result, sample complexity]
	\label{thm:informal_sample_complexity} 
	For any probability measure $\mu$, Problem \ref{prob:unformal_interp} can be solved using $q = O\left(\smu\cdot \log \smu\right)$ noisy signal samples $y(t_1) + n(t_1), \ldots, y(t_q) + n(t_q)$.
\end{theorem}
What does Theorem \ref{thm:informal_sample_complexity} imply for common classes of functions with constrained Fourier transforms? Table \ref{tab:stat_dim} includes a list of upper bounds on $\smu$ for many standard priors.

\begin{table}[h]
	\centering
	\renewcommand{\arraystretch}{1.2}
	\begin{tabular}{ l  l  l }
		\toprule
		Fourier prior, $\mu$ & Statistical dimension, $\smu$ & Proof \\ \midrule
		$k$-sparse & $k$ & Since $\Kmu$ has rank $k$. \\
		bandlimited to $[-F,F]$ & $O\left(FT + \log(1/\epsilon)\right)$ & Theorem \ref{thm:bandlimited_leverage_scores}.\\
		multiband, widths $F_1, \ldots, F_s$ & $O\left(\sum_i F_iT + s\log(1/\epsilon)
		\right)$ & Theorem \ref{theorem-multiband-statdim}.\footnotemark\\
		Gaussian, variance $F$ & $O\left(F T\sqrt{\log(1/\epsilon)} + \log(1/\epsilon) \right)$ & Theorem \ref{thm:guassianStatDim}.  \\
		Cauchy-Lorentz, scale $F$ & $O\left(F T\sqrt{1/\epsilon} + \sqrt{1/\epsilon}\right)$ & Theorem \ref{thm:lorentzStatDim}.  \\
		\bottomrule
	\end{tabular}
	\caption {Statistical dimension upper bounds for common Fourier interpolation problems. Our result (Theorem \ref{thm:informal_sample_complexity}) requires $O(\smu\cdot \log \smu)$ samples.} \label{tab:stat_dim}
\end{table}
\footnotetext{Just as Theorem \ref{thm:bandlimited_leverage_scores} intuitively matches the Nyquist sampling rate, Theorem \ref{theorem-multiband-statdim} intuitively matches the Landau rate for asymptotic recovery of multiband functions \cite{Landau:1967a}.}

A complexity of $O(\smu\cdot \log \smu)$ equates to $\tilde{O}(k)$ samples for $k$-sparse functions and $\tilde{O}(FT + \log1/\epsilon)$ for bandlimited functions. 
Up to log factors, these bounds are tight for these well studied problems. In Section \ref{sec:lb}, we show that Theorem \ref{thm:informal_sample_complexity} is actually tight for all common Fourier transform priors:  $\Omega(\smu)$ time points are required for solving Problem \ref{prob:unformal_interp} as long as $\smu$ grows slower than $1/\epsilon^p$ for some $p < 1$. This property holds for all $\mu$ in Table \ref{tab:stat_dim}. 
% -- the Cauchy-Lorentz prior has the worst dependence on $\epsilon$, with $p = 1/2$.
We conjecture that our lower bound can be extended to hold even {without this weak assumption}.

To compliment the sample complexity bound of Theorem \ref{thm:informal_sample_complexity}, we introduce a \emph{universal method} for selecting samples $t_1, \ldots, t_q$ that nearly matches this complexity. Our method selects samples at random, in a way that \emph{does not depend} on the specific prior $\mu$. 

\begin{theorem}[Main result, sampling distribution]
	\label{thm:informal_sample_dist}
	For any sample size $q$, there is a fixed probability density $p_q$ over $[0,T]$ such that, if $q$ time points $t_1, \ldots, t_q$ are selected independently at random according to $p_q$, and $q \geq c\cdot \smu\cdot\log^2 \smu$ for some fixed constant $c$, then it is possible to solve Problem \ref{prob:unformal_interp} with probability 99/100 using the noisy signal samples $y(t_1) + n(t_1), \ldots, y(t_q) + n(t_q)$.\footnote{In Section \ref{sec:puttingTogether}, we formally quantify the tradeoff between success probability and sample complexity.}
\end{theorem}
Theorem \ref{thm:informal_sample_dist} is our main technical contribution. By achieving near optimal sample complexity with a universal distribution, it shows that wide range of Fourier constrained interpolation problems considered in the literature are more closely related than previously understood.

Moreover, $p_q$ (which is formally defined in Theorem \ref{thm:fullBound}) is very simple to describe and sample from. As may be intuitive from results on polynomial interpolation, bandlimited approximation, and other function fitting problems, it is more concentrated towards the endpoints of $[0,T]$, so our sampling scheme selects more time points in these regions. The density is shown in Figure \ref{fig:our_dist}.

\begin{figure}[h]
	\centering
	\captionsetup{width=1\linewidth}
	\begin{subfigure}[t]{0.48\textwidth}
		\centering
		\includegraphics[width=.7\textwidth]{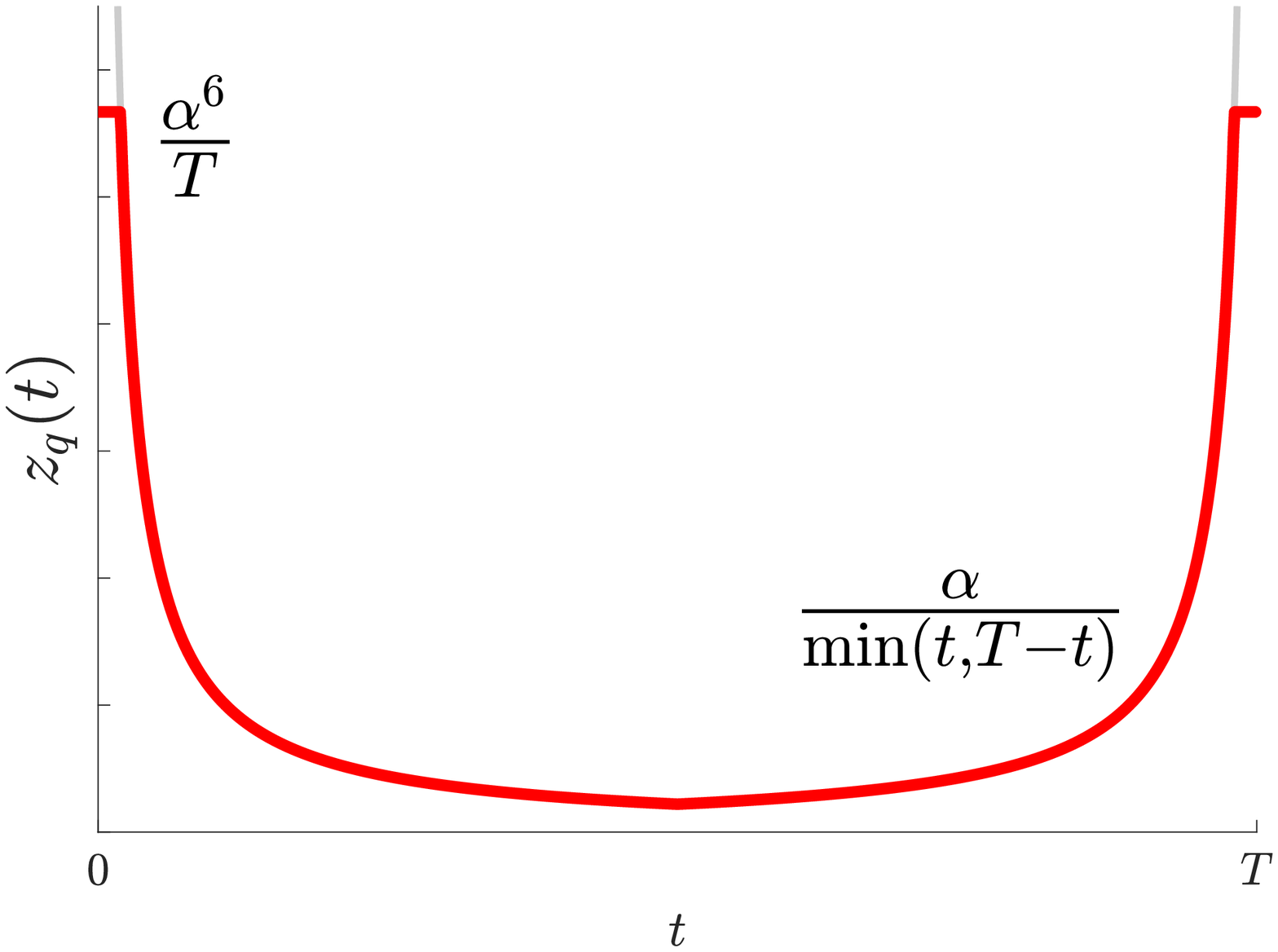}
		\caption{Density for selecting time points.}
	\end{subfigure}
	~
	\begin{subfigure}[t]{0.48\textwidth}
		\centering
		\includegraphics[width=.7\textwidth]{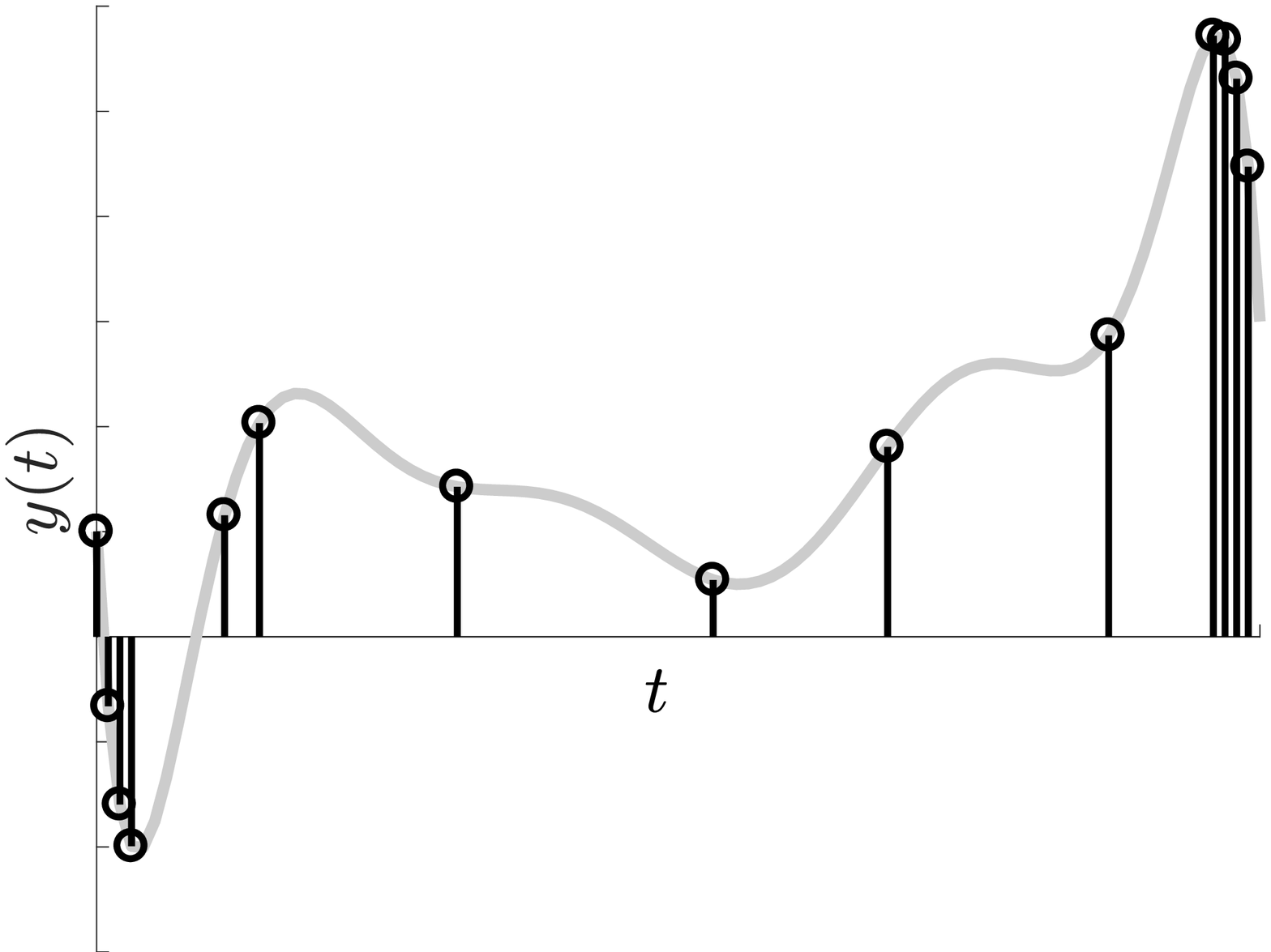}
		\caption{Example set of nodes sampled according to $p_q$.}
	\end{subfigure}
	\vspace{-.5em}
	\caption{A visualization of the universal sampling distribution, $p_q$, which can be used for reconstructing a signal under any Fourier transform prior $\mu$. To obtain $p_q$ for a given number of samples $q$, choose $\alpha$ so that $q = \Theta(\alpha \log^2 \alpha)$. Set $z_q(t)$ equal to $\alpha/\min(t,T-t)$, except near $0$ and $T$, where the function is capped at $z_q(t) = \alpha^6$. Construct $p_q$ by normalizing $z_q$ to integrate to 1.}
	\label{fig:our_dist}
\end{figure}

% The form of the function being interpolated has little impact on the method used to sample and reconstruct that function.

\subsection{Algorithmic complexity}\label{sec:ourAlgo} 
While Theorem \ref{thm:informal_sample_dist} immediately yields an approach for selecting samples $t_1,\ldots,t_q$, it is only useful if we can \emph{efficiently} solve Problem \ref{prob:unformal_interp} given the noisy measurements $y(t_1)+n(t_1),\ldots,y(t_q)+n(t_q)$. We show that this is possible
 for a broad class of constraint measures. Specifically, we need only assume that we can efficiently compute the positive-definite kernel function\footnote{When $y$ is real valued, it makes sense to consider symmetric $\mu$. In this case, $k_\mu$ is also real valued. However, in general it may be complex valued.}:
 \begin{align}
\label{eq:kernel_function}
k_\mu(t_1,t_2) = \int_{\xi\in \RR} e^{-2\pi i (t_1 - t_2) \xi} d\mu(\xi).
\end{align}
The above integral can be approximated via numerical quadrature, but for many of the aforementioned applications, it has a closed-form. For example, when $\mu$ is supported on just $k$ frequencies, it is a sum of these frequencies. When $\mu$ is uniform on $[-F, F]$, $k_\mu(t_1,t_2) = \sinc(2\pi F (t_1-t_2))$. For multiband signals with $s$ bands, $k_\mu(t_1,t_2)$ is a sum of $s$ modulated sinc functions. In fact, $k_\mu(t_1,t_2)$ has a closed-form for all $\mu$ illustrated in Figure \ref{fig:constraints}. Further details are discussed in Appendix \ref{app:kernel_computation}. Assuming a subroutine for computing $k_\mu(t_1,t_2)$, our main algorithmic result is as follows:

\begin{theorem}(Main result, algorithmic complexity)
	\label{thm:informal_main}
	There is an algorithm that solves Problem \ref{prob:unformal_interp} with probability $99/100$ which uses $O\left (\smu \cdot \log^2 (\smu)\right)$ time domain samples (sampled according to the distribution given by Theorem \ref{thm:informal_sample_dist}) and runs in $\tilde O(\smu^\omega + \smu^2 \cdot  Z)$ time, assuming the ability  to compute $k_\mu(t_1,t_2)$ for any $t_1,t_2 \in[0,T]$ in $Z$ time.\footnote{For conciseness, we use $\tilde{O}(z)$ to denote $\tilde{O}(z\log^c z)$, where $c$ is some fixed constant (usually $\le 2$). In formal theorem statements we give $c$ explicitly. $\omega < 2.373$ is the current  exponent of fast matrix multiplication \cite{williams2012multiplying}.} The algorithm returns a representation of $\tilde y(t)$ that can be evaluated in $\tilde O(\smu  \cdot Z)$ time for any $t$.
\end{theorem}
For bandlimited, Gaussian, or Cauchy-Lorentz priors $\mu$, $Z = O(1)$. For $s$ sparse signals or multiband signals with $s$ blocks, $Z = O(s)$.

We note that, while Theorem \ref{thm:informal_main} holds when  $\tilde O\left (\smu \right)$ samples  are taken, $\smu$ may be not be known and thus it may  be unclear how to set the sample size. In our full statement  of the Theorem in Section \ref{sec:puttingTogether} we make it clear that any upper bound on $\smu$ suffices to set the sample size. The sample complexity will depend on how tight this upper bound is. In Appendix \ref{app:stat_dimension} we give upper bounds on $\smu$ for a number of common $\mu$, which can be plugged into Theorem \ref{thm:informal_main}.

\subsection{Our approach}
Theorems \ref{thm:informal_sample_complexity}, \ref{thm:informal_sample_dist}, and \ref{thm:informal_main} are achieved through a simple and practical algorithmic framework. In Section \ref{sec:regresision}, we show that Problem \ref{prob:unformal_interp} can be modeled as a least squares regression problem with $\ell_2$ regularization. As long as we can compute $k_\mu(t_1, t_2)$, we can solve this problem using \emph{kernel ridge regression}, a popular function fitting technique in nonparametric statistics \cite{shawe2004kernel}.

Naively, the kernel regression problem is infinite dimensional: it needs to be solved over the \emph{continuous} time domain $[0,T]$ to solve our signal reconstruction problem. This is where sampling comes in. We need to  discretize the problem and establish that our solution over a fixed set of time samples nearly matches the solution over the continuous interval. To bound the error of discretization, we turn to a tool from randomized numerical linear algebra: \emph{statistical leverage score sampling} \cite{SpielmanSrivastava:2011,DrineasMahoney:2016}. We show how to \emph{randomly} discretize Problem \ref{prob:unformal_interp} by sampling time points with probability proportional to an appropriately defined non-uniform {leverage score distribution} on $[0,T]$. The required number of samples is $O(\smu\log \smu)$, which proves Theorem \ref{thm:informal_sample_complexity}.

Unfortunately, the leverage score distribution does not have a closed-form, varies depending on $\epsilon$, $T$,  and $\mu$, and likely cannot be sampled from exactly. To prove Theorem \ref{thm:informal_sample_dist}, we show that for any $\mu$, for large enough $q$, the closed form distribution $p_q$ \emph{upper bounds} the leverage score distribution.
This upper bound closely approximates the true leverage score distribution and, therefore, can be used in its place during sampling, losing only a $\log \smu$ factor in the sample complexity.

The leverage score distribution roughly measures, for each time point $t$, how large $|y(t)|^2$ can be compared to $\norm{y}_T^2$ when $y$'s Fourier transform is constrained by $\mu$ (i.e., when $\norm{x}_\mu^2$ as defined in Problem \ref{prob:unformal_interp} is bounded). To upper bound this measure we turn to another powerful result from the randomized numerical linear algebra literature: every  matrix contains a small subset of columns that span a near-optimal low-rank approximation to that matrix \cite{sarlos2006improved,BoutsidisMahoneyDrineas:2009a,DeshpandeRademacher:2010}. In other words, every  matrix admits a near-optimal low-rank approximation with \emph{sparse column support}.
By  extending this result to continuous linear operators, we prove that  the smoothness of a signal whose Fourier transform has $\norm{x}_\mu^2$ bounded can be bounded by the smoothness of an $O(\smu)$ sparse Fourier function. This lets us apply recent results of  \cite{ChenKanePrice:2016,ChenPrice:2018} that bound $|y(t)|^2$ in terms of $\norm{y}_T^2$ for any sparse Fourier function $y$ . Intuitively, our result shows that  the simplicity of sparse Fourier functions governs the simplicity of \emph{any class} of Fourier constrained functions.

The above argument yields Theorem \ref{thm:informal_sample_dist}. Since we can sample from $p_q$ in $O(1)$ time, we can efficiently sample the time domain to $O(\smu \cdot \log^2 \smu)$ points and then solve Problem \ref{prob:unformal_interp} by applying kernel ridge regression to these points, which takes $\tilde O(\smu^\omega + \smu^2 \cdot Z)$ time, assuming the ability  to compute $k_\mu(\cdot,\cdot)$ in $Z$ time. This yields the algorithmic result of Theorem \ref{thm:informal_main}.

\subsection{Roadmap}
\label{sec:prelims}
The rest of this paper is devoted to proving Theorems \ref{thm:informal_sample_complexity}, \ref{thm:informal_sample_dist}, and \ref{thm:informal_main}, and is structured as follows:
\begin{description}
	\item[Section \ref{sec:notation}] We lay out basic notation that is used throughout the paper.
	\item[Section \ref{sec:regresision}]  We reduce Problem \ref{prob:unformal_interp} to a kernel ridge regression problem and explain how to randomly discretize and solve this problem via leverage score sampling, proving Theorem \ref{thm:informal_sample_complexity}.
	\item[Section \ref{sec:general}] We give an upper bound on the leverage score distribution for general priors, proving Theorems \ref{thm:informal_sample_dist} and  \ref{thm:informal_main}.
	\item [Section \ref{sec:lb}] We prove that, under a mild assumption, the statistical dimension tightly characterizes the sample complexity  of solving Problem \ref{prob:unformal_interp}, and thus that our results are nearly optimal.
	\item[Section \ref{sec:conclusion}] We conclude with a discussion of open questions.
\end{description}

We defer an in depth overview of related work to Appendix \ref{app:prior_work}. In Appendix \ref{app:op} we give operator theory preliminaries.  In Appendix \ref{app:leverage_scores} we prove our extensions of a number of randomized linear algebra primitives to continuous operators.
In Appendix \ref{sec:bandlimited}, we bound the statistical dimension for the important case of bandlimited functions. We use this result in Appendix \ref{app:stat_dimension} to prove statistical dimension bounds for multiband, Gaussian, and Cauchy-Lorentz priors (shown in Table \ref{tab:stat_dim}). In Appendix \ref{app:kernel_computation}, we show how to compute the kernel function $k_\mu$ for these common priors. In Appendix \ref{app:bayes}, we discuss a Bayesian approach to signal reconstruction under a Fourier transform prior.

\section{Notation}\label{sec:notation}
Let $\mu$ be a probability measure on $(\RR,\mathcal{B})$, where $\mathcal{B}$ is the Borel $\sigma$-algebra on $\RR$. Let $L_2(\mu)$ denote the space of complex-valued square integrable functions with respect to $\mu$. For $a,b \in L_2(\mu)$, let $\langle a,b\rangle_\mu$ denote $\int_{\xi \in \RR} a(\xi)^* b(\xi) \, d\mu(\xi)$ where for any  $x \in \CC$, $x^*$ is its complex conjugate. Let $\|a\|_{\mu}^2$ denote $\langle a,a\rangle_\mu$. Let $\Imu$ denote the identity operator on $L_2(\mu)$. Note that for any $\mu$, $L_2(\mu)$ is a separable Hilbert space and thus has a countably infinite orthonormal basis  \cite{HunterNachtergaele:2001}. 

We overload notation and use $L_2(T)$ to denote the space of complex-valued square integrable functions with respect to the uniform probability measure on $[0,T]$. It will be clear from context that $T$ is not a measure. For $a,b \in L_2(T)$, let $\langle a,b\rangle_{T}$ denote $\frac{1}{T}\int_{0}^T a(t)^* b(t) \,dt$ and let $\|a\|_{T}^2$ denote $\langle a,a\rangle_{T}$. Let $\mathcal{I}_T$ denote the identity operator on $L_2(T)$.

Define the Fourier transform operator $\Fmu: L_2(T) \rightarrow L_2(\mu)$ as:
\begin{align}
\left[\Fmu\, f\right](\xi) = \frac{1}{T}\int_0^T f(t)e^{-2\pi i t \xi} dt.
\end{align}

The adjoint of $\Fmu$ is the unique operator $\Fmu^*: L_2(\mu) \rightarrow L_2(T)$ such that for all $f \in L_2(T), g\in L_2(\mu)$ we have
$\langle g,\Fmu\, f\rangle_\mu = \langle \Fmu^*\, g,f\rangle_T$. It is not hard to see that $\Fmu^*$ is the inverse Fourier transform operator with respect to $\mu$ as defined in
Section \ref{sec:prob_statement}, equation \eqref{eq:mu_transform}:
\begin{align}\label{eq:mu_transform2}
	 \left[\Fmu^* \,g\right](t) \eqdef \int_{\RR} g(\xi) e^{2\pi i \xi t}\, d\mu(\xi).
\end{align}
Note that the kernel operator $\Kmu: L_2(T) \rightarrow L_2(T)$ originally defined in \eqref{eq:kernel_op} is equal to
\begin{align*}
\Kmu = \Fmu^*\Fmu.
\end{align*}
 $\Kmu$ is self-adjoint, positive semidefinite and trace-class
%\footnote{
%Since $\Kmu$ is self-adjoint, positive semidefinite, it is trace-class if and only iff $\tr(\Kmu)$ is bounded. Let $\vartheta_\xi(t) = e^{-2\pi i \xi t}$ and
 %$\theta_n(t) = e^{-2\pi i n t/T}$, $\{\theta_n\}_{n=-\infty}^{\infty}$ is
 %$\{e_i\}_{i=1}^\infty$ be an orthonormal basis for $L_2(T)$. We have: $\tr(\Kmu) = \sum_{i=1}^\infty \langle e_i, \Kmu e_i \rangle_T = \int_{\xi \in \RR} \sum_{i=1}^\infty \langle e_i, [\vartheta_\xi \otimes \vartheta_\xi] e_i \rangle_T\, d\mu(\xi) =  \int_{\xi \in \RR} \sum_{i=1}^\infty |\langle e_i, \vartheta_\xi \rangle_T|^2\, d\mu(\xi) = \int_{\xi \in \RR} \norm{\vartheta_\xi }_T^2\, d\mu(\xi) = 1$.} 
and an integral operator with kernel $k_\mu$:
 \begin{align*}
     [\Kmu z](t) =\frac{1}{T} \int^T_0 k_\mu(s,t) z(s) ds,
\end{align*}
where $k_\mu$ is as defined in \eqref{eq:kernel_function}. The trace of $\Kmu$ is equal to $1$.\footnote{Since the kernel is a Fourier transform of a probability measure, it is Hermitian positive definite (Bochner's Theorem). Then we can conclude that $\Kmu$ is trace-class from Mercer's theorem, and calculate $\tr(\Kmu) = \frac{1}{T}\int^T_0 k_\mu(t,t)dt=1$.}
We will also make use of the Gram operator: $\cal{G}_\mu \eqdef \cal{F}_\mu \cal{F}^*_\mu$. ${\cal G}_\mu$ is also self-adjoint, positive semidefinite, and trace-class.

\spara{Remark:}
It may be useful  for the reader to informally regard $\Fmu$ as an infinite matrix with rows  indexed by $\xi \in \RR$ and columns indexed by $t \in [0,T]$. Following the definition of $\Fmu$ above, and assuming that $\mu$ has a density $p$, this infinite matrix has entries given by:
\begin{align}\label{eq:informal}
\Fmu(\xi,t) = \sqrt{\frac{p(\xi)}{T}} \cdot e^{-2\pi i t \xi}.
\end{align}
The results we apply on leverage score sampling can all be seen as extending results for finite matrices from the randomized numerical linear algebra literature  to this infinite matrix.

\section{Function fitting with least squares regression}
\label{sec:regresision}

Least squares regression provides a natural approach to solving the interpolation task of Problem \ref{prob:unformal_interp}. In particular, consider the following regularized minimization problem over functions $g \in L_2(\mu)$\footnote{The fact that the minimum is attainable is a simple consequence of the extreme value theorem, since the search space can be restricted to
	$\|g\|_\mu^2 \leq \|(y+n)\|^2_T / \epsilon$.}:
\begin{align}
\label{eq:least_squares_setup}
	\min_{g\in L_2(\mu)} \|\Fmu^* g - (y+n)\|_T^2 + \epsilon\|g\|_\mu^2.
\end{align}
The first term encourages us to find a function $g$ whose inverse Fourier transform is close to our measured signal $y+n$. The second term encourages us to find a low energy solution -- ultimately, we solve \eqref{eq:least_squares_setup} based on only a small number of samples $y(t_1), \ldots, y(t_k)$, and smoother, lower energy solutions will better generalize to the entire interval $[0,T]$. We remark that it is well known that least squares approximations benefit from regularization even in the noiseless case~\cite{CDL13}.

We first state a straightforward fact: if we minimize \eqref{eq:least_squares_setup}, even to a coarse approximation, then we are able to solve Problem \ref{prob:unformal_interp}.
\begin{claim}
	\label{claim:regression_reduction}
	Let $y = \Fmu^* x$, $n \in L_2(T)$ be an arbitrary noise function, and for any $C\geq 1$, let $\tilde{g}\in L_2(\mu)$ be a function satisfying:
	\begin{align*}
%	\label{eq:approx_regress_1}
		\|\Fmu^* \tilde{g} - (y+n)\|_T^2 + \epsilon\|\tilde{g}\|_\mu^2 \leq C\cdot\min_{g\in L_2(\mu)}\left[  \|\Fmu^* g - (y+n)\|_T^2 + \epsilon \|g\|_\mu^2\right].
	\end{align*}
	Then
	\begin{align*}
		\|\Fmu^* \tilde{g} - y\|_T^2 \leq 2C\epsilon \|x\|_\mu^2 + 2(C+1)\|n\|_T^2.
	\end{align*}
\end{claim}
\begin{proof}
	Since $y = \Fmu^* x$, $\min_{g\in L_2(\mu)}\left[  \|\Fmu^* g - (y+n)\|_T^2 + \epsilon\|g\|_\mu^2\right] \leq \|n\|_T^2 + \epsilon\|{x}\|_\mu^2$. Thus, $\|\Fmu^* \tilde{g} - (y+n)\|_T^2 \leq C\epsilon \|x\|_\mu^2 + C\|n\|_T^2$. The claim then follows via triangle inequality:
	\begin{align*}
	 \|\Fmu^* \tilde g - y \|_T - \|n\|_T &\le \|\Fmu^* \tilde{g} - (y+n)\|_T\\
	 \|\Fmu^* \tilde g - y \|_T &\le \sqrt{C\epsilon \|x\|_\mu^2 + C\|n\|_T^2} + \|n\|_T\\
	 \|\Fmu^* \tilde g - y \|_T^2 &\le 2C\epsilon \|x\|_\mu^2 + 2(C+1)\|n\|_T^2.
	\end{align*}
\end{proof}

Claim \ref{claim:regression_reduction} shows that approximately solving the regression problem in \eqref{eq:least_squares_setup}, with regularization parameter $\epsilon$  gives a solution to Problem \ref{prob:unformal_interp} with parameter $2C\epsilon$ (decreasing the regularization parameter to $\frac{\epsilon}{2C}$ will let us solve with parameter $\epsilon$).
But how can we solve the regression problem efficiently? Not only does the problem involve a possibly infinite dimensional parameter vector $g$, but the objective function also involves the continuous time interval $[0,T]$.

\subsection{Random discretization via leverage function sampling}
\label{sec:random_discrete}
The first step is to deal with the latter challenge, i.e., that of a continuous time domain. We show that it is possible to \emph{randomly discretize} the time domain of \eqref{eq:least_squares_setup}, thereby reducing our problem to a regression problem on a finite set of times $t_1, \ldots, t_q$. In particular, we can sample time points with probability proportional to the so-called \emph{ridge leverage function}, a specific non-uniform distribution that has been applied widely in randomized algorithms for regression and other linear algebra problems on discrete matrices \cite{AlaouiMahoney:2015,CalandrielloLazaricValko:2016,CohenMuscoMusco:2017,MuscoMusco:2017,MuscoWoodruff:2017}.

While we cannot compute the leverage function explicitly for our problem, an issue highlighted in \cite{Bach:2017}, our main result (Theorem \ref{thm:informal_sample_dist}) uses a simple, but very accurate, closed form approximation in its place.
We start with the definition of the ridge leverage function:

\begin{defn}[Ridge leverage function]\label{def:ridgeScores}
	For a probability measure $\mu$ on $\RR$, time length $T > 0$, and $\epsilon \ge 0$, we define the $\epsilon$-ridge leverage function for $t \in [0,T]$ as\footnote{Formally $L_2(T)$ is a space of equivalence classes of functions that differ at a set of points with measure $0$. For notational simplicity, here and throughout we use $\Fmu^* \alpha$ to denote the specific representative of the equivalence class $\Fmu^* \alpha \in L_2(T)$ given by \eqref{eq:mu_transform2}. In this way, we can consider the pointwise value $[\Fmu^* \alpha](t)$, which we could alternatively express as $\langle \varphi_t,\alpha\rangle_\mu$, for $\varphi_t(\xi) \eqdef e^{-2\pi i t \xi}$.}:
	\begin{align}
	\label{eq:leverage_def}
	%\langle z_t , \alpha \rangle_\mu^2
		\tmu(t) = \frac{1}{T} \cdot \max_{\{\alpha \in L_2(\mu):\, \norm{\alpha}_\mu > 0\}} \frac{\left|[\Fmu^* \alpha](t) \right|^2 }{\norm{\Fmu^* \alpha}_T^2 + \epsilon \norm{\alpha}_\mu^2}.
	\end{align}
\end{defn}

Intuitively, the ridge leverage function at time $t$ is an upper bound of how much a function can  ``blow up'' at $t$ when its Fourier transform is constrained by $\mu$. The denominator term $\norm{\Fmu^* \alpha}_T^2$ is the average squared magnitude of the function $F_\mu^*\alpha$, while the numerator term, $|[\Fmu^* \alpha](t)|^2$, is the squared magnitude at $t$.
The regularization term $\epsilon \norm{\alpha}_\mu^2$ reflects the fact that, to solve \eqref{eq:least_squares_setup}, we only need to bound the smoothness for functions with bounded Fourier energy under $\mu$. As observed in \cite{pauwels2018relating}, the ridge leverage function can be viewed as a type of \emph{Christoffel function}, studied in the literature on orthogonal polynomials  and approximation theory \cite{pauwels2018relating,nevai1986geza,totik2000asymptotics,borwein2012polynomials}.

The larger the leverage ``score'' $\tmu(t)$, the higher the probability we will sample time $t$, to ensure that our sample points well reflect any possibly significant components or `spikes' of the function $y$. Ultimately, the integral of the ridge leverage function $\int_0^T \tmu(t) dt$ determines how many samples we require to solve \eqref{eq:least_squares_setup} to a given accuracy. Theorem \ref{thm:leverageProps} below states the already known fact that  the  ridge leverage function integrates to the statistical dimension \cite{AvronKapralovMusco:2017}, which will ultimately allow us to achieve the $\tilde O(\smu)$ sample complexity bound of Theorems \ref{thm:informal_sample_complexity} and \ref{thm:informal_sample_dist}. Theorem \ref{thm:leverageProps} also gives two alternative characterizations of the leverage function that will prove useful. The theorem is proven in Appendix \ref{app:leverage_scores}, using techniques for finite matrices, adapted to the operator setting.

\begin{theorem}[Leverage function properties]\label{thm:leverageProps}
Let $\tmu(t)$ be the ridge leverage function (Definition \ref{def:ridgeScores}) and define $\varphi_t \in L_2(\mu)$ by $\varphi_t(\xi) \eqdef e^{-2\pi i t \xi}$. We have:
\begin{itemize}
\item The ridge  leverage function integrates to the statistical dimension:
\begin{align}\label{eq:StatDim}
	\int_0^T \tmu(t) dt= \smu \eqdef  \tr(\Kmu (\Kmu + \epsilon \mathcal{I}_T)^{-1}).
\end{align}
\item Inner Product characterization:
\begin{align}\label{eq:InnerProd}
\tmu(t) =\frac{1}{T} \cdot \langle \varphi_t,  (\Gmu + \epsilon \Imu)^{-1}  \varphi_t \rangle_\mu.
\end{align}
\item Minimization Characterization:
\begin{align}\label{eq:Min}
\tmu(t) = \frac{1}{T} \cdot \min_{\beta \in L_2(T)} \frac{\norm{\Fmu \beta - \varphi_t}_\mu^2}{\epsilon} + \norm{\beta}_{T}^2.
\end{align}
\end{itemize}
\end{theorem}
\ifdraft
\Haim{In the theorems we require $\tilde{\tau}$ to be measurable. One can ask whether $\tmu$ is measurable. Property \eqref{eq:InnerProd} makes sure that $\tmu$ is continuous, so it is measurable with respect to the Lebesgue $\sigma$-algebra~\cite[Example 12.22]{HunterNachtergaele:2001}. But we say we work with the Borel $\sigma$-algebra...}
\fi

In Theorem \ref{thm:baseSampling}, 
we give our formal statement that the ridge leverage function can be used to randomly sample time domain points to discretize the regression problem in \eqref{eq:least_squares_setup} and solve it approximately. While complex in appearance, readers familiar with randomized linear algebra will recognize  Theorem \ref{thm:baseSampling} as closely analogous to standard approximate regression results for leverage score sampling from finite matrices \cite{ClarksonWoodruff:2013}. As discussed, since we are typically unable to sample according to the true ridge leverage function, we give a general result, showing that sampling with any upper bound function with a finite integral suffices.

\begin{theorem}[Approximate regression via leverage function sampling]\label{thm:baseSampling}
	Assume that $\epsilon \leq \opnorm{\Kmu}$.\footnote{If $\epsilon > \opnorm{\Kmu}$ then \eqref{eq:least_squares_setup} is  solved to a constant approximation factor by the trivial solution $g = 0$.} Consider a measurable function $\ttmu(t)$ with $\ttmu(t)  \ge \tmu(t)$ for all $t$ and let $\tsmu  = \int_0^T \ttmu(t) dt$. 
	Let $s = c \cdot \tsmu \cdot \left(\log\tsmu + 1/\delta\right)$  for sufficiently large fixed constant $c$ and let $t_1,\ldots,t_s$ be time points selected by drawing each randomly from $[0,T]$  with probability proportional to $\ttmu(t)$. For $j \in 1,\ldots,s$, let $w_j =  \sqrt{\frac{1}{sT} \cdot \frac{\tsmu}{\ttmu(t_j)}} $. Let $\bv{F}: \CC^s \rightarrow L_2(\mu)$ be the operator defined by:
	\begin{align*}
		\left[\bv{F} \,g\right](\xi) = \sum_{j=1}^s w_j \cdot g(j) \cdot e^{- 2\pi i \xi t_j}
	\end{align*}
	and $\bv{y},\bv{n} \in \RR^s$ be the vectors with $\bv{y}(j) = w_j \cdot y(t_j)$ and $\bv{n}(j) = w_j \cdot n(t_j)$.
	Let:
	\begin{align}
		\tilde{g} = \argmin_{g \in L_2(\mu)}\left[  \|\bv{F}^* g - (\bv y+\bv n)\|_2^2 + \epsilon\|g\|_\mu^2\right]\label{eq:approxkrr}
	\end{align}
	With probability $\ge 1- \delta$:
	\begin{align}
	\label{12again2}
		\|\Fmu^* \tilde{g} - (y+n)\|_T^2 + \epsilon\|\tilde{g}\|_\mu^2 \leq 3\min_{g \in L_2(\mu)}\left[  \|\Fmu^* g - (y+n)\|_T^2 + \epsilon\|g\|_\mu^2\right].
	\end{align}
\end{theorem}
A generalized version of this result is proven in Appendix \ref{app:leverage_scores}, which holds even when $\tilde{g}$ is only an approximate minimizer of \eqref{eq:approxkrr}.

Theorem \ref{thm:baseSampling} shows that
 $\tilde{g}$ obtained from solving the discretized regression problem provides an approximate solution to \eqref{eq:least_squares_setup} and by Claim \ref{claim:regression_reduction}, $\tilde{y} = \Fmu^* \tilde{g}$ solves Problem \ref{prob:unformal_interp} with parameter $\Theta(\epsilon)$. 
If we have $\ttmu(t) = \tmu(t)$, Theorem \ref{thm:baseSampling} combined with Claim \ref{claim:regression_reduction} shows that Problem \ref{prob:unformal_interp} with parameter $\Theta(\epsilon)$ can be solved with sample complexity $O\left ( \smu \cdot \log \smu \right )$, since by \eqref{eq:StatDim}, $\int_0^T \tmu(t) dt= \smu$. Note that, by simply decreasing the regularization parameter in \eqref{eq:least_squares_setup} by a constant factor, we can solve Problem \ref{prob:unformal_interp} with parameter $\epsilon.$ The asymptotic complexity is identical since, by \eqref{eq:Min}, for any $c \le 1$ and any $t \in [0,T]$, $\tau_{\mu,c\epsilon}(t) \le \frac{1}{c} \tmu(t)$ and so: 
\begin{align}\label{constAdjust}
s_{\mu,c\epsilon} \le \frac{1}{c} \smu.
\end{align}

This proves the sample complexity result of Theorem \ref{thm:informal_sample_complexity}.
 However, since it is not clear that sampling according to $\tmu(t)$ can be done efficiently (or at all), it does not yet give an algorithm yielding this complexity.\footnote{We conjecture that the existential sample complexity  can in fact be upper bounded by $O(\smu)$ by adapting deterministic sampling methods for finite matrices to the operator setting \cite{CohenNelsonWoodruff16}, like we do in Lemma \ref{lem:cssSpectral}.} This issue will be addressed in Section \ref{sec:general}, where we prove Theorem \ref{thm:informal_sample_dist}.

We prove Theorem \ref{thm:baseSampling} in Appendix \ref{app:leverage_scores}.
We show that leverage function sampling satisfies, with good probability, an affine embedding guarantee: that $\|\bv{F}^* {g} - (\bv{y}+ \bv n)\|_2^2 + \epsilon\|{g}\|_\mu^2$ closely approximates $\|\Fmu^* {g} - (y+n)\|_T^2 + \epsilon\|{g}\|_\mu^2$ \emph{for all} $g \in L_2(\mu)$. Thus, a (near) optimal solution to the discretized problem, $\min_{g \in L_2(\mu)}\left[  \|\bv{F}^* g - (\bv y+\bv n)\|_2^2 + \epsilon\|g\|_\mu^2\right]$, gives a near optimal solution to the original problem, $\min_{g \in L_2(\mu)}\left[  \|\Fmu^* g - (y+n)\|_T^2 + \epsilon\|g\|_\mu^2\right]$. Our proof of the affine embedding property is analogous to existing proofs for finite dimensional matrices \cite{ClarksonWoodruff:2013,AvronClarksonWoodruff:2017}.

\subsection{Efficient solution of the discretized problem}

Given an upper bound on the ridge leverage function $\ttmu(t) \ge  \tmu(t)$, we can apply  Theorem \ref{thm:baseSampling} to approximately solve the ridge regression problem of
\eqref{eq:least_squares_setup} and  therefore Problem \ref{prob:unformal_interp} by Claim \ref{claim:regression_reduction}. In Section
\ref{sec:general} we show how to obtain such an upper bound for any {$\mu$} %that achieves $\tsmu = O(\smu \log(\smu))$,
using a universal distribution.
%In Appendix \ref{sec:bandlimited} and \ref{app:stat_dimension} we in turn bound $\smu$ for a number of common priors $\mu$.
%we give a nearly tight closed form upper bound on $\tmu(t)$ for bandlimited signals achieving sample complexity  $\tsmu = O(\smu)$, which improves on our general bound by a logarithmic factor. \todo{Discuss how it also lets us actually bound $\smu$ for bandlimited and other constraints.}

First, however, we demonstrate how to apply Theorem \ref{thm:baseSampling} algorithmically. Specifically, we show how to solve the randomly discretized problem of \eqref{eq:approxkrr} efficiently.
%$\poly(\tsmu)$ evaluations of $k_\mu$ and $\poly(\tsmu)$ runtime, making just $O(\tsmu)$ noisy queries to $y(t)$.\footnote{For simplicity, we hide dependencies on the failure probability $\delta$ in this high level description.} 
Combined with Theorem \ref{thm:baseSampling} and our bound on $\tmu(t)$ given in Section \ref{sec:general}, this yields a randomized algorithm (Algorithm \ref{alg:main}) for Problem \ref{prob:unformal_interp}. The formal analysis of Algorithm \ref{alg:main} is given in Theorem \ref{thm:mainAlg}.

 \begin{algorithm}[H]
\caption{\algoname{Time Point Sampling and Signal Reconstruction}}
{\bf input}: Probability measure $\mu(\xi)$, $\epsilon,\delta > 0$, time bound $T$, and function $y: [0,T] \rightarrow \RR$. Ridge leverage function upper bound $\ttmu(t) \ge \tmu(t) $ with $\tsmu=\int_0^T \ttmu(t) dt$. \\
{\bf output}: $t_1,\ldots, t_{s} \in [0,T]$ and $\bv{z} \in \CC^{s}$.
\begin{algorithmic}[1]
\State{Let $s = c \cdot \tsmu \cdot \left(\log \tsmu + \frac{1}{\delta}\right)$ for a sufficiently  large constant $c$.\label{step:sets}}
\State{Independently sample $ t_{1},\ldots, t_{s} \in [0,T]$ with probability  proportional to $\ttmu(t)$ and set the weight $w_i :=  \sqrt{\frac{1}{sT} \cdot \frac{ \tsmu}{\ttmu(t_i)}}$.\label{step:initialSample}}
%\tilde t_{2,1},\ldots,\tilde t_{2,\tilde s}
\State{Let $\bv{ K} \in \CC^{ s \times  s}$ be the matrix with $\bv{ K}(i,j) = w_{i}  w_{j} \cdot k_\mu( t_{i}, t_{j})$.\label{step3}}
\State{Let $\bv{\bar y} \in \CC^s$ be the vector with $\bv{\bar y}(i) = w_i \cdot [y(t_i) + n(t_i)]$.\label{formY}}
\State{Compute $\bv{\bar z} := (\bv{K} + \epsilon  \bv{I})^{-1}\bv{\bar y}$.\label{step8}}
\\\Return{${t}_1,\ldots,{t}_{s}\in [0,T]$ and $\bv{z} \in \CC^{s}$ with $\bv{z}(i) = \bv{\bar z}(i) \cdot {w}_i$.\label{finalStep}}
\end{algorithmic}
\label{alg:main}
\end{algorithm}

 \begin{algorithm}[H]
\caption{\algoname{Evaluation of Reconstructed Signal}}
{\bf input}: Probability measure $\mu(\xi)$, $t_1,\ldots,t_s \in [0,T]$, $\bv{z} \in \CC^s$, and evaluation point $t \in [0,T]$. \\
{\bf output}: Reconstructed function value $\tilde y(t)$.
\begin{algorithmic}[1]
\State{For $i \in \{1,\ldots,s\}$, compute $k_\mu(t_i,t) = \int_{\xi \in \RR} e^{-2\pi i(t_i-t)} d\mu(\xi)$.}
\\\Return{$\tilde y(t) = \sum_{i=1}^s \bv{z}(i) \cdot k_\mu(t_i,t)$.}
\end{algorithmic}
\label{alg:main2}
\end{algorithm}

 \begin{theorem}[Efficient signal reconstruction given leverage function upper bounds]\label{thm:mainAlg} 	Assume that $\epsilon \leq \opnorm{\Kmu}$.\footnote{As discussed for Theorem \ref{thm:baseSampling}, if $\epsilon > \opnorm{\Kmu}$, Problem \ref{prob:unformal_interp} is trivially solved by $\tilde y  = 0$.}
 Algorithm \ref{alg:main} returns $t_1,\ldots,t_{s} \in [0,T]$ and $\bv{z} \in \CC^{s}$ such that $\tilde y(t) = \sum_{i=1}^{s} \bv{z}(i) \cdot k_\mu(t_i,t)$ (as computed in Algorithm \ref{alg:main2}) satisfies with probability $\ge 1-\delta$:
\begin{align*}
\norm{\tilde y - y}_T^2  \le 6 \epsilon \norm{x}_\mu^2 + 8\norm{n}_T^2.
\end{align*}
Suppose we can sample $t \in [0,T]$ with probability  proportional to $\ttmu(t)$ in time $W$ and compute the kernel function $k_\mu(t_1,t_2) = \int_{\xi \in \RR}  e^{-2\pi i (t_1 - t_2)} d\mu(\xi)$ in time $Z$.
Algorithm \ref{alg:main} queries $y+n$ at $s$ points and runs in $ O \left  (s \cdot W + s^2 \cdot Z + s^{\omega} \right )$ time\footnote{Here $\omega < 2.373$ is the exponent of fast matrix multiplication. $s^\omega$ is the theoretically fastest runtime required to invert a dense $s \times s$ matrix. We note that the $s^\omega$ term may be thought of as $s^3$ in practice, and potentially could be accelerated using a variety  of techniques for fast (regularized) linear system solvers.} where $s = O\left (\tsmu \cdot \left(\log \tsmu + 1/\delta\right) \right  )$. Algorithm \ref{alg:main2} evaluates $\tilde y(t)$ in $O(s \cdot Z)$ time for any  $t$.
\end{theorem}
\begin{proof}
In Step \ref{step:initialSample} of Algorithm \ref{alg:main},
$t_{1},\ldots, t_{s}$ are sampled according to $\ttmu(t)$, which upper bounds $\tmu(t)$. We can thus apply Theorem \ref{thm:baseSampling}. If the constant $c$ in Step \ref{step:sets} is set large enough, with probability  $\ge 1-\delta$, letting $\bv{F},\bv{y},$ and $\bv{n}$ be as defined in that theorem, \eqref{12again2} holds for
\begin{align*}
		 \tilde g =\argmin_{g \in L_2(\mu)}\left[  \|\bv{F}^* g - (\bv y+\bv{n})\|_2^2 + \epsilon\|g\|_\mu^2\right].
	\end{align*}
%we have:
%	\begin{align*}
%		\|\Fmu^*  \tilde g - (y+n)\|_T^2 + \epsilon\| \tilde g\|_\mu^2 \leq 3\cdot\min_{g\in L_2(\mu)}\left[  \|\Fmu^* g - (y+n)\|_T^2 + \epsilon\|g\|_\mu^2\right].
%	\end{align*}
	Therefore, 	letting $\tilde y \eqdef \Fmu^* \tilde {g}$ and applying Claim \ref{claim:regression_reduction}, with probability $\ge 1-\delta$,
	\begin{align}\label{eq:mainAc}
	\norm{\tilde y - y}_T^2  \le 6\epsilon \norm{x}_\mu^2 + 8 \norm{n}_T^2.
	\end{align}
	Further, the minimizer $ \tilde g$ is indeed unique and can be written as (see Lemma \ref{lem:ridgeMinimizer} in Appendix \ref{app:leverage_scores}):
	$$
	\tilde  g = \bv{F} (\bv{K}  + \epsilon \bv{I})^{-1} (\bv y+\bv{n}) = \bv{F} (\bv{K}  + \epsilon \bv{I})^{-1} \bv{\bar y}
	$$
where $\bv{K} = \bv{F}^* \bv{F}$ is as defined in Step \ref{step3} of Algorithm \ref{alg:main} and $\bv{\bar y} = \bv{y} + \bv{n}$ is formed in Step \ref{formY}.
If we let $\bv{\bar z} = (\bv{K}  + \epsilon \bv{I})^{-1} \bv{\bar y}$ and let $\bv{z}$ have $\bv{z}(i) = \bv{\bar z}(i) \cdot w_i$ as in Steps \ref{step8} and \ref{finalStep}, we can see that:
\begin{align*}
\tilde y = \Fmu^*  \tilde g &= \sum_{i=1}^s  \bv{\bar z}(i) \cdot w_i \cdot k_\mu(t_i,t)\\
& = \sum_{i=1}^s  \bv{ z}(i) \cdot k_\mu(t_i,t),
\end{align*}
giving the expression returned in Algorithm \ref{alg:main2}. Combined with \eqref{eq:mainAc}, this completes the accuracy bound of the theorem.
The runtime and sample complexity bounds follow from observing that:
\begin{itemize}
\item $s \cdot W$ time is required to sample $t_1,\ldots,t_s$ in Step \ref{step:initialSample}.
\item $s^2 \cdot Z$ time is required to form $\bv{K}$ in Step \ref{step3}.
\item $s$ queries to $y+n$ are required to form $\bv{\bar y}$ in Step \ref{formY}.
\item $O(s^\omega)$ time is required to compute $\bv{\bar z} := (\bv{K} + \epsilon  \bv{I})^{-1}\bv{\bar y}$ in Step \ref{step8}. This runtime could potentially be improved with a variety of fast system solvers. We take  $s^\omega$ as a simple upper bound.
\item $O(s \cdot Z)$ time is required to compute $k(t_1,t),\ldots,k(t_s,t)$ to evaluate $\tilde y(t)$ in Algorithm \ref{alg:main2}.
\end{itemize}
This completes the proof of Theorem \ref{thm:mainAlg}.
\end{proof}

\spara{Remark: }
As discussed, in Section \ref{sec:general} we will give a ridge leverage function upper bound that can be sampled from in $W = O(1)$ time and closely  bounds the true leverage function for any $\mu$, giving $\tsmu= O( \smu \cdot \log \smu)$. Using this upper bound to sample time domain points, our sample complexity $s$ is thus within a $O(  \log \smu )$ factor of the best possible using Theorem \ref{thm:baseSampling}, which we would achieve if sampling using the true ridge leverage function.

In Appendix \ref{sec:bandlimited} we prove a tighter leverage function bound than the one in Section \ref{sec:general} for bandlimited signals, removing the logarithmic factor in this case.
It is not hard to see that for general $\mu$ we can also achieve optimal sample complexity by further subsampling $t_1,\ldots,t_s$ using the ridge leverage scores of $\bv{K}^{1/2}$. These scores can be computed in $\tilde O(s \cdot \smu^2)$ time using known techniques for finite kernel matrices \cite{MuscoMusco:2017}. Subsampling $O\left (\frac{\smu \log \smu}{\delta^2}\right )$ time domain points according to these scores lets us approximately solve the discretized problem of \eqref{eq:approxkrr} to error $(1+\delta)$. 

Applying the more general version of Theorem \ref{thm:baseSampling} stated in Appendix \ref{app:leverage_scores}, this yields an approximate solution to \eqref{eq:least_squares_setup} and thus to Problem \ref{prob:unformal_interp}. For constant $\delta$, we need just  $ O(\smu \cdot \log \smu)$ time samples to to solve the subsampled regression problem, matching the best possible sample complexity  of Theorem \ref{thm:baseSampling}. By the lower bound given in Section \ref{sec:lb}, Theorem \ref{thm:mainLB}, this complexity  is within a $O(\log \smu)$ factor of optimal in nearly all settings. We conjecture that one can in fact achieve within an $O(1)$ factor of the optimal sample complexity by applying deterministic selection methods to $\bv{F}$ \cite{CohenNelsonWoodruff16}, similar to the techniques used to prove Lemma \ref{lem:cssSpectral}.

\section{A near-optimal spectrum blind sampling distribution}\label{sec:general}
In the previous section, we showed how to solve Problem \ref{prob:unformal_interp} given the ability to sample time points according to the ridge leverage function $\tmu$. In general, this function depends strongly on $T$, $\mu$, and $\epsilon$, and it is not clear if it can be computed or sampled from directly.

Nevertheless, in this section we show that it is possible to efficiently obtain samples from a function that \emph{very closely} approximates the true leverage function for \emph{any} constraint measure $\mu$.
In particular we describe a set of closed form functions $\tilde \tau_{\alpha}(t)$, each parameterized by $\alpha > 0$. $\tilde \tau_{\alpha}$ upper bounds the leverage function $\tmu$ for {any} $\mu$ and $\epsilon$, as long as the statistical dimension $\smu \leq O(\alpha)$. Our upper bound satisfies
$$\int_0^T \tilde \tau_{\alpha}(t) dt = O(\smu \cdot \log \smu),$$  which means it can be used in place of the true ridge leverage function to give near optimal sample complexity via Theorem \ref{thm:baseSampling} and \ref{thm:mainAlg}. This result is proven formally in Theorem \ref{thm:fullBound}, which as a consequence immediately yields our main technical result, Theorem \ref{thm:informal_sample_dist}. The majority of this section is devoted towards building tools necessary for proving Theorem \ref{thm:fullBound}. 

\subsection{Uniform leverage bound via Fourier sparsification}\label{sec:general1}
We seek a simple closed form function that upper bounds the leverage function $\tmu$. Ultimately, we want this upper bound to be very tight, but a natural first question is whether it should exists at all. Is it possible to prove any finite upper bound on $\tmu$ without using specific knowledge of $\mu$?

We answer this first question by showing that $\tmu$ can be upper bounded by a constant function. Specifically, we show that for $t\in [0,T]$, $\tmu(t) \leq C$ for $C = \poly(\smu)$. This upper bound depends on the statistical dimension, but importantly, it does not depend on $\mu$. 
Formally we show:

\begin{theorem}[Uniform leverage function bound]\label{thm:uniformBound}
	For all $t \in [0,T]$ and $\epsilon \le 1$\footnote{If $\epsilon > 1 = \tr(\Kmu)$, Problem \ref{prob:unformal_interp} is trivially solved by returning $\tilde y  = 0$.}
	$$ \tmu(t)  \le  \frac{ 2^{41} (\smu)^5 \log^3  (40 \smu)}{T}.$$
\end{theorem}

While Theorem \ref{thm:uniformBound} appears to give a relatively weak bound, proving this statement is a key technical challenge. Ultimately, it is used in Section \ref{sec:tight_bound} as one of two main ingredients in proving the much tighter leverage function bound that yields Theorem \ref{thm:fullBound} and Theorem \ref{thm:informal_sample_dist}.

%\medskip
%\noindent\textbf{Proof outline for Theorem \ref{thm:uniformBound}} :
Towards a proof of Theorem \ref{thm:uniformBound}, we consider the operator $\Fmu$ defined in Section \ref{sec:notation}. % and the associated ridge leverage function $\tmu(t)$ of Definition \ref{def:ridgeScores}. Our goal is to exhibit a function $\ttmu(t)$ that upper bounds this ridge leverage function, is efficient to sample from, and is not too large -- with average value $\tilde O(\smu)$.
Since $\Fmu$ has statistical dimension $\smu$, $\Kmu = \Fmu^*\Fmu$ can have at most $2\smu$ eigenvalues $\ge \epsilon$:
\begin{align}\label{eq:nespBound}
\smu = \sum_{i=1}^\infty \frac{\lambda_i(\Kmu)}{\lambda_i(\Kmu)+\epsilon} \ge \sum_{i: \lambda_i(\Kmu) \ge \epsilon} \frac{\lambda_i(\Kmu)}{\lambda_i(\Kmu)+\epsilon} \ge \frac{\left | i: \lambda_i(\Kmu) \ge \epsilon \right |}{2}.
\end{align}
Thus, if we project
%it can be approximated up to additive error $\epsilon$ on its eigenvalues by a rank-$\smu$ operator. In our setting, this operator could be obtained by  projecting
$\Fmu$ onto the span of $\Kmu$'s top $2 \smu$ eigenfunctions (when $\mu$ is uniform on an interval these are the prolate spherical wave functions of Slepian and Pollak \cite{SlepianPollak:1961}) we will approximate $\Kmu$ up to its small eigenvalues. The total mass of these eigenvalues is bounded by:
\begin{align*}
\sum_{i: \lambda_i(\Kmu) \le \epsilon} \lambda_i(\Kmu) \le 2\epsilon \cdot
\sum_{i: \lambda_i(\Kmu) \le \epsilon} \frac{\lambda_i(\Kmu)}{\lambda_i(\Kmu)+\epsilon}
\le 2\epsilon \cdot \smu.
\end{align*}

Alternatively, instead of projecting onto the span of the eigenfunctions, we can approximate $\Kmu$ nearly optimally by projecting $\Fmu$ onto the span of a subset of $O(\smu)$ of its	``rows" -- i.e. frequencies in the support of $\mu$.
For finite linear operators, is well known that such a subset exists: the problem of finding these subsets has been studied extensively in the literature on randomized low-rank  matrix approximation under the name \emph{column subset selection} \cite{sarlos2006improved,BoutsidisMahoneyDrineas:2009a,DeshpandeRademacher:2010}. In Appendix \ref{app:leverage_scores} we show that an analogous result extends to the continuous operator $\Fmu$:

%\Cam{Check what I did below and in the proof of Theorem \ref{thm:css}/ Lemma \ref{lem:cssSpectral} with `distinct frequencies'. I kept it brief since its a minor point and didn't want to overcomplicate.}
 \begin{theorem}[Frequency subset selection]\label{thm:css}
 %Let $\Fmu$ be as in \eqref{eq:informal}.
For some $s \le \lceil 36 \cdot \smu\rceil$ there exists a set of distinct frequencies $\xi_1,\ldots,\xi_s \in \CC$ such that, if $\bv{C}_s: L_2(T) \rightarrow \CC^{s}$ and $\bv{Z}: L_2(\mu) \rightarrow \CC^{s}$ are defined by:
 \begin{align}\label{eq:needForPoint}
 [\bv{C}_sg](j) &= \frac{1}{T}\int_0^T g(t) e^{-2 \pi i \xi_j t} dt & \bv{Z} &= (\bv{C}_s \bv{C}_s^*)^{-1} \bv{C}_s \Fmu^*,\footnotemark
%F_s \bv{x}(t) = \sum_{i=1}^n x_i e^{2\pi i t \xi_i}
 \end{align}
 \footnotetext{The fact that $\xi_1,\ldots,\xi_s$ are distinct ensures that $(\bv{C}_s \bv{C}^*_s)^{-1}$ exists.}
then
\begin{align}\label{eq:frobNormBound}
 \tr(\Kmu - \bv{C}_s^*\bv{Z}\bv{Z}^*\bv{C}_s)  \le 4\epsilon \cdot \smu.
\end{align}
\end{theorem}
\noindent Note that, if $\varphi_t \in L_2(\mu)$ is defined $\varphi_t(\xi) = e^{-2\pi i t \xi}$ and  $\bs{\phi}_t \in \CC^s$ is defined $\bs{\phi}_t(j) = \varphi_t(\xi_j)$, we have:
\begin{align*}
\tr(\Kmu - \bv{C}_s^*\bv{Z}\bv{Z}^*\bv{C}_s) = \frac{1}{T} \int_{t \in[0,T]}  \norm{\varphi_t - \bv Z^* \bs{\phi}_t}_\mu^2\, dt.
\end{align*}

%We give a proof of Theorem \ref{thm:css} in Appendix \ref{app:leverage_scores}. The theorem can be seen as a generalization of known column subset selection results for finite matrices to the operator setting.

\spara{Leverage function bound proof sketch.}
With Theorem \ref{thm:css} in place, we explain how to use this result to prove Theorem \ref{thm:uniformBound}, i.e., to establish a universal bound on the leverage function of $\Fmu$.
For the sake of exposition, we use the term ``row'' of an operator $\mathcal{A}: L_2(\mu) \rightarrow L_2(T)$ to refer to the corresponding operator restricted to some time $t$. We use the term ``column'' of an operator as the adjoint of a row of $\mathcal{A}^*: L_2(T) \rightarrow L_2(\mu)$, i.e., the adjoint operator restricted to some frequency  $\xi$.

By Theorem \ref{thm:css}, $\bv C_s^* \bv Z:  L_2(\mu) \rightarrow  L_2(T)$ (the projection of $\Fmu^*$ onto the range of $\bv{C}_s$) closely approximates the operator $\Fmu^*$ yet has columns spanned by just $O(\smu)$ frequencies: $\xi_1,\ldots,\xi_s$. Thus, for any $\alpha \in L_2(\mu)$, $\bv C_s^* \bv Z \alpha \in L_2(T)$ is  just an $O(\smu)$ sparse Fourier function. Using the maximization characterization of Definition \ref{def:ridgeScores}, we can thus bound the time domain ridge leverage function of $\bv C_s^* \bv Z$ by appealing to known smoothness bounds for Fourier sparse functions \cite{ChenPrice:2018}, even for $\epsilon = 0$. When $\epsilon = 0$, the ridge leverage function is known as the \emph{standard leverage function} in the randomized numerical linear algebra literature, and we will refer to them as such.

We can use a similar argument to bound the row norms  of the residual  operator $[\Fmu^* - \bv C_s^* \bv Z]$. The columns of this residual operator are each spanned by $O(\smu)$ frequencies, and so are again sparse Fourier functions whose smoothness we can bound. This smoothness ensures that no row can have norm significantly  higher than average.

Finally, we note that the time domain ridge leverage function of $\Fmu$ is approximated to within a constant factor by the sum of the standard row leverage function of $\bv C_s^* \bv Z$ along with row norms of $\Fmu-\bv  C_s^*  \bv Z$. This gives us a bound on $\Fmu$'s ridge leverage function. We prove this formally below:
%
%We start  by giving a simple argument that the time domain ridge leverage scores of  $\Fmu$ can be approximated as the sum of the standard row leverage scores $\bv F_s^* \bv Z$ plus the row norms of the error $\Fmu^* -  \bv  F_s \bv Z$.
%We have:
\begin{theorem}[Ridge leverage function approximation]\label{thm:css2}
Let $\bv C_s$ and $\bv Z$ be the operators guaranteed to exist by Theorem \ref{thm:css}. Let $\ell(t)$ be the standard leverage function of $t$ in $\bv C_s^* \bv Z$:\footnote{Analogously to how $[\Fmu^* \alpha](t)$ is used in Definition \ref{def:ridgeScores}, while $L_2(T)$ is formally a space of equivalence classes of functions, here we use $\bv{C}^*_s \bv{Z} \alpha$ to denote the specific representative of the equivalence class $\bv{C}^*_s \bv{Z} \alpha \in L_2(T)$ given by $[\bv{C}^*_s \bv{Z}\alpha](t) = \sum_{j=1}^s [\bv{Z}\alpha](j) \cdot e^{2\pi i \xi_j t} = \langle \bs{\phi}_t, \bv{Z}\alpha \rangle_{\CC^s}$. In this way, we can consider the pointwise value $[\bv{C}^*_s \bv{Z} \alpha](t)$.}
\begin{align*}
\ell(t) \eqdef \max_{\{\alpha \in L_2(\mu):\, \norm{\alpha}_\mu > 0\}} \frac{1}{T} \cdot \frac{|[\bv  C_s^* \bv Z\alpha](t)|^2}{\norm{\bv  C_s^* \bv Z\alpha}_T^2}.
\end{align*}
%where $w_t(i) \eqdef e^{-2\pi i  t \xi_i}$.
 Let $r(t)$ be  the residual:
 \begin{align*}
\frac{1}{T} \cdot \norm{\varphi_t-\bv Z^* \bs{\phi}_t}_\mu^2
 \end{align*}
 where $\varphi_t$ and $\bs{\phi}_t$  are as defined in Theorem \ref{thm:css}.
 Then for all $t$:
\begin{align*}
\tmu(t) \le 2 \cdot \left (\ell(t)+\frac{r(t)}{\epsilon}\right)
\end{align*}
\end{theorem}
\begin{proof}
For any  $\alpha \in L_2(\mu)$
we can write $[\Fmu^* \alpha](t) = \langle \varphi_t , \alpha \rangle_\mu$ and $[\bv C_s^* \bv Z\alpha](t) = \langle \bs{\phi}_t, \bv Z \alpha \rangle_{\CC^s} = \langle \bv Z^* \bs{\phi}_t , \alpha \rangle_\mu $. 
By the maximization characterization of the ridge leverage function in Definition \ref{def:ridgeScores},
\begin{align*}
\tmu(t) &= \frac{1}{T} \cdot \max_{\{\alpha \in L_2(\mu):\norm{\alpha}_\mu > 0\}} \frac{ \langle \varphi_t , \alpha \rangle_\mu^2}{\norm{\Fmu^* \alpha}_T^2 + \epsilon \norm{\alpha}_\mu^2}\\
&\le  \frac{2}{T} \cdot \max_{\{\alpha \in L_2(\mu):\norm{\alpha}_\mu > 0\}} \left (\frac{ \langle
\bv Z^* \bs{\phi}_t , \alpha \rangle_\mu^2}{\norm{\Fmu^* \alpha}_T^2} + \frac{ \langle \varphi_t -\bv  Z^* \bs{\phi}_t , \alpha \rangle_\mu^2}{ \epsilon \norm{\alpha}_\mu^2} \right )\\
&\le  \frac{2}{T} \cdot \max_{\{\alpha \in L_2(\mu):\norm{\alpha}_\mu > 0\}} \left (\frac{ \langle \bv Z^* \bs{\phi}_t , \alpha \rangle_\mu^2}{\norm{\bv  C_s^* \bv Z \alpha}_T^2} + \frac{\norm{\varphi_t -\bv Z^* \bs{\phi}_t}_\mu^2}{ \epsilon} \right )\\
&= 2 \cdot \left (\ell(t)+\frac{r(t)}{\epsilon}\right)
\end{align*}
where the second to last line follows from observing that due to Cauchy-Schwarz,
$$\langle \varphi_t -\bv Z^* \bs{\phi}_t , \alpha \rangle_\mu^2 \le \norm{\alpha}_\mu^2  \cdot \norm{\varphi_t -\bv  Z^* \bs{\phi}_t}_\mu^2,$$
and that, letting $\mathcal{P}_s =  \bv{C}_s^*(\bv{C}_s \bv{C}_s^*)^{-1} \bv{C}_s$:
\begin{align*}
\norm{\Fmu^* \alpha}_T^2 &= \langle \alpha, \Fmu \Fmu^* \alpha \rangle_\mu\\
&\ge \langle \alpha, \Fmu \mathcal{P}_s \Fmu^* \alpha \rangle_\mu\\
&= \langle \alpha, \bv Z^*\bv C_s \bv C_s^* \bv Z \alpha \rangle_\mu = \norm{\bv C_s^* \bv Z \alpha}_T^2.
\end{align*}
In the above, the inequality is due to the fact that $\mathcal{P}_s$ is an orthogonal projection, so $\mathcal{P}_s \preceq {\cal I}_\mu$.
This completes the proof.
\end{proof}

With Theorem \ref{thm:css2} in place, we now bound $\btmu(t) = 2 \left (\ell(t) + \frac{r(t)}{\epsilon}\right)$, which yields a uniform bound on the true ridge leverage scores.

\begin{lemma}\label{lem:levBound}
Let $\ell(t),r(t)$ be as defined in Theorem \ref{thm:css2} and $\btmu(t) \eqdef 2 \cdot \left (\ell(t)+\frac{r(t)}{\epsilon}\right)$.
For all  $t \in [0,T]$:
\begin{align*}
\btmu(t) \le  \frac{ 15400(36 \smu+2)^5 \log^3  (36 \smu+2)}{T}.\end{align*}
\end{lemma}
Combining Lemma \ref{lem:levBound} with Theorem \ref{thm:css2} yields Theorem \ref{thm:uniformBound}. We just simplify the constants by  noting that for $\epsilon \le 1$, $\smu \ge \frac{\tr(\Kmu)}{2} = \frac{1}{2}$ and so $36\smu + 2 \le 40 \smu$.

%\begin{proof}
%	This follows directly from Theorem \ref{thm:css2}  and Lemma \ref{lem:levBound}. We just simplify the constants by  noting that for $\epsilon \le 1$, $\smu \ge \frac{\tr(\Kmu)}{2} = \frac{1}{2}$ and so $36\smu + 2 \le 40 \smu$.
%	%log_2(15400*40^5*((8*log_2 10)^3+1))
%\end{proof}

\begin{proof}[Proof of Lemma \ref{lem:levBound}]
We separately bound the leverage score $\ell(t)$ and residual $r(t)$ components of $\btmu(t)$ using a similar argument based on the smoothness of sparse Fourier functions for both. %For simplicity of notation we will denote $\beta(t) = \frac{1}{T}  \cdot \min \left (\frac{\smu \log \smu}{1-|(2t/T-1|}, \smu^4 \log^3 \smu \right )$ throughout the proof.
Specifically, for both bounds we employ  the following smoothness bound of Chen et al.:

\begin{lemma}[Follows from Lemma 5.1 of \cite{ChenKanePrice:2016}]\label{lem:92}
For any $f(t) = \sum_{j=1}^k v_j e^{2\pi i \xi_j t}$,
\begin{align*}
\max_{x \in [0,T]} \frac{|f(x)|^2}{\norm{f}_T^2} = 1540 \cdot k^4 \log^3 k.
\end{align*}
\end{lemma}
\begin{proof}
This follows from Lemma 5.1 of \cite{ChenKanePrice:2016}, which gives the bound without an explicit constant. It is not hard to check that their proof gives the constant of $1540$ stated above.
\end{proof}

\medskip
\noindent \textbf{Bounding the leverage scores $\ell(t)$ of $\bv C_s^* \bv Z$.}
\medskip

For every $\alpha \in L_2(\mu)$, $\bv  C_s^*  \bv Z \alpha$ is an $s = O(\smu)$ sparse Fourier function. Specifically, we have:
$$[\bv C_s^*  \bv Z \alpha](t) = \sum_{j=1}^s [\bv{Z}\alpha](j) \cdot e^{2\pi i \xi_j t},$$
for frequencies $\xi_1,\ldots,\xi_s \in \CC$ given by Theorem \ref{thm:css}. We can thus directly apply  Lemma \ref{lem:92} giving for any  $t  \in [0,T]$:
\begin{align}\label{eq:lBound}
\ell(t) &\eqdef \max_{\{\alpha \in L_2(\mu):\norm{\alpha}_\mu > 0\}} \frac{1}{T} \cdot \frac{|[\bv  C_s^* \bv Z\alpha](t)|^2}{\norm{\bv  C_s^* \bv Z\alpha}_T^2}\nonumber\\
&\le \max_{\{\alpha \in L_2(\mu):\norm{\alpha}_\mu > 0\}} \left [\frac{1}{T} \cdot \max_{t' \in [0,T]} \frac{|[\bv  C_s^* \bv Z\alpha](t')|^2}{\norm{\bv  C_s^* \bv Z\alpha}_T^2}\right]\nonumber\\
&\le \frac{1540}{T} \cdot s^4 \log^3 s
\end{align}

\medskip
\noindent \textbf{Bounding the residuals $r(t)$.}
\medskip

We first give some intuition. To bound  the squared row norms of the residual $\Fmu^* - \bv C_s^* \bv Z$ we show that each ``column'' of this residual is an $s+1 = O(\smu)$ sparse Fourier function. Thus, applying Lemma \ref{lem:92}, no entry's squared value can significantly  exceed the average squared value in the column. This lets us show that no squared row norm $r(t)$ can significantly exceed the average squared row norm, which is bounded by Theorem \ref{thm:css}.

Concretely, define $\vartheta_\xi \in L_2(T)$ by $\vartheta_\xi(t) \eqdef e^{2\pi i t \xi}$, and notice that given $g \in L_2(T)$ the function $\xi \mapsto \langle \vartheta_\xi , g\rangle_T$ is equal to $\Fmu g$ in the $L_2(T)$ sense (i.e., is a member of the equivalence class $\Fmu g$).
For $\xi \in \RR$, let  $\bv  z_\xi \in \CC^s$ be given by $\bv  z_\xi(j) = \langle \vartheta_\xi, \bv C^*_s (\bv{C}_s \bv{C}_s^*)^{-1} \bv e_j \rangle^*_T$  where $\bv  e_j$ is the $j^{th}$ standard basis vector in $\CC^s$. The function $\xi \mapsto \langle z_\xi, \bs{\phi}_t \rangle = \sum_{j=1}^s \bv{z}_\xi^*(j) e^{- 2 \pi i \xi_j t}$ is equal in the $L_2(\mu)$ sense to $\bv Z^* \bs{\phi}_t$. Let us define: 
$$r_\xi(t) = e^{-2\pi i \xi t} - \sum_{j=1}^s \bv{z}_\xi^*(j) e^{- 2 \pi i \xi_j t}.$$
For a fixed $t$, consider the function $\xi \mapsto r_\xi(t)$, which we denote by $r_\cdot(t)$. We have $r_\cdot(t) =  \varphi_t - \bv Z^* \bs{\phi}_t$, again in the $L_2(\mu)$ sense. Thus, we can write 
\begin{align}\label{eq:colrewrite}
r(t) &= \frac{1}{T} \norm{\varphi_t - \bv Z^* \bs{\phi}_t}_\mu^2\nonumber\\
& = \frac{1}{T}\int_{\xi \in \RR} |r_\xi(t)|^2\, d\mu(\xi).
\end{align}
Further, for a fixed $\xi$, if we consider the function $t \mapsto {r}_\xi(t)$, which we denote by ${r}_\xi(\cdot)$, we notice that it is a $s+1 = O(\smu)$ sparse Fourier function,
so applying Lemma \ref{lem:92} we have for any $\xi\in \RR$ and $t \in [0,T]$:
\begin{align}\label{eq:columnBound}
\frac{|r_\xi(t)|^2}{\norm{r_\xi(\cdot)}_T^2} &\le 1540 (s+1)^4 \log^3  (s+1).
\end{align}
Combining \eqref{eq:columnBound} with \eqref{eq:colrewrite} we can thus bound for any  $t \in [0,T]$:
\begin{align}\label{eq:redBound}
r(t) &\le 1540 (s+1)^4 \log^3  (s+1) \cdot \frac{1}{T}\int_{\xi \in \RR} \norm{r_\xi(\cdot)}_T^2 \,  d\mu(\xi)\nonumber\\
& = 1540 (s+1)^4 \log^3  (s+1) \cdot \frac{1}{T^2} \int_{w \in[0,T]}\int_{\xi \in \RR} |r_\xi(w)|^2 \, d\mu(\xi)\, dw\nonumber\\
& = 1540 (s+1)^4 \log^3  (s+1) \cdot \frac{1}{T^2} \int_{w \in[0,T]}  \norm{\varphi_w - \bv Z^* \bs{\phi}_w}_\mu^2\, dw
\end{align}
where the last bound again follows from \eqref{eq:colrewrite}. By Theorem \ref{thm:css} we have $ \frac{1}{T} \int_{w \in[0,T]}  \norm{\varphi_w - \bv Z^* \bs{\phi}_w}_\mu^2\, dw \le 4\epsilon \cdot \smu$. Plugging into \eqref{eq:redBound} and using that we can choose $s \le 36 \cdot \smu +1$, for all $t \in [0,T]$:
\begin{align}\label{eq:rBound}
r(t) \le  \frac{\epsilon \cdot 6160(36 \smu+2)^5 \log^3  (36 \smu+2)}{T}.
\end{align}
Combining \eqref{eq:lBound} and \eqref{eq:rBound} completes the proof of Lemma \ref{lem:levBound} since $\btmu(t) \eqdef 2 \cdot \left (\ell(t)+\frac{r(t)}{\epsilon}\right)$ and thus
$$\btmu(t) \le  \frac{ 15400(36 \smu+2)^5 \log^3  (36 \smu+2)}{T}.$$
\end{proof}

Theorem \ref{thm:uniformBound} gives a universal uniform bound on the ridge leverage scores corresponding to measure $\mu$ in terms of $\smu$. If we directly sample time points according to the uniform distribution over $[0,T]$, this theorem shows that $\poly(\smu)$ samples and $\poly(\smu)$ runtime suffice to apply Theorem \ref{thm:mainAlg} and solve Problem \ref{prob:unformal_interp} with good probability. This is already a surprising result, showing that the simplest sampling scheme, uniform random sampling, can give bounds in terms of the optimal complexity  $\smu$ for \emph{any $\mu$}. Existing methods with similar complexity, such as those that interpolate bandlimited signals using prolate spheroidal wave functions \cite{XiaoRokhlinYarvin:2001,ShkolniskyTygertRokhlin:2006} require nonuniform sampling. Methods that use uniform sampling, such as truncated Whittaker-Shannon, have sample complexity depending polynomially rather than logarithmically on the desired error $\epsilon$.

\subsection{Gap-based leverage score bound}\label{sec:gap}

Our final result gives a much tighter bound on the ridge leverage scores than the uniform bound of Theorem \ref{thm:uniformBound}. The key idea is to show that the bound is loose for $t$ bounded away from the edges of $[0,T]$. Specifically we have:
\begin{theorem}[Gap-Based Leverage Score Bound]\label{thm:gapBound}
For all $t$,
$$ \tmu(t)  \le  \frac{ \smu}{  \min(t,T-t)}.$$
\end{theorem}
\begin{proof}

Consider $t \in [0,T/2]$. We will show that $ \tmu(t)  \le \frac{\smu}{t} .$ A symmetric proof will hold for $t \in [T/2,T]$, giving the theorem.
We define an auxiliary operator: $\Fmut: L_2(T) \rightarrow L_2(\mu)$ which is given by restricting the integration in $\Fmu$ to $[0,t]$. Specifically, for $f \in L_2(T)$ we have:
\begin{align}\label{eq:f1def}
[\Fmut f](\xi) = \frac{1}{T} \int_{0}^{t} f(s) e^{-2 \pi i s \xi}\, ds.
\end{align}
We can see that $[\Fmut^* g](s) = \int_\RR g(\xi) e^{2\pi i s \xi}\, d\mu(\xi)$ for $s \in [0,t]$ and $[\Fmut^* g](s) = 0$ for $s \in (t,T]$. %Similarly,  $[\Fmut^* g](t) = \int_\RR g(\xi) e^{2\pi i t \xi}\, d\mu(\xi)$ for $t \in [T/2,T]$ and $[\Fmut^* g](t) = 0$ for $t \in [0,T/2)$.
We will use the leverage score of some $s \in [0,t]$ in the restricted operator $\Fmut$ to upper bound those of $t$ in $\Fmu$. We start by defining these scores analogously to Definition \ref{def:ridgeScores} for $\Fmu$.

\begin{defn}[Restricted ridge leverage scores]\label{def:generalRidgeScores}
For probability measure $\mu$ on $\RR$, time length $T$, $t \in [0,T]$ and $\epsilon \ge 0$, define the $\epsilon$-ridge leverage score of $s \in [0,t]$ in $\Fmut$ as:
\begin{align*}
 \tau_{\mu,\epsilon,t}(s) = \frac{1}{T} \cdot  \max_{\{\alpha \in L_2(\mu):\, \norm{\alpha}_\mu > 0\}} \frac{|[\Fmut \alpha](s) |^2 }{\norm{\Fmut^* \alpha}_T^2 + \epsilon \norm{\alpha}_\mu^2}.
 \end{align*}
\end{defn}

We have the following leverage score properties, analogous to those given for $\Fmu$ in Theorem \ref{thm:leverageProps}:
\begin{theorem}[Restricted leverage score properties]\label{thm:genLeverageProps}
Let  $ \tau_{\mu,\epsilon,t}(s)$ be as defined in Definition \ref{def:generalRidgeScores}.
\begin{itemize}
\item The leverage scores integrate to the statistical dimension:
\begin{align}\label{eq:genStatDim}
\int_{0}^t \tau_{\mu,\epsilon,t}(s)\, ds = s_{\mu,\epsilon,t} \eqdef \tr(\Fmut^*\Fmut(\Fmut^*\Fmut+\epsilon \mathcal{I}_T)^{-1}).
\end{align}
\item Inner Product Characterization: Letting $\varphi_{s} \in L_2(\mu)$ have $\varphi_{s}(\xi) = e^{-2\pi i s \xi}$ for $s \in [0,t]$,
\begin{align}\label{eq:genInnerProd}
\tau_{\mu,\epsilon,t}(s) = \frac{1}{T} \cdot \langle \varphi_{s},  (\Fmut \Fmut^* + \epsilon \Imu)^{-1}  \varphi_{s} \rangle_\mu.
\end{align}
\item Minimization Characterization:
\begin{align}\label{eq:genMin}
\tau_{\mu,\epsilon,t}(s) = \frac{1}{T} \cdot \min_{\beta \in L_2(T)} \frac{\norm{\Fmut \beta - \varphi_{s}}_\mu^2}{\epsilon} + \norm{\beta}_{T}^2.
\end{align}
\end{itemize}
\end{theorem}
We first show that the restricted leverage scores of Definition \ref{def:generalRidgeScores} are not too large on average.
\begin{claim}[Restricted statistical dimension bound]\label{clm:splitAverage}
\begin{align}
\int_{0}^T \tau_{\mu,\epsilon,t}(s)\, ds  \le \smu.
\end{align}
\end{claim}
\begin{proof}
Via \eqref{eq:genStatDim} we have $\int_{0}^t \tau_{\mu,\epsilon,t}(s)\, ds = s_{\mu,\epsilon,t}$ which we can write as:\begin{align*}
s_{\mu,\epsilon,t}  =  \tr(\Fmut^*\Fmut(\Fmut^*\Fmut+\epsilon \mathcal{I}_T)^{-1}) = \sum_{i=1}^\infty \frac{\lambda_i(\Fmut^*\Fmut)}{\lambda_i(\Fmut^*\Fmut)+\epsilon}.
\end{align*}
From Claim~\ref{claim:dominated-expectation} we have $\Fmut \Fmut^* \preceq \Fmu \Fmu^* = \Gmu$. Since $(\Fmut \Fmut^*+\epsilon{\cal I}_\mu)^{-1/2}(\Fmut \Fmut^*)(\Fmut \Fmut^*+\epsilon{\cal I}_\mu)^{-1/2} = {\cal I}_\mu - \epsilon(\Fmut \Fmut^*+\epsilon{\cal I}_\mu)^{-1}$ and $(\Gmu +\epsilon{\cal I}_\mu)^{-1/2}\Gmu(\Gmu +\epsilon{\cal I}_\mu)^{-1/2} = {\cal I}_\mu - \epsilon(\Gmu+\epsilon{\cal I}_\mu)^{-1}$ we have from Claim~\ref{claim:bound-to-inverse-bound} $(\Fmut \Fmut^*+\epsilon{\cal I}_\mu)^{-1/2}(\Fmut \Fmut^*)(\Fmut \Fmut^*+\epsilon{\cal I}_\mu)^{-1/2} \preceq (\Gmu +\epsilon{\cal I}_\mu)^{-1/2}\Gmu(\Gmu +\epsilon{\cal I}_\mu)^{-1/2}$, and since the trace is monotone for trace-class operators, we have
\begin{align*}
s_{\mu,\epsilon,t} =& \tr((\Fmut \Fmut^*+\epsilon{\cal I}_\mu)^{-1/2}(\Fmut \Fmut^*)(\Fmut \Fmut^*+\epsilon{\cal I}_\mu)^{-1/2})\\ \leq & \tr((\Gmu +\epsilon{\cal I}_\mu)^{-1/2}\Gmu(\Gmu +\epsilon{\cal I}_\mu)^{-1/2}) = \smu
\end{align*}
which gives the claim.
\end{proof}
From Claim \ref{clm:splitAverage} we immediately have:
\begin{claim}\label{clm:splitAverage2}
There exists  $s^\star \in [0,t]$ with $\tau_{\mu,\epsilon,t}(s^\star) \le \frac{\smu}{t}$.
\end{claim}
\begin{proof}
Assume for the sake of contradiction that $ \tau_{\mu,\epsilon,t}(s) > \frac{\smu}{t}$ for all $s \in [0,t]$.
The by \eqref{eq:genStatDim},
$$\int_{0}^t \tau_{\mu,\epsilon,t}(s)\, ds > t \cdot \frac{\smu}{t} =  \smu.$$ This contradicts Claim \ref{clm:splitAverage}, giving the claim.
\end{proof}

We now show that the leverage score of $s^\star$ in $\Fmut$ upper bounds the leverage score of $t$ in $\Fmu$, completing the proof of Theorem \ref{thm:gapBound}.
We apply the minimization characterization of Theorem \ref{thm:genLeverageProps}, equation \eqref{eq:genMin}, showing that by simply shifting an optimal solution for $s^\star$ we can show the existence of a good solution for $t$, upper bounding its leverage score by that of $s^\star$ and giving $\tmu(t) \le \tau_{\mu,\epsilon,t}(s^\star) \le \frac{\smu}{t}$ by Claim \ref{clm:splitAverage2}.

 Formally, by Claim \ref{clm:splitAverage2} and \eqref{eq:genMin}, there is some $\beta^\star \in L_2(T)$ achieving:
\begin{align}\label{eq:initialShiftBound}
\frac{1}{T} \cdot \frac{\norm{\Fmut \beta^\star - \varphi_{s^\star}}_\mu^2}{\epsilon} + \norm{\beta^\star}_{T}^2 &= \tau_{\mu,\epsilon,t}(s^\star) \le \frac{\smu}{t}.
\end{align}
We can assume without loss of generality that $\beta^\star(s) = 0$ for $s \notin [0,t]$, since $\Fmut \beta^\star$ is unchanged if we set $\beta^\star(s) = 0$ on this range and since doing this cannot increase $\norm{\beta}_T^2$. Now, let $\bar \beta \in L_2(T)$ be given by $\bar \beta(s) = \beta^\star(s- (t-s^\star))$. That is, $\bar  \beta$ is just $\beta^\star$ shifted from the range $[0,t]$ to the range $[t-s^\star,2t-s^\star]$. Note that since we are assuming $t \le T/2$, $[t-s^\star,2t-s^\star] \subset [0,T]$. For any $\xi$:
\begin{align}
[\Fmu \bar \beta](\xi) &= \frac{1}{T} \int_{0}^T \bar \beta(s) e^{-2\pi i s \xi} ds\nonumber\\
& = \frac{1}{T} \int_{t-s^\star}^{2t-s^\star} \beta^\star(s - (t-s^\star)) e^{-2\pi i s \xi} d s\nonumber\\
& =  \frac{1}{T} \int_{0}^{t} \beta^\star(s) e^{-2\pi i (s +  (t-s^\star))} \xi ds\nonumber\\
& =  [\Fmut \beta^\star](\xi) \cdot e^{-2\pi i (t-s^\star)\xi}
\label{eq:shift}.
\end{align}
Now,
$$\varphi_{t}(\xi) = e^{-2\pi i t \xi} = e^{-2\pi i (t-s^\star) \xi} \cdot \varphi_{s^\star}(\xi).$$
Combined with \eqref{eq:shift} this gives:
\begin{align}\label{eq:shiftNorm}
\norm{\Fmu \bar \beta - \varphi_{t}}_\mu^2 = \int_\xi \left |[\Fmu \bar \beta](\xi) - \varphi_{t}\right |^2\, d\mu(\xi) &= \int_\xi \left |([\Fmut \beta^\star](\xi) - \varphi_{s^\star})\cdot e^{-2\pi i (t-s^\star) \xi} \right |^2\, d\mu(\xi)\nonumber\\
&= \int_\xi \left |([\Fmut \beta^\star](\xi) - \varphi_{s^\star}) \right |^2\, d\mu(\xi)\nonumber\\
&= \norm{\Fmut \beta^\star-\varphi_{s^\star}}_\mu^2.
\end{align}
Finally, noting that $\norm{\bar \beta}_T = \norm{\beta^\star}_T$ and applying the minimization characterization of Theorem \ref{thm:leverageProps}, the bound in \eqref{eq:shiftNorm} along with \eqref{eq:initialShiftBound} gives:
\begin{align*}
\tmu(t) \le \frac{1}{T} \cdot \frac{\norm{\Fmu \bar \beta - \varphi_{t}}_\mu^2}{\epsilon} + \norm{\bar \beta}_T^2 = \frac{\norm{\Fmut \beta^\star - \varphi_{s^\star}}_\mu^2}{\epsilon} + \norm{\beta^\star}_T^2 \le \frac{\smu}{t},
\end{align*}
which completes the theorem.

\end{proof}

\subsection{Nearly tight leverage score bound}
\label{sec:tight_bound}
Combining Theorems \ref{thm:uniformBound} and \ref{thm:gapBound} gives our tight, spectrum blind leverage score bound:
\begin{theorem}[Spectrum Blind Leverage Score Bound]\label{thm:fullBound}
For any $\alpha,T \ge 0$ let $\tilde \tau_{\alpha}(t)$ be given by:
\begin{align*}
\tilde \tau_{\alpha}(t) = \begin{cases} \frac{\alpha}{256 \cdot \min(t,T-t)}\text{  for }t \in [T/\alpha^6, T(1-1/\alpha^6)]\\
\frac{\alpha^6}{ T} \text{ for }t \in [0,T/\alpha^6] \cup [T(1-1/\alpha^6),T].
\end{cases}
\end{align*}
For any probability measure $\mu$, $T \ge 0$, $0 \le \epsilon \le 1$ and $t \in [0,T]$, if $\alpha \ge 256 \cdot \smu$:
\begin{align*}
\tmu(t) \le \tilde \tau_{\alpha}(t)\text{ and }\tilde s_{\alpha} \eqdef \int_0^T \tilde \tau_{\alpha}(t)\, dt \le \frac{ \alpha \cdot \log \alpha}{19}.
\end{align*}
\end{theorem}
A visualization of $\tilde \tau_{\alpha}$ is given in Figure \ref{fig:our_dist}.
\begin{proof}
The fact that $\tmu(t) \le \tilde \tau_{\alpha}(t)$ follows from Theorems \ref{thm:uniformBound} and \ref{thm:gapBound}:
\begin{itemize}
\item For $t \in [T/\alpha^6, T(1-1/\alpha^6)]$,  by Theorem \ref{thm:gapBound} if $\alpha \ge 256 \cdot \smu$ we have $$\tilde \tau_{\alpha}(t) = \frac{\alpha}{256\cdot \min(t,T-t)} \ge \tmu(t).$$
\item For $t \in [0,T/\alpha^6] \cup [T(1-1/\alpha^6),T]$, by Theorem \ref{thm:uniformBound} we can bound,
$$\tmu(t) \le \frac{2^{41} \smu^5 \log^3 (40\smu)}{T} \le \frac{2^{47} \smu^6}{T} \le \frac{\alpha^6}{T}$$
for $\alpha \ge 256 \cdot \smu$. Note that the second inequality uses that $\log^3 (40x) \le 64 x$ for any $x$.
\end{itemize}
The integral of the approximate scores $\tilde s_{\alpha}$ is bounded as:
\begin{align}
\int_0^T \tilde \tau_{\alpha}(t)\, dt &= \int_{T/\alpha^6}^{T(1-1/\alpha^6)} \frac{\alpha}{256 \cdot \min(t,T-t) }\, dt + 2 \int_0^{T/\alpha^6} \frac{\alpha^6}{T}\, dt\nonumber\\
&= \frac{2}{256}\int_{T/\alpha^6}^{T/2} \frac{\alpha}{t }\, dt + 2\nonumber\\
&= \frac{\alpha}{128} \cdot [\log(T/2)-\log(T/\alpha^6)] + 2\label{eq:cdf}\\
&\le \frac{6 \alpha \log \alpha}{128} + 2 \le \frac{\alpha \log \alpha}{19}.\nonumber
\end{align}
where the last inequality follows since for $\epsilon \le 1$, $\smu \ge 1/2$ and so $\alpha \ge 128$.
\end{proof}

\subsection{Putting it all together: generic signal reconstruction}\label{sec:puttingTogether}

Finally, we combine the leverage score bound of Theorem \ref{thm:fullBound} with Theorem \ref{thm:mainAlg} to give our main algorithmic result,  Theorem \ref{thm:informal_main} (and as a corollary, Theorem \ref{thm:informal_sample_dist}). We state the full theorem below:

\begin{reptheorem}{thm:informal_main}[Main result, algorithmic complexity]
Consider any measure $\mu$, for which we can compute the kernel function $k_\mu(t_1,t_2) = \int_{\xi \in \RR}  e^{-2\pi i (t_1 - t_2)} d\mu(\xi)$ for any $t_1,t_2\in[0,T]$ in time $Z$.

Let $\tilde \tau_{\alpha}(t)$ be as defined in Theorem \ref{thm:fullBound}.
For any $\epsilon \le \opnorm{\Kmu}$ and $T > 0$, let $\ttmu(t) = \tilde \tau_{\alpha}(t)$ for $\alpha = \beta \cdot \smu$ with $\beta \ge 256$.
Algorithm \ref{alg:main} applied with $\ttmu(t)$ and failure probability $\delta$ returns $t_1,\ldots,t_{s} \in [0,T]$ and $\bv{z} \in \CC^{s}$ such that $\tilde y(t) = \sum_{i=1}^{s} \bv{z}(i) \cdot k_\mu(t_i,t)$ solves Problem \ref{prob:unformal_interp} with parameter $6 \epsilon$ and probability $\ge 1-\delta$. That is, with probability of at least $1-\delta$:
\begin{align*}
\norm{\tilde y - y}_T^2  \le 6 \epsilon \norm{x}_\mu^2 + 8\norm{n}_T^2.
\end{align*}
The algorithm queries $y+n$ at $s$ points and runs in $ O \left  (s^2 \cdot Z + s^\omega \right )$ time where $$s = O\left (\beta \cdot \smu \log(\beta \cdot \smu) \cdot [ \log(\beta \cdot \smu)  + 1/\delta] \right  ) = \tilde O \left (\frac{\beta \cdot \smu}{\delta}\right).$$ The output $\tilde y(t)$ can be evaluated in $O(s \cdot Z)$ time for any  $t$ using Algorithm \ref{alg:main2}.

\end{reptheorem}
Note that if we want to solve Problem \ref{prob:unformal_interp} with parameter $\epsilon$, it suffices to apply Theorem \ref{thm:informal_main} with parameter $\epsilon' = \epsilon/6$. The asymptotic complexity will be identical since, by \eqref{constAdjust}, $s_{\mu,\epsilon/6} \le 6 \smu$.
\begin{proof}
The theorem follows directly from Theorem \ref{thm:mainAlg}, along with Theorem \ref{thm:fullBound} which shows that, for $\alpha = \beta \cdot \smu$ with $\beta \ge c_1$ and  $\ttmu(t) = \tilde \tau_{\alpha}(t)$ we have:
\begin{enumerate}
\item $\ttmu(t) \ge \tmu(t)$ for all $t \in [0,T]$.
\item $\tsmu = \int_{0}^T \ttmu(t)\, dt = O \left (\beta \cdot \smu \log(\beta \cdot \smu) \right )$.
\end{enumerate}
The runtime bound follows after noting that we can sample according to $\tau_{\alpha}$ in $W = O(1)$ time using inverse transform sampling since it is straightforward to derive an explicit expression for the CDF and compute the inverse (see~\eqref{eq:cdf}).

\end{proof}

\section{Lower bound}\label{sec:lb}

We conclude by  showing that the statistical dimension $\smu$ tightly characterizes the sample complexity  of solving Problem \ref{prob:unformal_interp}, under a very mild assumption on $\mu$ that holds for all natural constraints we discuss in this paper. Thus, Theorem \ref{thm:informal_sample_complexity} is tight up to logarithmic factors.

We first define a quantity, $n_{\mu,\epsilon}$ that gives a natural lower bound on $\smu$.
For any $\mu,\epsilon$, let
\begin{align}
n_{\mu,\epsilon} \eqdef \sum_{i=1}^\infty \mathbb{I}[\lambda_i(\Kmu) \ge \epsilon].
\end{align}
That is,  $n_{\mu,\epsilon}$ is the number of eigenvalues of $\Kmu$ that are larger than $\epsilon$. As shown in \eqref{eq:nespBound},
we always have $n_{\mu,\epsilon} \le 2 \smu$.
 We first prove that solving Problem \ref{prob:unformal_interp} requires $\Omega(n_{\mu,\epsilon})$ samples. We then show that, under a very mild constraint on $\mu$ (which holds for all $\mu$ we consider including sparse, bandlimited, multiband, Gaussian, and Cauchy-Lorentz), $n_{\mu,\epsilon} = \Omega(\smu)$.
 % -- that is, $n_{\mu,\epsilon}$ tightly lower bounds $\smu$.
 Thus, $\smu$ gives a tight bound on the query complexity of solving Problem \ref{prob:unformal_interp}.
\begin{theorem}[Lower bound in terms of eigenvalue count]\label{thm:initialLB} Consider a measure $\mu$, an error parameter $\epsilon > 0$, and any (possibly  randomized) algorithm that solves Problem \ref{prob:unformal_interp} with probability  $\ge 2/3$ for  any  function $y$ and makes at most $r$ (possibly adaptive) queries on any  input. Then $r \ge  n_{\mu,\ce\epsilon}/20$.
\end{theorem}
\begin{proof}
We describe a distribution on inputs $y$ on which any  deterministic algorithm that takes $r = o( n_{\mu,\ce\epsilon})$ samples on any input fails with probability $> 1/3$. The theorem then follows by Yao's principle.

\medskip
\spara{Notation:}
Let $v_1,\ldots,v_{n_{\mu,\ce\epsilon}} \in L_2(\mu)$ be the eigenfunctions of $\Gmu$ corresponding to its top  $n_{\mu,\ce\epsilon}$ eigenvalues. Let $\bv{Z}: L_2(\mu) \rightarrow \mathbb{C}^{n_{\mu,\ce\epsilon}}$ be the operator with $v_i$ as its $i^{th}$ row -- i.e., $[\bv{Z} g](i) = \langle v_i,g\rangle_{\mu}$. Note that $\bv{Z}$ has orthonormal rows. Let  $\bv{D} \in \RR^{ n_{\mu,\ce\epsilon} \times n_{\mu,\ce\epsilon}}$ be a diagonal matrix with $\bv{D}_{ii} = \sqrt{\lambda_i(\Kmu)}$. Let $\bv{U} =  \Fmu^* \bv{Z}^* \bv{D}^{-1}$. We can see that $\bv{Z} \Fmu \Fmu^* \bv{Z}^* = \bv{Z} \Gmu \bv{Z}^* = \bv{D}^2$ and hence,  $\bv{U}^* \bv{U} = \bv{D}^{-1} \bv{Z} \Fmu  \Fmu^* \bv{Z}^* \bv{D}^{-1} = \bv{I}$. While not needed for our proof, we can check that $\bv{U}: \CC^{n_{\mu,\ce\epsilon}} \rightarrow L_2(T)$ is an operator with columns corresponding to all eigenfunctions of $\Kmu$ with eigenvalue $\ge \ce\epsilon$.

\medskip
\spara{Hard Input Distribution:}
Let $\bv{c} \in \RR^{n_{\mu,\ce\epsilon}}$ be a random vector with each entry distributed independently as a Gaussian: $\bv{c}(i) \sim \mathcal{N}(0,\frac{1}{n_{\mu,\ce\epsilon}})$. Let $\bv{\bar c} = \bv{D}^{-1} \bv{c}$, $x = \bv{Z}^* \bv{\bar c}$, and the random input be $y  =  \Fmu^* x$. That is, $y =  \Fmu^* \bv{Z}^*  \bv{D}^{-1} \bv{c} = \bv{U} \bv{c}$ is a random linear combination of the top eigenfunctions of $\Kmu$. While, formally, $\Fmu^* x \in L_2(T)$ is an equivalence class of functions, since our input model requires that $y$ admits pointwise evaluation, we will abuse notation, letting $y$ denote the member of this class with $y(t) = \langle \varphi_t, \bv{Z}^*\bv{D}^{-1}\bv{c}\rangle_\mu = \langle \bv{D}^{-1}\bv{Z} \varphi_t,\bv{c} \rangle$, where $\varphi_{t}(\xi) = e^{-2\pi i t \xi}$.
%We assume that our algorithm is given an oracle for exactly applying $\Fmu$, $\bv{D}$, and $\bv{Z}$. This oracle can only increase its power, and so a lower bound under this assumption implies the theorem.

We prove that accurately reconstructing $y$ drawn from the hard input distribution yields an accurate reconstruction of the random vector $\bv{c}$. Since $\bv{c}$ is $n_{\mu,\ce\epsilon}$ dimensional, this reconstruction requires $\Omega(n_{\mu,\ce\epsilon})$ samples, giving us a lower bound for accurately  reconstructing $y$.

\begin{claim}\label{clm:smallNorm}
For random $x$ distributed as described above, with probability $\ge 5/6$,
$\norm{x}_\mu^2 \le \frac{1}{12\epsilon}$.
\end{claim}
\begin{proof}
$$\norm{x}_\mu^2 = \langle \bv{Z}^*\bv{\bar c}, \bv{Z}^*\bv{\bar c}\rangle_\mu = \langle \bv{\bar c}, \bv{Z}\bv{Z}^*\bv{\bar c}\rangle = \norm{\bv{\bar c}}_2^2.$$
We then bound $\norm{\bv{\bar c}}_2^2 \le \norm{\bv{c}}_2^2/ \lambda_{n_{\mu,\ce\epsilon}}(\Kmu) \le \frac{\norm{\bv c}_2^2}{\ce\epsilon}$ since $ \lambda_{n_{\mu,\ce\epsilon}}(\Kmu) \ge \ce\epsilon$ by definition. Finally, note that $\norm{\bv{c}}_2^2$ is a Chi-squared random variable, with $\E[\norm{\bv{c}}_2^2] = 1$. So loosely, by  Markov's inequality, with probability $\ge 5/6$, ${\norm{\bv{c}}_2^2} \le 6$, which gives the claim.
\end{proof}
%\begin{claim}\label{clm:largeNorm}
%$\norm{\Fmu^* x}_T^2 = 1$
%\end{claim}
%\begin{proof}
%$$\Fmu^* x = \Fmu^* \bv{Z}^* \bv{D} \bv{c}$$
%where $\bv{D} \in \RR^{ n_{\mu,2\epsilon},n_{\mu,2\epsilon}}$ is a diagonal matrix with $\bv{D}_{ii} = 1/\sqrt{\lambda_i(\Kmu)}$.
%We can check that $\Fmu^* \bv{Z}^* \bv{D}$ is  the operator whose columns are the top $n_{\mu,\ce\epsilon}$ eigenvectors of $\Kmu$. Since the columns of this operator are orthonormal and since $\bv{c}$ is a random unit vector, $\Fmu^* x$ is also a  unit vector.
%\end{proof}

From Claim \ref{clm:smallNorm} we have:
\begin{claim}\label{clm:acracyProb}
Given random input $y = \Fmu^* x$ generated as described above, with probability $\ge 5/6$,
to solve Problem \ref{prob:unformal_interp}, an algorithm must return a representation of $\tilde y$ with $\norm{y - \tilde y}_T^2 \le \frac{1}{\cef}$.
\end{claim}
\begin{proof}
Solving Problem \ref{prob:unformal_interp} requires finding a representation of $\tilde y$ with $\norm{y-\tilde y}_T^2 \le \epsilon \norm{x}_\mu^2 + C \norm{n}_T^2$. By Claim \ref{clm:smallNorm}  and the fact that for our input $\norm{n}_T^2 = 0$, with probability $\ge 5/6$ one has that $\epsilon \norm{x}_\mu^2 + C \norm{n}_T^2 \le \frac{1}{\cef}$, yielding the claim.
\end{proof}
We next show that finding a $\tilde y$ satisfying the condition of Claim \ref{clm:acracyProb} is at least as hard as finding an accurate approximation to $\bv{c}$.

\begin{claim}\label{clm:conversion} For $\tilde y$ with $\norm{y - \tilde y}_T^2 \le \frac{1}{\cef}$, $\bv{\tilde c} = \bv{U}^* \tilde  y$ satisfies $\norm{\bv{c}-\bv{\tilde c}}_2^2 \le \frac{1}{\cef}$.
%so that $sign(\bv{\tilde c}(i)) = sign(\bv{ c}(i))$ for at least $1/3$ of the entries.
\end{claim}
\begin{proof}
Recalling that $y =\bv{U} \bv{c}$, for $\bv{\tilde c} = \bv{U}^* \tilde  y$ we have:
\begin{align*}
\bv{\tilde c} = \bv{U}^* y  +\bv{U}^*( \tilde y - y) =\bv{U}^* \bv{U} \bv{c}  +\bv{U}^*( \tilde y- y).
\end{align*}
Recalling that $\bv{U}^* \bv{U} = \bv{I}$ we thus have:
\begin{align*}
\norm{\bv{c}-\bv{\tilde c}}_2^2 &= \norm{\bv{U}^*( \tilde y- y)}_2^2\\
&\le \norm{\tilde y- y}_T^2 \le \frac{1}{12}.
\end{align*}
The second to last inequality follows since $\bv{U}^* \bv{U} = \bv{I}$ and $\bv{U}\bv{U}^*$ are finite rank, so are compact and share the same non-zero eigenvalues. %Thus, due to the spectral theorem, their norm is bounded by the largest eigenvalue, which is $1$, which implies that 
Thus, $\bv{U}\bv{U}^* \preceq {\cal I}_T$~\cite[Lemma 8.26]{HunterNachtergaele:2001}.
This completes the claim.
%So, given $\tilde y$ with $\norm{y-\tilde y}_T^2 \le \frac{1}{\cef}$, it is possible to compute $\bv{\tilde c}$ with $\norm{\bv{c}-\bv{\tilde c}}_2^2 \le \frac{1}{\cef}$.
\end{proof}

Combining Claims \ref{clm:acracyProb} and \ref{clm:conversion} we have:
\begin{claim}\label{clm:overalsuccesss}
If a deterministic algorithm solves Problem \ref{prob:unformal_interp} with probability $\ge 2/3$ over our random input $y = \bv{U}\bv{c}$, then with probability $\ge 1/2$, letting $\tilde y$ be the output of the algorithm, $\bv{\tilde c} = \bv{U}^* \tilde y$ satisfies $\norm{\bv{c}-\bv{\tilde c}}_2^2 \le \frac{1}{\cef}$.
\end{claim}
\begin{proof}
If an algorithm solves Problem \ref{prob:unformal_interp} probability $\ge 2/3$ then by Claim \ref{clm:acracyProb}, it returns $\tilde y$ with $\norm{y - \tilde y}_T^2 \le \frac{1}{\cef}$ with probability $\ge 2/3-1/6 = 1/2$. Thus, by  Claim \ref{clm:conversion}, $\bv{\tilde c}$ satisfies $\norm{\bv{c}-\bv{\tilde c}}_2^2 \le \frac{1}{\cef}$ with probability $\ge 1/2$.
\end{proof}

%Additionally, we have:
%\begin{claim}\label{clm:overalsuccesss2}
%Under Assumption \ref{ass:assumption}, if a deterministic algorithm solves Problem \ref{prob:unformal_interp} with probability $\ge 2/3$ over our random input $y = \Fmu^* \bv{Z}^* \bv{D}^{-1} \bv{c}$, then there is a deterministic algorithm that, with probability $\ge 1/2$, outputs $\bv{\tilde c}$ satisfying $\norm{\bv{c}-\bv{\tilde c}}_2^2 \le \frac{1}{\cef}$.
%\end{claim}
%\begin{proof}
%We can first run the assumed algorithm to give query access to $\tilde y$ solving Problem \ref{prob:unformal_interp} with probability $2/3$ over the random input. By Claim \ref{clm:overalsuccesss}, with probability $\ge 1/2$, $\bv{\tilde c} = \bv{D}^{-1} \bv{Z} \Fmu \tilde y$ satisfies $\norm{\bv{c}-\bv{\tilde c}}_2^2 \le \frac{1}{\cef}$. Thus, it just remains to show that, given query access to $\tilde y$, we can deterministically compute $\bv{\tilde c}.$
%
%Let $\bv{U} = \Fmu^* \bv{Z}^* \bv{D}^{-1}$. As discussed, $\bv{D}^{-1} \bv{Z} \Fmu  \Fmu^* \bv{Z}^* \bv{D}^{-1} = \bv{U}^* \bv{U} = \bv{I}$. Thus, since $\bv{\tilde c} = \bv{U}^* \tilde y$, we have $\bv{U}^* \tilde y = \bv{U}^* \bv{U} \bv{\tilde c} = BLAH$.
%\end{proof}

Finally, we complete the proof of Theorem \ref{thm:initialLB} by  arguing that if $\tilde y$ is formed using $o(n_{\mu,\ce\epsilon})$ queries, then for $\bv{\tilde c}=\bv{U}^* \tilde y$, $\norm{\bv{c}-\bv{\tilde c}}_2^2 > \frac{1}{\cef}$ with good probability. Thus the bound of Claim \ref{clm:overalsuccesss} cannot hold and so $\tilde y$ cannot be a solution to Problem \ref{prob:unformal_interp} with good probability.

Assume for the sake of contradiction that there is a deterministic algorithm solving Problem \ref{prob:unformal_interp} with probability $\ge 2/3$ over the random input $\bv{U}\bv{c}$ that makes $r = \frac{n_{\mu,\ce\epsilon}}{20}$ queries on any input (note that if there exists an algorithm that makes fewer queries on some inputs, we can always modify it to make exactly $\frac{n_{\mu,\ce\epsilon}}{20}$ queries and return the same output.)

As discussed, each query to $y$ is a query to $y(t) =  \langle \bv{D}^{-1} \bv{Z} \varphi_t, \bv{c} \rangle$. Consider a deterministic function $Q$, that is given input $\bv{V} \in \CC^{i \times n_{\mu,\ce\epsilon}}$ (for any positive integer $i$) and outputs $Q(\bv{V}) \in \CC^{n_{\mu,\ce\epsilon} \times n_{\mu,\ce\epsilon}}$ such that $Q(\bv{V})$ has orthonormal rows with the first $i$ spanning the $i$ rows of $\bv{V}$. For example, $Q$ may run Gram-Schmidt orthogonalization on $\bv{V}$ fixing its first $\rank(\bv{V}) \le i$ rows and then fill out the remaining $ n_{\mu,\ce\epsilon} - \rank(\bv{V})$ rows using some canonical approach. Letting 
$\bv{D}^{-1} \bv{Z} \varphi_{t_1},\ldots,\bv{D}^{-1} \bv{Z} \varphi_{t_r}$ denote the queries made by our algorithm on random input $\bv{c}$, let $\bv{Q}^i = Q([\bv{D}^{-1} \bv{Z} \varphi_{t_1},\ldots,\bv{D}^{-1} \bv{Z} \varphi_{t_i}]^*)$. That is $\bv{Q}^i$ is an orthonormal matrix whose first $i$ rows span our first $i$ queries. Note that since our algorithm is deterministic, $\bv{Q}^i$ is a deterministic function of the random input $\bv{c}$. We have the following claim:

\begin{claim}\label{clm:stillGauss}
Conditioned on the queries $y(t_1),\ldots y(t_r)$, for $j > r$, each $[\bv{Q}^r \bv{c}](j)$ is distributed independently as $\mathcal{N}\left (0, \frac{1}{n_{\mu,\ce\epsilon}}\right )$.
\end{claim}
\begin{proof}
We prove the claim via induction on the number of queries considered. For the base case set $i = 1$. $\bv{Q}^1$ is a deterministic matrix (since the choice of our first query is made determinstically before seeing any input) and so by the rotational invariance of the Gaussian distribution, the entries of $\bv{Q}^1 \bv{c}$ are distributed independently as $\mathcal{N}\left (0, \frac{1}{n_{\mu,\ce\epsilon}}\right )$ (the same as the entries of $\bv{c}$). The first row of $\bv{Q}^1$ spans our first query, and thus this row is just equal to $\bv{D}^{-1} \bv{Z} \varphi_{t_1}$ scaled to have unit norm. Thus $y(t_1) = \bv{D}^{-1} \bv{Z} \varphi_{t_1} \bv{c}$ is just a fixed scaling of $[\bv{Q}^1 \bv{c}](1)$. So conditioning on $y(t_1) $, we still have [$\bv{Q}^1 \bv{c}](j)$ for $j > 1$ distributed independently as $\mathcal{N}\left (0, \frac{1}{n_{\mu,\ce\epsilon}}\right )$.

Now, consider $i > 1$. By the inductive assumption, conditioned on $y(t_1),\ldots y(t_{i-1})$, for $j \ge i$, $[\bv{Q}^{i-1} \bv{c}](j)$, are distributed independently as $\mathcal{N}\left (0, \frac{1}{n_{\mu,\ce\epsilon}}\right )$. We can see that both $\bv{Q}^{i-1}$ and $\bv{Q}^i$ are fixed conditioned on $y(t_1),\ldots y(t_{i-1})$ (since the $i^{th}$ query is chosen deterministically, possibly  adaptively as a function of the previously seen queries $y(t_1),\ldots y(t_{i-1})$). Additionally, since they share their first $i-1$ rows, the remaining $n_{\mu,\ce\epsilon} - i +1$ rows of $\bv{Q}^{i-1}$ and $\bv{Q}^i$ have the same rowspans. Thus we can write $\bv{Q}^{i} = [\bv{I};\bv{R}]\bv{Q}^{i-1} $ where $\bv{R} \in \CC^{n_{\mu,\ce\epsilon} - i +1 \times n_{\mu,\ce\epsilon} - i +1}$ is some fixed rotation with $\bv{R}^* \bv{R} = \bv{I}$. Thus, by the rotational  invariance of the Gaussian, for all $j \ge i$, $[\bv{Q}^{i}\bv{c}](j) $ are distributed independently  as $\mathcal{N}\left (0, \frac{1}{n_{\mu,\ce\epsilon}}\right )$ (the same as $[\bv{Q}^{i-1}\bv{c}](j) $). Further conditioning on $y(t_i)$, which is a deterministic function of $[\bv{Q}^i \bv{c}](i)$ and $y(t_1)\ldots y(t_{i-1})$, we still have that for $j > i$, $[\bv{Q}^{i}\bv{c}](j) $ are distributed independently  as $\mathcal{N}\left (0, \frac{1}{n_{\mu,\ce\epsilon}}\right )$.
This completes the inductive step and so the claim.
%We prove by induction on $r$.
% In the base case, if $r = 1$, the algorithm makes a single query $y(t_1)$. The query point $t_1$ must be chosen deterministically, and hence $\bv{Q}$ is deterministic. Thus by the rotational variance of the Gaussian distribution, each  
\end{proof}
%
%
%Let $\bv{u}_{t_i}$ denote the $i^{th}$ query  to $y$ made by the algorithm. Let $\bv{\bar u}_{t_i}$ be the unit vector in the direction of $\bv{u}_{t_i}$ orthogonal to all previous queries $\bv{u}_{t_1},\ldots,\bv{u}_{t_{i-1}}$ (or $\bv{\bar u}_{t_i} = \bv{0}$ if $\bv{u}_{t_i} $ is in the span of $\bv{u}_{t_1},\ldots,\bv{u}_{t_{i-1}}$). We can see that instead of returning $\langle \bv{u}_{t_i}, \bv{c} \rangle$ to the algorithm when it makes a query to $y(t_i)$, we can return $\langle  \bv{ \bar u}_{t_i}, \bv{c} \rangle$, since from this query and its previous queries, the algorithm can exactly reconstruct $\langle \bv{u}_{t_i}, \bv{c} \rangle$.
%
%Assume without loss of generality  that no queries are made for $\bv{u}_{t_i}$ falling within the span of all previous queries, since such queries provide the algorithm with no information.
%Let  $\bv{Q} \in \CC^{n_{\mu,\ce\epsilon} \times n_{\mu,\ce\epsilon}}$ be any matrix with orthogonal rows and columns whose first $k(\bv{c})$ rows are equal to $\bv{ \bar u}_{t_1},\ldots,\bv{ \bar u}_{t_{k(\bv{c})}}$.

%Since $\bv{c} \sim \mathcal{N}\left (\bv{0}, \bv{I} \cdot \frac{1}{n_{\mu,\ce\epsilon}}\right )$, by the rotational invariance of the Gaussian distribution, each entry of $\bv{Q} \bv{c}$ is distributed independently as $\mathcal{N}\left (0, \frac{1}{n_{\mu,\ce\epsilon}}\right )$.  Conditioned on the algorithm's queries $y(t_1),...,y(t_r)$, the last $n_{\mu,\ce\epsilon} - r$ entries of $\bv{Q} \bv{c}$ are thus distributed independently as $\mathcal{N}\left (0, \frac{1}{n_{\mu,\ce\epsilon}}\right )$.  We thus have:
Armed with Claim \ref{clm:stillGauss} we can compute:
\begin{align}\label{eq:le14}
\Pr\left [ \norm{\bv{c}-\bv{\tilde c}}_2^2  \le \frac{1}{\cef} \right ] &= \Pr\left [ \norm{\bv{Q}^r\bv{c}-\bv{Q}^r\bv{\tilde c}}_2^2\le \frac{1}{\cef} \right  ]\tag{Since $\bv{Q}^r$ is orthonormal.}\nonumber  \\
&\le \Pr\left [ \sum_{i=r+1}^{n_{\mu,\ce\epsilon}}\left |[\bv{Q}^r\bv{c}](i) - \bv{Q}^r \bv{\tilde c}(i)\right|^2  \le \frac{1}{\cef}  \right ]\nonumber\\
&= \E_{y(t_1),...,y(t_r)} \left [\Pr\left [ \sum_{i=r+1}^{n_{\mu,\ce\epsilon}}\left |[\bv{Q}^r\bv{c}](i) - \bv{Q}^r \bv{\tilde c}(i)\right|^2  \le \frac{1}{\cef} \mid  y(t_1)...,y(t_r)\right ] \right ]\nonumber\\
&\le  \E_{y(t_1),...,y(t_r)} \left [\Pr\left [ \sum_{i=r+1}^{n_{\mu,\ce\epsilon}}\left |[\bv{Q}^r\bv{c}](i)\right|^2  \le \frac{1}{\cef} \mid y(t_1)...,y(t_r)\right ] \right ]\label{eq:le14}
%&= \Pr\left [ \sum_{i=r+1}^{n_{\mu,\ce\epsilon}}\left |[\bv{Q}^r\bv{c}](i)\right|^2  \le \frac{1}{\cef}\right ]
\end{align}
where the last line follows since, conditioned on $y(t_1)...,y(t_r)$, $\bv{Q}^r \bv{\tilde c}$ is fixed and for $i \ge r+1$, $\bv{Q}^r \bv{c}(i)$ are distributed independently as Gaussians centered around $0$ (by Claim \ref{clm:stillGauss}). So the probability of the sum of differences being small is only smaller than if we replaced each $\bv{Q}^r \bv{\tilde c}(i) $ by $0$. 

%from the fact that, $\bv{Q}^r\bv{c}(i) \sim \mathcal{N}\left (0, \frac{1}{n_{\mu,\ce\epsilon}}\right )$ is independent of $y(t_1)...,y(t_r)$. The second to last line follows since, conditioned on $y(t_1)...,y(t_r)$, $\bv{Q}^r \bv{\tilde c}(i)$ is fixed. So the probability of the sum being small is only smaller than if we fixed $\bv{Q}^r \bv{\tilde c}(i) = 0$. 
Now, conditioned on $y(t_1)...,y(t_r)$,
$\sum_{i=r+1}^{n_{\mu,\ce\epsilon}}\left |[\bv{Q}^r\bv{c}](i)\right|^2$ is a Chi-squared random variable with
$$\E \left [ \sum_{i=r+1}^{n_{\mu,\ce\epsilon}}\left |[\bv{Q}^r\bv{c}](i)\right|^2  \mid y(t_1)...,y(t_r)\right ] =  \frac{n_{\mu,\ce\epsilon}-r}{n_{\mu,\ce\epsilon}}.$$
For $r = \frac{n_{\mu,\ce\epsilon}}{20}$, we thus have $\E \left [ \sum_{i=r+1}^{n_{\mu,\ce\epsilon}}\left |[\bv{Q}^r\bv{c}](i)\right|^2  \mid y(t_1)...,y(t_r)\right ] \ge \frac{19}{20}$. We can loosely upper bound the probability in \eqref{eq:le14}, using that for a Chi-squared random variable $X$ with $k$ degrees of freedom, $\Pr[X \le \delta \E[X] ] \le (\delta e^{1-\delta})^{k/2} \le (\delta e^{1-\delta})^{1/2}$. So,
\begin{align*}
\Pr\left [ \sum_{i=k(\bv{c})+1}^{n_{\mu,\ce\epsilon}}\left |[\bv{Q}^r\bv{c}](i)\right|^2  \le \frac{1}{\cef}\mid y(t_1)...,y(t_r) \right ] \le \left( \frac{20}{19\cdot \cef} e^{1-\frac{20}{19\cdot \cef}} \right  )^{1/2} < \frac{47}{100}.
\end{align*}
Plugging back into \eqref{eq:le14} gives:
\begin{align*}
\Pr\left [ \norm{\bv{c}-\bv{\tilde c}}_2^2  \le  \frac{1}{\cef} \right ] \le  \E_{y(t_1),...,y(t_r)} \left [\Pr\left [ \sum_{i=r+1}^{n_{\mu,\ce\epsilon}}\left |[\bv{Q}^r\bv{c}](i)\right|^2  \le \frac{1}{\cef} \mid y(t_1)...,y(t_r)\right ] \right ] < \frac{47}{100}.
\end{align*}
However, we have assumed that our algorithm  solves Problem \ref{prob:unformal_interp} with probability $\ge 2/3$, and hence, by Claim \ref{clm:overalsuccesss}, $\Pr\left [ \norm{\bv{c}-\bv{\tilde c}}_2^2  \le \frac{1}{\cef} \right ] \ge \frac{1}{2}$. This is a contradiction, yielding the theorem.

\end{proof}

\subsection{Statistical Dimension Lower Bound}
We now use Theorem \ref{thm:initialLB} to prove that the statistical dimension tightly  characterizes the sample complexity  of solving Problem \ref{prob:unformal_interp}  for any constraint measure $\mu$ satisfying a simple condition: we must have $\smu = O(1/\epsilon^p)$ for some $p < 1$. Note that this assumption holds for all $\mu$ considered in this work (including bandlimited, multiband, sparse, Gaussian, and Cauchy-Lorentz), where $\smu$ either grows as $\log(1/\epsilon)$ or $1/\sqrt{\epsilon}$. Also note that by \eqref{def:stat_dim_sum_version} we can always bound $\smu \le \tr(\Kmu)/\epsilon = 1/\epsilon$. So this assumption holds whenever we have a nontrivial upper bound on $\smu $.

\begin{theorem}[Statistical Dimension Lower Bound]\label{thm:mainLB}
For any probability measure $\mu$, suppose that $\smu = O(1/\epsilon^p)$ for some constant $p < 1$.
Consider any (possibly  randomized) algorithm that solves Problem \ref{prob:unformal_interp} with probability  $\ge 2/3$ for any  function $y$ and any $\epsilon > 0$ and makes at most $r_{\mu,\epsilon}$ (possibly adaptive) queries on any input. Then $r_{\mu,\epsilon} = \Omega(\smu)$.\footnote{Here we follow the Hardy-Littlewood definition \cite{hardy1914some}, using $f(\epsilon) = \Omega(g(\epsilon))$ to denote that $\limsup_{x \to \infty} \frac{f(\epsilon)}{g(\epsilon)} > 0$. Thus the lower bound shows that, for some fixed constant $c > 0$, for every $\epsilon$, there is at least some $\epsilon' < \epsilon$ where the number of queries used by any algorithm solving Problem \ref{prob:unformal_interp} with probability $\ge 2/3$ is at least $c \cdot \smu$. In other words, the lower bound rules out the possibility that the number of queries is $o(\smu)$.}
\end{theorem}
\begin{proof}
We simply prove that for this class of measures, $n_{\mu,\ce\epsilon} = \Omega(\smu)$ and then apply Theorem \ref{thm:initialLB}. It suffices to show that  $n_{\mu,\epsilon} = \Omega(s_{\mu,c\epsilon})$ for any fixed constant $c \ge 1$ since by \eqref{constAdjust}, $s_{\mu,c\epsilon} \ge \frac{\smu}{c}$.
Thus $n_{\mu,\epsilon} = \Omega(s_{\mu,c\epsilon})$ gives that  $n_{\mu,\ce\epsilon} = \Omega(s_{\mu,\ce c\epsilon}) = \Omega(\smu)$, giving the theorem.

%Assume without loss of generality that $p \ge 1/2$ and
Let $\ceu = 2^{\frac{4}{1-p}} > 1$. Assume for the sake of contradiction that $n_{\mu,\epsilon} = o(s_{\mu,c_p \epsilon})$. By this assumption, there is some fixed $\epsilon_0$ such that,
\begin{align}\label{eq:eps0}
\text{For all }\epsilon \le \epsilon_0\,, n_{\mu,\epsilon} \le \frac{s_{\mu,\ceu\epsilon}}{2}.
\end{align}
We can bound:
\begin{align*}
 s_{\mu,\ceu\epsilon} = \sum_{i=1}^\infty \frac{\lambda_i(\Kmu)}{\lambda_i(\Kmu)+\ceu\epsilon} \le n_{\mu,\epsilon} +  \sum_{i=n_{\mu,\epsilon}+1}^\infty \frac{\lambda_i(\Kmu)}{\ceu\epsilon}
\end{align*}
and thus by \eqref{eq:eps0} have for any $\epsilon \le \epsilon_0$:
\begin{align}\label{eq:smuLowerBound}
\frac{1}{2} \cdot  s_{\mu,\ceu\epsilon} & \le \sum_{i=n_{\mu,\epsilon}+1}^\infty \frac{\lambda_i(\Kmu)}{\ceu\epsilon}.
\end{align}
Now we also have:
\begin{align*}
\smu = \sum_{i=1}^\infty \frac{\lambda_i(\Kmu)}{\lambda_i(\Kmu)+\epsilon}&\ge  \sum_{i=n_{\mu,\epsilon}+1}^\infty \frac{\lambda_i(\Kmu)}{\lambda_i(\Kmu)+ \epsilon}\\
&\ge \sum_{i=n_{\mu,\epsilon}+1}^\infty \frac{\lambda_i(\Kmu)}{2\epsilon}\\
& = \frac{\ceu}{2} \cdot  \sum_{i=n_{\mu,\epsilon}+1}^\infty \frac{\lambda_i(\Kmu)}{\ceu \epsilon}.
\end{align*}
Combined with \eqref{eq:smuLowerBound} this gives that for any $\epsilon \le \epsilon_0$:
\begin{align}\label{eq:growthBound}
\smu  \ge \frac{\ceu}{4} \cdot s_{\mu,\ceu \epsilon}.
\end{align}
By \eqref{eq:growthBound} we in turn have that, for every $\epsilon \le \epsilon_0$,
\begin{align*}
\smu \ge s_{\mu,\epsilon_0} \cdot \left (\frac{\ceu}{4} \right)^{\lfloor \log_{\ceu} \epsilon_0/\epsilon \rfloor}.
\end{align*}
Using that $\lfloor \log_{\ceu} \epsilon_0/\epsilon \rfloor \ge \log_{\ceu} \epsilon_0/\epsilon-1$ and that $\ceu = 2^{\frac{4}{1-p}} \ge 16$ we can then bound, for all $\epsilon \le \epsilon_0$:
\begin{align*}
\smu \ge  \left (\frac{ \ceu}{4} \right )^{\log_{\ceu} \epsilon_0 - \log_{\ceu} \epsilon - 1} &= \left (\frac{ \ceu}{4} \right )^{\log_{\ceu} \epsilon_0 - 1} \cdot \ceu^{\log_{\ceu} 1/\epsilon} \cdot \left ( \frac{1}{4} \right )^{\log_{\ceu} 1/\epsilon}\\
&\ge \left (\frac{ \ceu}{4} \right )^{\log_{\ceu} \epsilon_0 - 1}  \cdot \frac{1}{\epsilon} \cdot \epsilon^{\frac{1-p}{2}}\\
&\ge \left (\frac{ \ceu}{4} \right )^{\log_{\ceu} \epsilon_0 - 1}  \cdot \frac{1}{\epsilon^{p + \frac{1-p}{2}}} .
\end{align*}
Note that $\left (\frac{ \ceu}{4} \right )^{\log_{\ceu} \epsilon_0 - 1}$ is a constant independent of $\epsilon$. Thus, the above contradicts the assumption that $\smu = O(1/\epsilon^p)$, giving the theorem.
\end{proof}

\spara{Remark}
We remark that a similar technique to Theorem \ref{thm:mainLB} can be used to show that $n_{\mu,\epsilon} = \Omega(\smu/\epsilon^p)$ for any $p > 0$, without any  assumptions on $\smu$. %Combined with Theorem \ref{thm:mainLB} this shows that $\Omega(\smu^{1-p})$ samples are required for any $p > 0$, showing that $\smu$ nearly optimally characterizes the sample complexity of solving Problem \ref{prob:unformal_interp}.

%\todo{
%It feels like we are doing something wrong and that the bound should hold in general with just an $\smu = o(1/\epsilon)$ assumption rather than requiring $\smu = o(1/\epsilon^p)$ for some $p < 1$.
%One thing to note is that we're not using the eigenvalue falloff explicitly  at all here. E.g. the fact that since $\tr(\Kmu) = 1$ we must have $\lambda_i(\Kmu) \le \frac{1}{i}$. If this bound is saturated, then setting $\epsilon = \lambda_i(\Kmu)$ we actually have $n_{\mu,\epsilon}  = i$ and $\smu \le i + \sum_{j=i+1}^{\infty} \frac{\lambda_j(\Kmu)}{\lambda_i(\Kmu)} \le i + \sum_{j=i+1}^{\infty} \frac{\lambda_j(\Kmu)}{1/i} \le 2i$ so we actually have the result. So weirdly enough we just need to deal with the case where the eigenvalues fall off faster.
%}

\section{Conclusion and Open Problems}\label{sec:conclusion}

We view our work as the starting point for further exploring the application of techniques from the randomized numerical linear algebra literature (such as leverage score sampling, column based matrix reconstruction, and random projection) in signal processing. We lay out a number of open directions that we consider interesting below:

\begin{itemize}
\item The most immediate question is to generalize our results for interpolation over an interval to higher dimensional spaces. Fourier constrained interpolation in two or three dimensions is important in many areas, such as the earth and geosciences \cite{Ripley:1989,Ripley:2005} and image processing \cite{Pesquet-PopescuVehel:2002,RamaniVilleUnser:2006}. Interpolation in even higher dimensions is common in Gaussian process methods in machine learning. We believe that our techniques should extend to higher dimensions in a similar manner to prior related work on kernel approximation \cite{AvronKapralovMusco:2017}.
\item We have considered a simple signal reconstruction problem, where we wish to reconstruct  a function over a fixed interval given sample access at points in that interval. There are many interesting variations of this problem. For example, can better  bounds be achieved if samples can be taken from the interval $[0,T]$, but  we only  consider reconstruction error over a subset  of this interval? In this setting, can uniform sampling give optimal bounds? How can one formulate a similar reconstruction problem and adapt our techniques to the streaming setting, where we hope to estimate a signal at any  given point in time using measurements at past samples (and perhaps must limit memory/computation at  any given time)? Can our techniques be extended to the setting where the error is averaged using a non-uniform measure in time domain? This question is especially relevant  for applications to machine learning, where we may wish to approximate the signal well on average on input points drawn from some non-uniform distribution. In traditional supervised learning, reconstruction would be performed using points drawn from this same distribution.  However, in an \emph{active learning} setting, we may be allowed to drawn points from some other distribution, such as the leverage score distribution, which yields better error.  
\item In our work we have assumed knowledge of the constraint $\mu$. However, as discussed, in the case of sparse and multiband signal reconstruction, it is important to learn $\mu$ (i.e., the locations of the frequencies or frequency bands) as part of the reconstruction process.  Understanding how to do this, perhaps by combining existing techniques \cite{MishaliEldar:2009,Moitra:2015,PriceSong:2015,ChenKanePrice:2016} with our own is an important direction. More generally, in many applications, $\mu$ is derived from the signal itself, by estimating the signal's autocorrelation, which corresponds to our kernel function $k_\mu$. Can our techniques be used to give bounds in this setting?
\ifdraft\Haim{An interesting case might be to know/assume that $\mu$ comes from a family of distributions, and learn the hyperparameters. For examples, in Gaussian Process Regression they assume that the covariance function has certain hyperparameters, and optimize them using maximum likelihood.}\fi
\item Can our techniques be extended to learning signals giving constraints on other transforms such as the short-time Fourier transform (the signal's spectrogram), the wavelet transform, etc.? More generally, can leverage  score sampling  be used to approximate these transforms and to approximately  apply filters or other signal modifications based on them?
\item What is the connection between our randomized leverage score sampling method and deterministic `sampling' methods such as Chebyshev interpolation for low-degree polynomials, uniform sampling for bandlimited signal reconstruction,  and non-uniform ``multicoset'' sampling schemes considered in the signal processing literature \cite{FengBresler:1996,VenkataramaniBresler:2000,Bresler:2008,MishaliEldar:2009}. Can our results be made deterministic, perhaps using deterministic sampling methods for operator approximation like those employed in our proof of Lemma \ref{lem:cssSpectral}?
\end{itemize}

\subsection*{Acknowledgements}\label{sec:ack}

We thank Ron Levie for helpful discussions on weak integrals in Hilbert spaces and Zhao Song for discussions on smoothness bounds for sparse Fourier functions. We also thank Yonina Eldar for helpful discussions and pointers to related work.  

Haim Avron's work is supported in part by Israel Science Foundation (grant no. 1272/17) and United States-Israel Binational Science Foundation (grant no. 2017698). Michael Kapralov is supported in part by ERC Starting Grant 759471.
 
\bibliographystyle{alpha}
\bibliography{sinc}	

\begin{thebibliography}{BWvOGD01}
\bibitem[Whi15]{Whittaker:1915}
Edmund~T. Whittaker.
\newblock On the functions which are represented by the expansions of the
  interpolation theory.
\newblock {\em Proceedings of the Royal Society of Edinburgh}, 35:181--194,
  1915.
\end{thebibliography}

\pagebreak
\appendix

\section{Prior work on Fourier constrained interpolation}
\label{app:prior_work}

As mentioned in Section \ref{sec:intro}, constrained interpolation problems similar to Problem \ref{prob:unformal_interp} have been studied for decades in a number of different communities, often with widely varying computational models, assumptions, and goals. We discuss the most relevant prior work here.

\spara{Bandlimited functions.} The most well studied special case of Problem \ref{prob:unformal_interp} is when $\mu$ is uniform on an interval $[-F,F]$, which corresponds to reconstructing a bandlimited function from discrete samples. Work on this problem goes back to famous results of Nyquist, Shannon, and others \cite{Whittaker:1915,Nyquist:1928,Kotelnikov:1933,Shannon:1949}, who showed that it suffices to sample time uniformly with frequency $O(1/F)$. While this rate roughly suggests that $O(FT)$ samples should be required to reconstruct a signal on $[0,T]$, this does not follow directly: common reconstruction methods like Whittaker-Shannon sinc intetrpolation infer $y(t)$ from an \emph{infinite} sum of past and future samples from $y$. It is possible to perform approximate reconstruction by truncating this sum, however the number of samples required to gives error $\epsilon$ will be large: $\Omega(1/\epsilon)$. See Example \ref{shannonLB} at the end of the section.

Progress on the finite time reconstruction question beyond truncated Whittaker-Shannon began with the pioneering work Slepain, Landau, and Pollak,  who study the operator $\Kmu$ for uniform, bandlimited measures $\mu$ \cite{SlepianPollak:1961,LandauPollak:1961,LandauPollak:1962}. They bound the number of eigenvalues of $\Kmu$ above $\epsilon$, a quantity that is at most a constant factor larger than our $\smu$.
Using this bound, it is possible to argue that Problem \ref{prob:unformal_interp} can be solved via regression onto at most $O(\smu)$ \emph{prolate spheroidal wave functions (PSWFs)}.

While the prolate spheroidal wave functions cannot be  explicitly represented and used directly in a regression algorithm, later work presents practical methods for working with them using quadrature rules and a finite number of time samples \cite{XiaoRokhlinYarvin:2001,ShkolniskyTygertRokhlin:2006, OsipovRokhlin:2014}.
For the noiseless version of Problem \ref{prob:unformal_interp}, that work, combined with the statistical dimension bound of Landau and Widom \cite{LandauWidom:1980}, yields algorithms that take roughly $\tilde{O}\left(FT + \log(1/\epsilon)\right)$ samples and $\tilde{O}\left(\left(FT + \log(1/\epsilon)\right)^\omega\right)$ time, matching our results up to log factors.\footnote{Landau and Widom's bounds can be used to show that $\smu = \tilde \Theta(FT + \log(1/\epsilon))$, however their result only holds asymptotically as $FT$ goes to infinity. Ours holds for all values of $F,T,\epsilon$ -- see Theorem \ref{thm:bandlimited_leverage_scores}.}

We note that existing quadrature methods access the function $f(t)$ at a pre-determined set of time domain points. Thus, they are inherently not robust to noise, since the noise function $n(t)$ of Problem \ref{prob:unformal_interp} can place arbitrarily bad corruptions on $\Fmu x$ at the pre-determined sample points. To the best of our knowledge our work is the first to solve Problem \ref{prob:unformal_interp} for bandlimited signals in the adversarial noise setting.

\spara{Sparse functions.}
Signal interpolation has also been studied extensively when $y$ is assumed to have a sparse Fourier transform: this is the basic problem of compressed sensing and sparse recovery. While most results in the compressed sensing literature are for {discrete functions} and address sparsity in the \emph{discrete Fourier transform}, there has been interest in extending that work to the continuous case \cite{DavenportWakin:2012}. Furthermore, there are a number of results on the continuous problem that predate compressed sensing: reader's may be familiar with Prony's method \cite{Prony:1795}, Pisarenko's method \cite{Pisarenko:1973}, the matrix pencil method \cite{BreslerMacovski:1986}, or the MUSIC algorithm \cite{Schmidt:1986}.

While these methods do not provide direct guarantees for Problem \ref{prob:unformal_interp}, recently Chen, Kane, Price, and Song study a formulation of the sparse signal interpolation problem that closely matches our formulation \cite{ChenKanePrice:2016}. Follow-up work in \cite{ChenPrice:2018} achieves a sample complexity of $\tilde{O}(k\log^2 k)$, exactly matching our bounds. In fact, our proofs for general constraint measures rely directly on two essentially lemmas on the smoothness of Fourier-sparse functions from \cite{ChenKanePrice:2016} and \cite{ChenPrice:2018}.

We note that most compressed sensing type results, including those of \cite{ChenKanePrice:2016,ChenPrice:2018}, are distinguished from our work in that they can also learn the support of $\mu$ -- our methods assume that this support is known \emph{a priori}. We believe that our methods can be combined with existing techniques for learning the Fourier support and view this is an interesting open direction.

\spara{Multiband functions.}
Due to applications in radio, radar, medical imaging, and many other areas, there has  been substantial interest in sample efficient algorithms for reconstructing multiband functions \cite{Eldar:2015}. Landau \cite{Landau:1967a,Landau:1967} was the first to characterize the sample complexity of reconstructing such functions in the sense of the Nyquist-Shannon sampling theorem, showing that to recover a signal with $s$ frequency bands of widths $F_1, \ldots, F_s$, the average sampling rate must be at least $1/\sum_i F_i$. Unlike bandlimited interpolation, it is not obvious how to construct sampling schemes that achieve this optimal rate, and doing so has been the subject of a rich line of work on non-uniform ``multicoset'' sampling schemes \cite{FengBresler:1996,VenkataramaniBresler:2000,Bresler:2008,MishaliEldar:2009}.

As in the bandlimited setting, the rate of the infinite time-horizon problem suggests, but does not imply, that the finite-time problem can be solved with roughly $\sum_i F_iT$ samples. Via a direct analysis of $\Kmu$, results on prolate spheroidal wave functions can be used to upper bound the Fourier statistical dimension for a multiband support by $O\left(\sum_i F_iT + s\log(1/\epsilon)\right)$ \cite{LandauWidom:1980}, matching our bound in Theorem \ref{theorem-multiband-statdim}. However, we are unaware of existing work that solves Problem \ref{prob:unformal_interp} with a number of samples matching this statistical dimension bound. We suspect that, as in the bandlimited case, in the noiseless setting, our methods could be matched via a combination of numerical quadrature and PSWF regression. There has been some initial work in that direction \cite{LakeyHogan:2012}.

\spara{General constraints.}
Beyond the three standard settings discussed above, there has been an effort to understand the complexity of approximately reconstructing functions with more general Fourier transform constraints. In the discrete setting, \emph{model based compressed sensing} has proven to be a powerful framework \cite{BaraniukCevherDuarte:2010,HegdeIndykSchmidt:2015}. Similar ideas have been extended to continuous functions \cite{EldarMishali:2009,Eldar:2009}. However, the constraints considered in prior work do not correspond with those captured by Problem \ref{prob:unformal_interp}. We are interested in a more refined understanding of how sample complexity depends on a the complexity of a function's representation in the Fourier basis. Model based compressed sensing focuses on functions that can be sparsely represented in other bases.

\spara{Leverage score sampling.} Finally, we note that,
beyond widespread applications in randomized numerical linear algebra, there has been prior work studying leverage score sampling  schemes for discretizing continuous operators, which is the approach we take to solving Problem \ref{prob:unformal_interp}. See for example \cite{CohenMigliorati:2017,Bach:2017,pauwels2018relating}. Our main contribution is demonstrating how to actually upper bound the leverage scores for operators of interest, which is the missing ingredient that typically prevents such sampling results from being algorithmic. We think that the tools presented in this paper offer a powerful approach to discretization in general, with significant potential for future research. We use similar methods in our recent work on randomized approximation schemes for Gaussian kernel matrices \cite{AvronKapralovMusco:2017}.

\begin{example}\label{shannonLB}
Truncated Whittaker-Shannon interpolation requires $\Omega(1/\epsilon)$ samples to approximate $y(t)$ with bandlimit $F = 1/2$ on $[0,1]$ up to error $\epsilon$ (i.e., to solve Problem \ref{prob:unformal_interp}, outputting $\tilde y(t)$ with $\norm{y-\tilde y}_{[0,1]}^2 \le \epsilon \norm{\hat y}_{\mu}^2$ where $\mu$ is the uniform probability measure on $[-1/2,1/2]$ and $\hat y$ is the Fourier transform of $y$ with $y = \Fmu^* \hat y$.)
\end{example}
\begin{proof}
Let $\mathcal{E}$ be the set of even integers in $[\lfloor 1/2\epsilon \rfloor,\lfloor 1/\epsilon \rfloor]$. Note that $|\mathcal E| = \Theta(1/\epsilon)$.
Let $y$ be a sum of $\Theta(1/\epsilon)$ standard sinc functions centered at the points in $\mathcal E$: 
\begin{align*}
y(t) =  \sum_{k\in \mathcal E} y_k(t)\text{ where } y_k(t) = \frac{\sin(\pi\cdot (t-  k))}{\pi \cdot (t-  k)}.
\end{align*}
The Fourier transform $\hat y_k(\xi)$ is the box on $[-1/2,1/2]$ multiplied by $e^{-2\pi i k}$. Thus the Fourier transform $\hat y(\xi) = \sum_{k\in \mathcal E} \hat y_k(\xi)$ is supported on $[-1/2,1/2]$ and so the Nyquist rate is $1$ and Whittaker-Shannon interpolation reconstructs $y(t)$ as a sum of sinc functions centered at the integers: 
$$y(t) = \sum_{k=-\infty}^\infty y(k) \cdot  \frac{\sin(\pi\cdot (t-  k))}{\pi \cdot (t- k)}.$$ We can see that this reconstruction is exact since $y(k) = 0$ for all integer $k$ except $y(k) =1$ for $k \in \mathcal{E}$. However, if we approximate $y(t)$ on the range $[0,1]$ by truncating the Whittaker-Shannon sum to $\le \lfloor 1/\epsilon\rfloor$ samples centered at $0$, we will not include the terms corresponding to  $k \in [\lfloor 1/2\epsilon\rfloor,\lfloor 1/\epsilon\rfloor] \supseteq \mathcal{E}$ and so will have $\tilde y(t) = 0$ and so $\norm{y - \tilde y}_{[0,1]}^2 = \norm{y}_{[0,1]}^2$. Since $\mathcal{E}$ is the set of even integers in $[\lfloor 1/2\epsilon\rfloor,\lfloor 1/\epsilon\rfloor]$:
\begin{align}\label{eq:errorCam}
\norm{y}_{[0,1]}^2 = \int_0^1 y(t)^2 \, dt &= \int_{0}^{1} \left (\sum_{k \in \mathcal{E}} \frac{\sin(\pi\cdot (t-  k))}{\pi \cdot (t-  k)} \right )^2 \, dt\nonumber\\
&= \Omega \left (\epsilon^2 \right) \int_{0}^{1} \left (\sum_{k \in \mathcal{E}} \sin(\pi\cdot t)\right )^2 \, dt = \Omega \left ( 1 \right ).
\end{align}
Finally we note that the for $j \neq k$ $\langle \hat y_k,\hat y_j \rangle_\mu = \int_{-1/2}^{1/2} e^{-2\pi i (j-k)\xi}d\xi = 0$. Thus $$\norm{\hat y}_\mu^2 = \sum_{k \in \mathcal E} \norm{\hat y_k}_\mu^2 = \Theta(1/\epsilon).$$
Combined with \eqref{eq:errorCam} this gives:
\begin{align*}
\norm{y - \tilde y}_{[0,1]}^2 = \norm{y}_{[0,1]}^2 = \Omega \left ( 1 \right ) = \Omega(\epsilon) \cdot \norm{\hat y}_\mu^2 
\end{align*}
which completes the lower bound.
\end{proof}
%To see this, consider the sum of a standard sinc function and one shifted to be centered at $\lfloor 1/\epsilon \rfloor$: 
%$$y(t) = y_1(t) + y_2(t)\text{ where } y_1(t) = \frac{\sin(\pi*t)}{\pi \cdot t} \text{ and } y_2(t)=\frac{\sin(\pi*(t-  \lfloor 1/\epsilon \rfloor))}{\pi \cdot (t-  \lfloor 1/\epsilon \rfloor)}.$$
%We can see that the Fourier transform $\hat y(t)$ is supported on $[-1/2,1/2]$. Thus the Nyquist rate is $1$ and Whittaker-Shannon interpolation reconstructs $y(t)$ as a sum of sinc functions centered at the integers: $y(t) = \sum_{k=\infty}^\infty y(k) \cdot  \frac{\sin(\pi*(t-  k))}{\pi \cdot (t- k)}$. We can see that this reconstruction is exact since $y(k) = 0$ for all integer $k$ except $y(0) = y(\lfloor 1/\epsilon \rfloor)=1$. However, if we approximate $y(t)$ on the range $[0,1]$ by truncating the Whittaker-Shannon sum to $o(1/\epsilon)$ samples centered at $0$, we will not include the term corresponding to the sinc centered at $\lfloor 1/\epsilon \rfloor$ and so will have $\tilde y(t) = y_1(t)$. Since $|y_2(t)| = \Theta(\epsilon)$ on this range on average and since $y_1(t) = \Theta(1)$ on average, we can see that $$\norm{y(t) - \tilde y(t)}_{[0,1]}^2= \Omega(\epsilon).$$

\section{Operator theory preliminaries}\label{app:op}

Throughout the paper, we use the term {\em operator} for linear transformation between two Hilbert spaces. In this section we discuss and prove basic results on operators that we use throughout the paper.

\subsection{Basic definitions and the Loewner partial ordering}
Consider two Hilbert spaces ${\cal H}_1$ and ${\cal H}_2$ with inner products $\langle \cdot ,\cdot \rangle_{\mathcal{H}_1}$ and $\langle \cdot ,\cdot \rangle_{\mathcal{H}_2}$.
We denote by ${\mathbb B}({\cal H}_1, {\cal H}_2)$ the set of bounded operators from ${\cal H}_1$ to ${\cal H}_2$, and abbreviate ${\mathbb B}({\cal H})$ if ${\cal H}_1 = {\cal H}_2 = {\cal H}$. We denote by $\BTC({\cal H})$ and $\BHS({\cal H})$ the set of trace-class and Hilbert-Schmidt operators (respectively) on ${\cal H}$ (i.e. from ${\cal H}$ to ${\cal H}$). Note that $\BTC({\cal H})\subset \BHS({\cal H})\subset {\mathbb B}({\cal H})$. Recall that for operators,  boundedness is equivalent to continuity. The open mapping theorem states that if ${\cal A}$ is invertible, then ${\cal A}^{-1}$ is bounded. This implies that a compact operator is not invertible. If ${\cal A} \in {\mathbb B}({\cal H})$ and ${\cal B} \in \BTC({\cal H})$ then ${\cal A}{\cal B}, {\cal B}{\cal A} \in \BTC({\cal H})$ and $\tr({\cal A}{\cal B}) = \tr({\cal B}{\cal A})$.

We call self-adjoint ${\cal A}$ {\em positive semidefinite} (or simply \emph{positive}) and write $\mathcal{A} \succeq 0$ if $\langle x, {\cal A} x \rangle_{\mathcal{H}} \geq 0$ for all $x\in{\cal H}$. We write ${\cal A} \succsim 0$ if ${\cal A}$ is {\em positive definite}, i.e. $\langle x, {\cal A} x \rangle > 0$ for all $x\in{\cal H}$. We denote ${\cal A} \succ 0$ is ${\cal A}$ if ${\cal A}$ is {\em strictly positive}, i.e.  there exist a $c > 0$ such that ${\cal A} \succsim c \cdot {\cal I}_{\mathcal{H}}$ where ${\cal I}_{\mathcal{H}}$ is the identity operator on $\mathcal{H}$. Note that for operators on finite dimensional Hilbert spaces, ${\cal A} \succsim 0$ if and only if ${\cal A} \succ 0$, but this is not always the case for infinite dimensional Hilbert spaces.  The notation for ${\cal A} \succeq {\cal B}$, ${\cal A} \succsim {\cal B}$, and ${\cal A} \succ {\cal B}$ follow in the standard way.

If ${\cal A} \succeq 0$ is self-adjoint and bounded, then it possesses a unique self-adjoint bounded square root ${\cal A}^{1/2} \succeq 0$~\cite{Wouk66}. Furthermore, if ${\cal A}$ is strictly positive then so is ${\cal A}^{1/2}$. This implies that if ${\cal A}$ is strictly positive and bounded, then ${\cal A}$ is bounded below and that the inverse of of the square root of ${\cal A}$ is ${\cal A}^{-1/2} \eqdef ({\cal A}^{-1})^{1/2}$. Lidskii's theorem states the that trace of a trace-class operator is the sum of its eigenvalues.
\ifdraft
\Haim{1) if ${\cal A}$ is bounded, then ${\cal A}^{1/2}$ is bounded because of Cauchy-Schwartz 2) Strict positivity + bounded + self-adjoint implies invertible (too basic to include, but keeping it as comments in case we wonder later): Since ${\cal A}$ bounded (continuous) and and bounded below, then it is one-on-one and has a closed range~\cite[Theorem 2.5]{AA02}. Thus, ${\cal H} = \range({\cal A}) \oplus \ker({\cal A}^*)$. But ${\cal A}$ is self-adjoint so $\ker({\cal A}^*) = \ker({\cal A}) = {0}$ where the last equality is due to the fact that {\cal A} is positive definite. So, $\range({\cal A}) = {\cal H}$ and ${\cal A}$ is onto as well as one-on-one.}
\fi

Many of the following claims are well known of matrices, and the proofs in most cases, but not all, mirror the matrix case. However, for the operator case we need to be more careful with the conditions due to the aforementioned distinction between $\succsim$ and $\succ$.

\begin{claim}
	Suppose that ${\cal A}$ is a self-adjoint bounded positive semidefinite operator on an Hilbert space ${\cal H}$. For every $\epsilon > 0$, the operator ${\cal A} + \epsilon {\cal I}_{\cal H}$ is bounded, strictly positive and invertible, and the inverse is bounded.
\end{claim}
\begin{proof}
	The operator ${\cal A} + \epsilon {\cal I}_{\cal H}$ is the sum of two bounded operators, and so it is bounded. It is also clearly bounded below, since ${\cal A} + \epsilon {\cal I}_{\cal H} \succeq \epsilon {\cal I}_{\cal H} \succ 0$. A continuous (i.e., bounded) bounded-below operator is invertible, so ${\cal A} + \epsilon {\cal I}_{\cal H}$ is invertible. The inverse is bounded due to the open mapping theorem.
\end{proof}

\begin{claim}
	\label{claim:bound-to-inverse-bound}
	Suppose that $0 \prec {\cal A} \preceq {\cal I}_{\cal H}$ for a self-adjoint operator ${\cal A}$. Then, ${\cal A}^{-1} \succeq {\cal I}_{\cal H}$.
\end{claim}
\begin{proof}
	For every $x\in{\cal H}$ we have $\langle x, {\cal A} x \rangle_\mathcal{H} \leq \langle x, x \rangle_\mathcal{H}$. Given $y$, let $x={\cal A}^{-1/2} y$. Then $\langle y, y \rangle_\mathcal{H} = \langle {\cal A}^{1/2} x, {\cal A}^{1/2} x\rangle_\mathcal{H} = \langle x, {\cal A} x \rangle_\mathcal{H} \leq \langle x, x \rangle_\mathcal{H} =  \langle {\cal A}^{-1/2} y, {\cal A}^{-1/2} y\rangle_\mathcal{H}  = \langle y, {\cal A}^{-1} y\rangle_\mathcal{H}$ so ${\cal A}^{-1} \succeq {\cal I}_{\cal H}$.
\end{proof}

\begin{claim}
	\label{clm:trace-sqrt}
	Suppose that ${\cal A}\in {\mathbb B}({\cal H})$ and that ${\cal B}\succeq 0$ is self-adjoint trace-class operator. Then, ${\cal B}^{1/2} {\cal A}{\cal B}^{1/2}$ is trace-class, and $\tr({\cal B}^{1/2} {\cal A}{\cal B}^{1/2}) = \tr({\cal A}{\cal B})$.
\end{claim}
\begin{proof}
Since ${\cal B}$ is trace-class, ${\cal B}^{1/2} \in \BHS({\cal H})$. This implies that ${\cal A}{\cal B}^{1/2}$ is also Hilbert-Schmidt. Thus, ${\cal B}^{1/2} {\cal A}{\cal B}^{1/2}$ is the product of two Hilbert-Schmidt operators, so it is trace-class. The trace equality follows from the cyclic property of the trace.
\end{proof}

\begin{claim}
	\label{clm:trace-to-ord}
	Suppose that ${\cal A} \succ 0$ is a self-adjoint bounded operator, and that ${\cal B}\succeq 0$ is self-adjoint trace-class operator, both on a separable Hilbert space ${\cal H}$. Suppose we have $\tr({\cal A}{\cal B}) \leq 1$. Then, ${\cal B} \preceq {\cal A}^{-1}$.
\end{claim}
\begin{proof}
	Due to the cyclicity of the trace $\tr({\cal A}^{1/2}{\cal B}{\cal A}^{1/2}) \leq 1$. The operator ${\cal A}^{1/2}{\cal B}{\cal A}^{1/2}$ is positive semidefinite, so due to Lidskii's theorem it's largest eigenvalue $\leq 1$. For ${\cal A}^{1/2}{\cal B}{\cal A}^{1/2}$, the largest eigenvalue is equal to the operator norm \ifdraft\Haim{\cite[Page 225]{HunterNachtergaele:2001}}\fi, so for any $y$, \ifdraft\Haim{using \cite[Lemma 8.26]{HunterNachtergaele:2001}}\fi
	$$
	\langle y, {\cal A}^{1/2}{\cal B}{\cal A}^{1/2} y\rangle_\mathcal{H} \leq \langle y, y\rangle_\mathcal{H}.
	$$
	Since ${\cal A}^{1/2}$ is invertible, with inverse ${\cal A}^{-1/2}$, the conclusion of the claim follows.
\end{proof}
\begin{claim}
	\label{clm:invordering}
	Let ${\cal A}, {\cal B}$ be two self-adjoint, bounded, strictly positive operators. If ${\cal A} \preceq {\cal B}$ then ${\cal A}^{-1} \succeq {\cal B}^{-1}$.
\end{claim}
\begin{proof}
	Since ${\cal B}$ is bounded and strictly positive, then it is invertible and has an invertible square root. For any $y \in {\cal H}$  let $x = {\cal B}^{-1/2} y$. We have
	\begin{eqnarray*}
		\langle y, {\cal B}^{-1/2} {\cal A} {\cal B}^{-1/2} y \rangle_\mathcal{H} &= &\langle {\cal B}^{-1/2} y, {\cal A} {\cal B}^{-1/2} y \rangle_\mathcal{H}\\
		& = & \langle x, {\cal A}x \rangle_\mathcal{H} \\
		& \leq &  \langle x, {\cal B}x \rangle_\mathcal{H} \\
		& = & \langle y, y \rangle_\mathcal{H}.
	\end{eqnarray*}
	So $ {\cal B}^{-1/2} {\cal A} {\cal B}^{-1/2} \preceq {\cal I}_{\cal H}$. Since both ${\cal A}$ and ${\cal B}$ are strictly positive, then $ {\cal B}^{-1/2} {\cal A} {\cal B}^{-1/2}$ is also strictly positive. Thus, according to Claim~\ref{claim:bound-to-inverse-bound}, ${\cal B}^{1/2} {\cal A}^{-1} {\cal B}^{1/2} \succeq {\cal I}_{\cal H}$, from which the claim easily follows.
\end{proof}
\begin{claim}
	\label{claim:stays-positive}
	Suppose that ${\cal A} \succ 0$ and ${\cal A} \succeq {\cal B}$. Then for any $0 \leq c < 1$ we have ${\cal A} - c{\cal B} \succ 0$.
\end{claim}
\begin{proof}
	Suppose by contradiction that ${\cal A} - c{\cal B} \not\succ 0$. Then for any $\epsilon > 0$ there exists a $x$ with unit norm ($\langle x, x\rangle_\mathcal{H} = 1$) such that $\langle x, ({\cal A} - c{\cal B})x \rangle_\mathcal{H} \leq \epsilon$. We have $\langle x, {\cal B} x \rangle_\mathcal{H} \geq (\langle x, {\cal A}x\rangle_\mathcal{H} - \epsilon) / c$, and since $\langle x, {\cal A}x\rangle_\mathcal{H}$ is bounded away from zero and $c < 1$, for sufficiently small $\epsilon$ we have $\langle x, {\cal B} x \rangle_\mathcal{H} > \langle x, {\cal A} x \rangle_\mathcal{H}$ so $\langle x, ({\cal A} - {\cal B})x \rangle_\mathcal{H} < 0$ which contradicts the assumption that ${\cal A} \succeq {\cal B}$.
\end{proof}

\begin{defn}
Given $x\in{\cal H}_1$ and $y \in {\cal H}_2$, we define the operator $x \otimes y : {\cal H}_2 \to {\cal H}_1$ by
\begin{equation*}
(x \otimes y)z \eqdef \langle y,z\rangle_{{\cal H}_2} x.
\end{equation*}
\end{defn}
	
\begin{claim}\label{claim:int-inner-to-trace1}
	Let ${\cal H}$ be a separable Hilbert space,  and assume that ${\cal A} \in \mathbb{B}({\cal H})$ and $v\in {\cal H}$. Then, $\langle v, {\cal A} v \rangle_{\cal H} = \tr({\cal A} (v \otimes v))$. (We remark that ${\cal A} (v \otimes v)$ is trace-class since $v \otimes v$ has finite-rank and ${\cal A}$ is bounded.)
\end{claim}
\begin{proof}
	Let $e_1, e_2,\ldots$ be an orthonormal basis for ${\cal H}$. Write $v = \sum^\infty_{i=1} \alpha_i e_i$. On one hand we have
	\begin{align*}
	\langle v, {\cal A} v \rangle_{\cal H} &= \left\langle \sum^\infty_{i=1} \alpha_i e_i, {\cal A} v \right\rangle_{\cal H} \\
	&= \frac{1}{T} \sum^\infty_{i=1} \alpha^*_i \langle e_i, {\cal A} v \rangle_{\cal H}.
	\end{align*}
	On the other we have
	\begin{align*}
	\tr({\cal A} (v \otimes v)) &= \sum^\infty_{i=1} \langle e_i, {\cal A} (v \otimes v)e_i \rangle_{\cal H} \\
	&= \sum^\infty_{i=1} \langle e_i, {\cal A}\langle v, e_i \rangle_{\cal H} v \rangle_{\cal H}  \\
	&= \sum^\infty_{i=1} \langle v, e_i \rangle_{\cal H} \langle e_i, {\cal A} v \rangle_{\cal H} \\
	&=  \sum^\infty_{i=1} \alpha^*_i \langle e_i, {\cal A} v \rangle_{\cal H},
	\end{align*}
	so the two terms are equal.
\end{proof}

\subsection{Weak integrals of operators}

We are going to work with operator-valued random variables. To reason about the expected value, we need a notion of an integral of operator-valued functions. We use a generalization of the concept of weak integrals (also called Pettis integral) of vector-valued functions~\cite{Pettis38}. 
\begin{defn}
  Let ${\cal H}_1, {\cal H}_2$ be two separable Hilbert spaces,  $G$ a measurable space and $\mu$ a measure on $G$, and consider a mapping $\mathcal{A}:G\to\mathbb{B}({\cal H}_1, {\cal H}_2)$.
  % be a mapping from $G$ to ${\cal A}$.
  If the mapping
  $(x,z)\mapsto \int_G \langle x, {\cal A}(\xi)z \rangle_{{\cal H}_2}\, d\mu(\xi)$ is a bounded sesquilinear map in $x,z$, then we say that ${\cal A}$ is a {\em weakly integrable operator valued function} and the {\em weak operator integral} is defined to be the unique bounded operator
  $$
  \int_G {\cal A}(\xi)\, d\mu(\xi) \in \mathbb{B}({\cal H}_1, {\cal H}_2)
  $$
  such that for all $x$ and $z$ we have
  $$
  \int_G \langle x, {\cal A}(\xi)z \rangle_{{\cal H}_2}\, d\mu(\xi) = \left\langle x, \left( \int_G {\cal A}(\xi)\, d\mu(\xi) \right)z \right\rangle_{{\cal H}_2}.
  $$
  The existence and uniqueness of such an operator is guaranteed by the Riesz representation theorem for sesquilinear maps~\cite[Page 92, Theorem 5]{Gilbert69}.\footnote{We remark that \cite[Page 92, Theorem 5]{Gilbert69} is stated and proved only for sesquilinear forms on the same Hilbert space (i.e., ${\cal H}_1 = {\cal H}_2$). However, it is easy to verify that the result also holds for sesquilinear forms between two Hilbert spaces.}
\end{defn}

\begin{claim}
\label{claim:linear-weak-integral}
Suppose that ${\cal A}:G\to \mathbb{B}({\cal H}_1, {\cal H}_2)$ is weakly integrable operator valued function, and ${\cal S}\in\mathbb{B}({\cal H}_2),{\cal T}\in\mathbb{B}({\cal H}_1)$. Then $\xi \mapsto {\cal T} {\cal A}(\xi){\cal S}$ is also a weakly integrable operator valued function and $$\int_G{\cal T} {\cal A}(\xi) {\cal S}\, d\mu(\xi) = {\cal T} \left(\int_G {\cal A}(\xi)\, d\mu(\xi) \right){\cal S}.$$
\end{claim}
\begin{proof}
Recall that $(x,z)\mapsto \int_G \langle x, {\cal A}(\xi) z\rangle_{{\cal H}_2}\, d\mu(\xi)$ is bounded, so there exists a $\gamma$ such that for every $x\in{\cal H}_2,z\in{\cal H}_1$
we have
$$
\left| \int_G \langle x, {\cal A}(\xi) z\rangle_{{\cal H}_2}\, d\mu(\xi) \right| \leq \gamma \|x\|_{{\cal H}_2} \|z\|_{{\cal H}_1}
$$
Let $x\in{\cal H}_2,z\in{\cal H}_1$. We have
$$
\left| \int_G \langle x, {\cal T}{\cal A}(\xi) {\cal S}z\rangle_{{\cal H}_2}d\mu(\xi) \right| =
\left| \int_G \langle {\cal T}^*x, {\cal A}(\xi) {\cal S}z\rangle_{{\cal H}_2}d\mu(\xi) \right| \leq
\gamma \|{\cal T}^* x\|_{{\cal H}_2} \|{\cal S}z\|_{{\cal H}_1} \leq \gamma \opnorm{{\cal T}}\opnorm{{\cal S}} \|x\|_{{\cal H}_2} \|z\|_{{\cal H}_1}
$$
where we used the fact that both ${\cal S}$ and ${\cal T}$ are bounded. So the mapping $(x,z)\mapsto \int_G \langle x, {\cal T}{\cal A}(\xi){\cal S}z \rangle_{{\cal H}_2} d\mu(\xi)$ is bounded and $\xi \mapsto {\cal T}{\cal A}(\xi) {\cal S}$ is weakly integrable.

We now show that the value of the integral is ${\cal T} \left(\int_G {\cal A}(\xi)d\mu(\xi) \right){\cal S}$. Again, for any $x\in{\cal H}_2,z\in{\cal H}_1$:
$$
\left\langle x,  {\cal T}\left( \int_G {\cal A}(\xi) d\mu(\xi) \right){\cal S}z\right\rangle_{{\cal H}_2} = \left\langle {\cal T}^* x,  \left( \int_G {\cal A}(\xi) d\mu(\xi) \right){\cal S}z\right\rangle_{{\cal H}_2}
$$
By definition of $\int_G {\cal A}(\xi) d\mu(\xi)$ we have
$$
\left\langle {\cal T}^* x,  \left( \int_G {\cal A}(\xi) d\mu(\xi) \right){\cal S}z\right\rangle_{{\cal H}_2} = \int_G \langle {\cal T}^* x, {\cal A}(\xi) {\cal S}z\rangle_{{\cal H}_2}
= \int_G \langle  x, {\cal T} {\cal A}(\xi) {\cal S}z\rangle_{{\cal H}_2}
$$
so indeed $\int_G{\cal T} {\cal A}(\xi) {\cal S}d\mu(\xi) = {\cal T} \left(\int_G {\cal A}(\xi)d\mu(\xi) \right){\cal S}$.
\end{proof}

\begin{claim}
	\label{claim:int-inner-to-trace}
	Let $\rho, \mu$ be two, possibly different, probability measures, on $\RR$, and let ${\cal A} \in {\mathbb B}(L_2(\rho))$ be self-adjoint and positive semi-definite, and let ${\cal B} \in \BTC(L_2(\rho))$. Assume that there exists an orthonormal basis for $L_2(\rho)$ consisting of eigenvectors of ${\cal A}$. Given a mapping $\eta\in\RR \mapsto v_\eta \in L_2(\rho)$ such that ${\cal B} = \int_\RR (v_\eta \otimes v_\eta)d\mu(\eta)$ we have:
	$$
	\int_\RR \langle v_\eta, {\cal A} v_\eta \rangle_\rho d\mu(\eta) = \tr({\cal A B})
	$$
\end{claim}
\begin{proof}
	Let $e_1, e_2,\ldots$ be an orthonormal basis for $L_2(\rho)$ consisting of eigenvectors of ${\cal A}$. Using Claim \ref{claim:int-inner-to-trace1}, we have
	\begin{align*}
	\int_\RR \langle v_\eta, {\cal A} v_\eta \rangle_\rho d\mu(\eta)
	&= \int_\RR \tr({\cal A}(v_\eta \otimes v_\eta))d\mu(\eta)\\
	&=  \int_\RR \sum^\infty_{i=1} \langle e_i, {\cal A}(v_\eta \otimes v_\eta)e_i \rangle_\rho d\mu(\eta)\\
    &=   \sum^\infty_{i=1} \int_\RR \langle e_i, {\cal A}(v_\eta \otimes v_\eta)e_i \rangle_\rho d\mu(\eta)\\
    &=   \sum^\infty_{i=1}  \langle e_i, \int_\RR {\cal A}(v_\eta \otimes v_\eta)d\mu(\eta) e_i \rangle_\rho \\
    &=   \sum^\infty_{i=1}  \langle e_i, {\cal A} \int_\RR (v_\eta \otimes v_\eta)d\mu(\eta) e_i \rangle_\rho \\
	&= \sum^\infty_{i=1} \langle e_i, {\cal A B} e_i \rangle_\mu \\
	& = \tr( {\cal A B})
	\end{align*}
	where the exchange of the integral and infinite sum in the third equality is justified by Tonelli's Theorem. In order to apply Tonelli's theorem we need to show that $\langle e_i, {\cal A}(v_\eta \otimes v_\eta)e_i \rangle_\rho \geq 0$ for every $i$ and $\eta$. This is indeed the case since $\langle e_i, {\cal A}(v_\eta \otimes v_\eta)e_i \rangle_\rho = \langle {\cal A} e_i, (v_\eta \otimes v_\eta)e_i \rangle_\rho = \lambda_i \langle  e_i, (v_\eta \otimes v_\eta)e_i \rangle_\rho \geq 0$ where $\lambda_i$ is the eigenvalue corresponding to $e_i$. Note that since ${\cal A}$ is self-adjoint and positive semi-definite, $\lambda_i$ is real and non-negative. We also used the immediate fact that  $v_\eta \otimes v_\eta$ is positive semi-definite.
\end{proof}
\spara{Remark:}
One way to guarantee that there is an orthonormal basis of eigenvectors of ${\cal A}$ is to require ${\cal A}$ to be compact. However, it is quite possible for ${\cal A}$ not to be compact, and still have an orthonormal basis of eigenvectors. In fact, we primarily apply Claim~\ref{claim:int-inner-to-trace} to operators of the form $({\cal C} + \epsilon {\cal I})^{-1}$ where ${\cal C}$ is compact, and such operators have an orthonormal basis of eigenvectors (since they share eigenvectors with ${\cal C}$).

We say that a weakly integrable ${\cal A}(\cdot)$ is {\em self-adjoint} if ${\cal A}(\xi)$ is self-adjoint for all $\xi$. It is easy to verify that if ${\cal A}(\cdot)$ is self-adjoint, then $\int_G A(\xi)d\mu(\xi)$ is self-adjoint as well.
\begin{claim}
	\label{claim:dominated-expectation}
	Suppose that ${\cal A}, {\cal B}:G\to\mathbb{B}({\cal H})$ are two self-adjoint weakly integrable operator valued functions. If, with respect to a measure $\mu$ on $G$,  ${\cal A}(\xi) \preceq {\cal B}(\xi)$ almost everywhere, then $\int_G {\cal A}(\xi) d\mu(\xi) \preceq \int_G {\cal B}(\xi) d\mu(\xi)$.
\end{claim}
\begin{proof}
	For every $x\in {\cal H}$,
	$$
	\left\langle x, \int_G {\cal A}(\xi) d\mu(\xi) x \right\rangle_{\cal H} = \int_G \langle x, {\cal A}({\xi}) x \rangle_{\cal H}d\mu(\xi) \leq \int_G \langle x, {\cal B}({\xi}) x \rangle_{\cal H}d\mu(\xi) = \left\langle x, \int_G {\cal B}(\xi) d\mu(\xi) x \right\rangle_{\cal H}
	$$
	so indeed $\int_G {\cal A}(\xi) d\mu(\xi) \preceq \int_G {\cal B}(\xi) d\mu(\xi)$.
\end{proof}

\begin{claim}
	\label{claim:dominated-expectation2}
	Suppose that ${\cal B}:G\to\mathbb{B}({\cal H})$ is a self-adjoint weakly integrable operator valued function. Consider another self-adjoint operator valued function ${\cal A}:G\to\mathbb{B}({\cal H})$. If for every $\xi\in G$ we have $0 \preceq {\cal A}(\xi) \preceq {\cal B}(\xi)$, then ${\cal A}$ is weakly integrable and $\int_G {\cal A}(\xi) d\mu(\xi) \preceq \int_G {\cal B}(\xi) d\mu(\xi)$.
\end{claim}

\begin{proof}
	We need to prove only that ${\cal A}$ is weakly integrable, since the integral bound follows from Claim \ref{claim:dominated-expectation}. A sesquilinear form is bounded if and only if the associated quadratic form is bounded~\cite[Page 92, Theorem 3]{Gilbert69}, so we need to show that the integral of the quadratic form associated with ${\cal A}$ is bounded. Since ${\cal A}(\xi)$ is always positive semidefinite, for any $x$
    $$
   \left|\int_G \langle x, {\cal A}(\xi) x \rangle_{\cal H} d\mu(\xi) \right| = \int_G \langle x, {\cal A}(\xi) x \rangle_{\cal H} d\mu(\xi) \leq
    \int_G \langle x, {\cal B}(\xi) z \rangle_{\cal H} d\mu(\xi) = \left|\int_G \langle x, {\cal B}(\xi) x \rangle_{\cal H} d\mu(\xi)\right|
    $$
    and since the integral of the quadratic form associated with ${\cal B}$ is bounded (since ${\cal B}$ is weakly integrable) we conclude that integral quadratic form associated with ${\cal A}$ is bounded, so indeed ${\cal A}$ is weakly integrable.
\end{proof}

\subsection{Concentration of random operators}

Let ${\cal A}:G\to\mathbb{B}({\cal H})$ be a weakly integrable operator valued function. If the underlying measure $\mu$ is a probability measure, then we shall call ${\cal A}$ a {\em random operator}, and write
$$
\E({\cal A}) = \int_G {\cal A}(\xi) d\mu(\xi).
$$
Certain matrix concentration results can be generalized to the case that ${\cal A}$ is a random operator which takes only self-adjoint Hilbert-Schmidt values. The underlying reason is that Hilbert-Schmidt operators can be well-approximated using finite rank operators. The basic technique is outlined in \cite[Section 3.2]{Minsker:2017}. We use this technique to prove the following lemma.
\begin{lemma}
\label{lem:op-concentrate}
Suppose that ${\cal H}$ is a separable Hilbert space, and that ${\cal B}$ is a fixed self-adjoint Hilbert-Schmidt operator on ${\cal H}$. Let ${\cal R}$ be a self-adjoint Hilber-Schmidt random operator that satisfies
$$
\E({\cal R}) = {\cal B}\quad and \quad \opnorm{{\cal R}} \leq L$$
Let ${\cal M}$ be another self-adjoint trace-class operator such that $\E({\cal R}^2) \preceq {\cal M}$.
Form the operator sampling estimator
$$
\bar{{\cal R}}_n = \frac{1}{n} \sum^n_{k=1} {\cal R}_k
$$
where each ${\cal R}_k$ is an independent copy of ${\cal R}$. Then, for all $t > \sqrt{\opnorm{{\cal M}} / n} + 2L/3n$,
\begin{equation}
\label{eq:conc-operator}
\Pr\left(\opnorm{\bar{{\cal R}}_n - {\cal B}}  > t\right) \leq \frac{8\tr({\cal M})}{\opnorm{{\cal M}}}\exp\left(\frac{-nt^2 / 2}{\opnorm{{\cal M}} + 2Lt/3}\right).
\end{equation}
\end{lemma}
\begin{proof}
	Let $e_1, e_2, \dots$ be the eigenvectors of ${\cal M}$, ordered according to the magnitude of the corresponding eigenvalue, and let ${\cal P}_j$ be the orthogonal projector on the span of $e_1, e_2, \dots, e_j$. Consider the finite-rank operators ${\cal R}^{(j)} = {\cal P}_j {\cal R} {\cal P}_j$, ${\cal R}^{(j)}_k = {\cal P}_j {\cal R}_k {\cal P}_j$,  $\bar{{\cal R}}^{(j)}_n = {\cal P}_j \bar{{\cal R}}_n {\cal P}_j$, ${\cal B}^{(j)} = {\cal P}_j {\cal B} {\cal P}_j$ and ${\cal M}^{(j)} = {\cal P}_j {\cal M} {\cal P}_j$. We will apply on these operator sequences the matrix version of the current lemma~\cite{AvronKapralovMusco:2017}\footnote{The lemma in~\cite{AvronKapralovMusco:2017} is stated as a bound on $\Pr\left(\opnorm{\bar{{\cal R}}_n - {\cal B}}  \geq t\right)$, while for operators strict inequality is necessary. It is easy to verify that the matrix version of the Lemma continues to hold for $\Pr\left(\opnorm{\bar{{\cal R}}_n - {\cal B}}  > t\right)$.}
	
	Due to linearity of weak operator integrals we have $\E({\cal R}^{(j)}) = {\cal P}_j {\cal B}^{(j)} {\cal P}_j$. We can bound the operator norm of ${\cal R}^{(j)}$: $\opnorm{{\cal R}^{(j)}} \leq \opnorm{{\cal P}_j {\cal R} {\cal P}_j} \leq \opnorm{{\cal P}_j}^2 \opnorm{{\cal R}} \leq L$ since the operator norm of a projection operator is $1$. Using the fact that ${\cal P}_j \preceq {\cal I}_{\cal H} $ and so ${\cal R} {\cal P}_j {\cal R} \preceq {\cal R}^2$ we have
	$$
	\E(({\cal R}^{(j)})^2) = {\cal P}_j \E({\cal R} {\cal P}_j {\cal R}) {\cal P}_j \preceq  {\cal P}_j \E({\cal R}^2) {\cal P}_j \preceq {\cal M}^{(j)}
	$$
	Now applying the aforementioned matrix version of the current lemma\footnote{Technically, the aforementioned concentration result is for {\em matrices}, while here we deal with abstract operators on finite dimensional subspaces. We can address this issue by using the corresponding transformation matrices, but we find that to be tedious details.} we find that
	\begin{equation}
	\label{eq:conc-finite-approx}
	\Pr\left(\opnorm{\bar{{\cal R}}^{(j)}_n - {\cal B}^{(j)}}  \geq t\right) \leq \frac{8\tr({\cal M}^{(j)})}{\opnorm{{\cal M}^{(j)}}}\exp\left(\frac{-nt^2 / 2}{\opnorm{{\cal M}^{(j)}} + 2Lt/3}\right).
	\end{equation}
	Due to the way we constructed ${\cal P}_j$, and ${\cal M}$ being trace-class, we have $\tr({\cal M}^{(j)}) \to \tr({\cal M})$ as $j\to\infty$.
	Furthermore, since ${\cal M}$ is trace-class, ${\cal P}_j {\cal M} \to {\cal M}$ uniformly~\cite[Theorem 9.21]{HunterNachtergaele:2001}, and so ${\cal M}^{(j)} \to {\cal M}$ while implies that $\opnorm{{\cal M}^{(j)}} \to \opnorm{{\cal M}}$. Thus, the entire right side of~\eqref{eq:conc-finite-approx} converges to the right side of~\eqref{eq:conc-operator}, so
	$$
	\liminf_{j\to\infty} \Pr\left(\opnorm{\bar{{\cal R}}^{(j)}_n - {\cal B}^{(j)}}  > t\right) \leq \frac{8\tr({\cal M})}{\opnorm{{\cal M}}}\exp\left(\frac{-nt^2 / 2}{\opnorm{{\cal M}} + 2Lt/3}\right).
	$$
	
	 Let $G$ and $\mu$ denote the underlying probability space and probability measure. Let $f_j$ now denote the indicator function for the event $\opnorm{\bar{{\cal R}}^{(j)}_n - {\cal B}^{(j)}}  > t$, and $f$ the indicator for the event
	$\opnorm{\bar{{\cal R}}_n - {{\cal B}}}  > t$. Again, due to the fact that $\bar{{\cal R}}_n - {{\cal B}}$ is Hilbert-Schmidt
	we have $\bar{{\cal R}}^{(j)}_n - {\cal B}^{(j)} \to \bar{{\cal R}}_n - {{\cal B}}$, while implies that that for any $\xi \in G$, $f(\xi) = \liminf_{j\to\infty}f_j(\xi)$.
	Now due to Fatou's lemma:
	\begin{align*}
	\Pr\left(\opnorm{\bar{{\cal R}}_n - {\cal B}}  > t\right) &= \int_G f(\xi)d\mu(\xi)\\
	&= \int_G \liminf_{j\to\infty}f_j(\xi) d\mu(\xi)\\
	&\leq \liminf_{j\to\infty} \int_G f_j(\xi) d\mu(\xi)\\
	& =\liminf_{j\to\infty} \Pr\left(\opnorm{\bar{{\cal R}}^{(j)}_n - {\cal B}^{(j)}}  > t\right)\\
	&\leq \frac{8\tr({\cal M})}{\opnorm{{\cal M}}}\exp\left(\frac{-nt^2 / 2}{\opnorm{{\cal M}} + 2Lt/3}\right).
	\end{align*}

\end{proof}

\section{Properties of the ridge leverage scores}
\label{app:leverage_scores}

\subsection{Basic facts about leverage scores}

In this section we prove Theorem \ref{thm:leverageProps}, giving fundamental properties of the ridge leverage scores that we use both in bounding these scores and proving that ridge leverage score sampling can be used to solve the regularized regression problem of \eqref{eq:least_squares_setup} (and hence Problem \ref{prob:unformal_interp} by Claim \ref{claim:regression_reduction}).
\begin{reptheorem}{thm:leverageProps}[Leverage Function Properties]
Letting $\tmu(t)$ be the ridge leverage function of Definition \ref{def:ridgeScores}, that is
	\begin{align}
	\label{eq:MaxRep}
	\tmu(t) = \frac{1}{T} \cdot \max_{\{\alpha \in L_2(\mu):\, \norm{\alpha}_\mu > 0\}} \frac{|[\Fmu^* \alpha](t) |^2 }{\norm{\Fmu^* \alpha}_T^2 + \epsilon \norm{\alpha}_\mu^2},
	\end{align}
 and let $\varphi_t \in L_2(\mu)$ be defined by $\varphi_t(\xi) = e^{-2\pi i t \xi}$, we have the following basic properties:
\begin{itemize}
\item The leverage scores integrate to the statistical dimension:
\begin{align}\label{eq:StatDimRep}
	\int_0^T \tmu(t) dt= \smu \eqdef  \tr(\Kmu (\Kmu + \epsilon \mathcal{I}_T)^{-1}).
\end{align}
\item Inner Product characterization:
\begin{align}\label{eq:InnerProdRep}
\tmu(t) =\frac{1}{T} \cdot \langle \varphi_t,  (\Gmu + \epsilon \Imu)^{-1}  \varphi_t \rangle_\mu.
\end{align}
\item Minimization Characterization:
\begin{align}\label{eq:MinRep}
\tmu(t) = \frac{1}{T} \cdot \min_{\beta \in L_2(T)} \frac{\norm{\Fmu \beta - \varphi_t}_\mu^2}{\epsilon} + \norm{\beta}_{T}^2.
\end{align}
\end{itemize}
\end{reptheorem}

\begin{proof}
	
Recall, that given $t \in [0,T]$, we defined $\varphi_t(\xi) \eqdef e^{-2\pi i t \xi}$ ($\varphi_t \in L_2 (\mu)$).
It is easy to verify that:
\begin{equation}
\label{eq:Gint}
{\cal G}_\mu = \frac{1}{T}\int^T_0 (\varphi_t \otimes \varphi_t) dt.
\end{equation}
	
To prove the equality between Equations~\eqref{eq:MaxRep}, \eqref{eq:InnerProdRep}, and \eqref{eq:MinRep}, we first show that the right hand side of \eqref{eq:InnerProdRep} is equal to the right hand side of \eqref{eq:MinRep} and then show that the right hand side of \eqref{eq:InnerProdRep} is equal to the right hand side of \eqref{eq:MaxRep}. 

First, we need an auxiliary lemma regarding the solution of regularized least squares problems. If $\cal A$ is matrix with full column rank or a one-to-one linear operator between finite-dimensional Hilbert spaces, and $b$ some vector, then $F(x) = \norm{{\cal A} x - b}^2$ has a unique minimizer. In infinite dimension spaces, this remains true if only the co-domain of $\cal A$ is infinite dimensional. However, if both the domain and co-domain are infinite dimensional there might not be a minimizer even if the $\cal A$ is bounded: the range of $\cal A$ might not be closed, so it is possible that $\norm{{\cal A}x - b} > 0$ for every $x$, but also that there exists a series $\{x_n\}$ such that $\norm{\mathcal{A}x_n - b} \to 0$ as $n \to \infty$.  However, once we introduce a ridge term (i.e., minimize $F(x) = \norm{\mathcal{A} x - b}^2 + \lambda \norm{x}^2$ for some $\lambda > 0$) there is always a unique minimizer (as long as $\mathcal{A}$ is bounded), due to the extreme value theorem (since we can bound the search domain). Furthermore, we can write an analytic expression for the minimizer in an analogous way to the finite dimensional case, as the following lemma shows.
\begin{lemma}[Regularized Least Squares Minimizer]\label{lem:ridgeMinimizer}
	Let $\HH_1$ and $\HH_2$ be two Hilbert spaces, and $\mathcal{A}:\HH_1 \to \HH_2$ be a bounded linear operator. Let $b \in \HH_2$ and $\lambda > 0$. The function
	$$
	F(x) = \norm{\mathcal{A} x - b}^2_{\HH_2} + \lambda \norm{x}^2_{\HH_1}
	$$
	has a unique minimizer which is $x^\star = \mathcal{A}^* (\mathcal{A}\mathcal{A}^* + \lambda \mathcal{I}_{\HH_2})^{-1} b$.
\end{lemma}
\begin{proof}
	Consider the Hilbert space $\HH_1 \times \HH_2$ equipped with the inner product
	$$
	\langle (a_1, a_2), (b_1, b_2) \rangle_{\HH_1 \times \HH_2} \eqdef \langle a_1, b_1 \rangle_{\HH_1} + \langle a_2, b_2 \rangle_{\HH_2}.
	$$
	Define the operator $\mathcal{T}:\HH_1 \to \HH_1 \times \HH_2,\quad \mathcal{T}(x) = (\sqrt{\lambda} x,\mathcal{A}x)$. Let $y = (0,b) \in \HH_1 \times \HH_2$. We have $F(x) = \norm{\mathcal{T}x - y}^2_{ \HH_1 \times \HH_2}$.
	Thus, we need to show that there is a unique point $\tilde y \in \range(\mathcal T)$ that minimizes $\norm{\tilde y - y}^2_{\HH_1 \times \HH_2}$ and that $\tilde y = \mathcal{T}x^\star$ for  $x^\star = \mathcal{A}^* (\mathcal{A}\mathcal{A}^* + \lambda \mathcal{I}_{\HH_2})^{-1} b$.
	
	The operator $\mathcal{T}$ is a bounded linear operator, so it is continuous. We also have that for every $x\in \HH_1$, $\norm{\mathcal T x}^2_{ \HH_1 \times \HH_2} \geq \lambda \norm{x}^2_{\HH_1}$ where $\lambda > 0$, so $\mathcal T$ is bounded from below. So $\mathcal T$ has a closed range~\cite[Theorem 2.5]{AA02}. Thus, there is a unique $\tilde y \in \range(\mathcal T)$ that minimizes $\norm{\tilde y - y}^2_{\HH_1 \times \HH_2}$, and that $\tilde y$ is the unique element of $\range(\mathcal T)$ with the property $y - \tilde y \perp \range(\mathcal T)$~\cite[Theorem 6.13]{HunterNachtergaele:2001}. So it suffices to show that for every $x\in \HH_1$ we have $y - \mathcal Tx^\star \perp \mathcal Tx$. We compute:
	\begin{eqnarray*}
	\langle y - \mathcal Tx^\star, \mathcal Tx \rangle_{\HH_1 \times \HH_2} & = & \langle (-\sqrt{\lambda}\mathcal A^* (\mathcal A \mathcal A^* + \lambda \mathcal{I}_{\HH_2})^{-1} b,b-\mathcal A \mathcal A^* (\mathcal A\mathcal  A^* + \lambda \mathcal{I}_{\HH_2})^{-1} b), (\sqrt{\lambda} x,\mathcal Ax) \rangle_{\HH_1 \times \HH_2} \\
	& = & \langle (-\sqrt{\lambda}\mathcal A^* (\mathcal A \mathcal A^* + \lambda \mathcal{I}_{\HH_2})^{-1} b, \lambda (\mathcal A \mathcal A^* + \lambda \mathcal{I}_{\HH_2})^{-1} b) , (\sqrt{\lambda} x,\mathcal Ax) \rangle_{\HH_1 \times \HH_2} \\
	& = & -\lambda \langle  \mathcal A^* (\mathcal A \mathcal A^* + \lambda \mathcal{I}_{\HH_2})^{-1} b, x \rangle_{\HH_1} + \lambda \langle (\mathcal A \mathcal A^* + \lambda \mathcal{I}_{\HH_2})^{-1} b, \mathcal{A} x\rangle_{\HH_2} \\
	& = & -\lambda \langle  \mathcal A^*( \mathcal A \mathcal A^* + \lambda \mathcal{I}_{\HH_2})^{-1} b,  x \rangle_{\HH_1} + \lambda \langle \mathcal A ^*(A A^* + \lambda \mathcal{I}_{\HH_2})^{-1} b, x\rangle_{\HH_1} = 0.
	\end{eqnarray*}
	So indeed, for every $x\in \HH_1$ we have $y - \mathcal Tx^\star \perp \mathcal Tx$ and $x^\star$ is the unique minimizer.
\end{proof}

Using Lemma \ref{lem:ridgeMinimizer} we now proceed with the proof of Theorem \ref{thm:leverageProps}.
\begin{corollary}
	Let
	$$
	\beta^\star = \Fmu^* (\Gmu + \epsilon \Imu)^{-1} \varphi_t.
	$$
	Then,
	$$
	\frac{1}{T} \cdot \left( \frac{\norm{\Fmu \beta^\star - \varphi_t}_\mu^2}{\epsilon} + \norm{\beta^\star}_{T}^2 \right) =
	\frac{1}{T} \cdot \min_{\beta \in L_2(T)} \frac{\norm{\Fmu \beta - \varphi_t}_\mu^2}{\epsilon} + \norm{\beta}_{T}^2.
	$$
\end{corollary}

\begin{claim}
	We have
	$$
	\langle \varphi_t,  (\Gmu + \epsilon \Imu)^{-1}  \varphi_t \rangle_\mu =  \frac{\norm{\Fmu \beta^\star - \varphi_t}_\mu^2}{\epsilon} + \norm{\beta^\star}_{T}^2
	$$
	so the right hand side of \eqref{eq:InnerProdRep} is equal to the right hand side of \eqref{eq:MinRep}.
\end{claim}

\begin{proof}
	We compute:
	\begin{eqnarray*}
	\norm{\beta^\star}_{T}^2 &=  & \langle \Fmu^*(\Gmu + \epsilon \Imu)^{-1}  \varphi_t, \Fmu^* (\Gmu + \epsilon \Imu)^{-1}  \varphi_t \rangle_\mu\\
	& = & \langle(\Gmu + \epsilon \Imu)^{-1}  \varphi_t, \Gmu (\Gmu + \epsilon \Imu)^{-1}  \varphi_t \rangle_\mu \\
	& = & \langle(\Gmu + \epsilon \Imu)^{-1}  \varphi_t, (\Gmu + \epsilon \Imu - \epsilon \Imu)(\Gmu + \epsilon \Imu)^{-1}  \varphi_t \rangle_\mu \\
	& = & \langle \varphi_t,  (\Gmu + \epsilon \Imu)^{-1}  \varphi_t \rangle_\mu - \epsilon \langle \varphi_t,  (\Gmu + \epsilon \Imu)^{-2}  \varphi_t \rangle_\mu
	\end{eqnarray*}
	and
	\begin{eqnarray*}
	\norm{\Fmu \beta^\star - \varphi_t}_\mu^2 & = &\norm{\Fmu\Fmu^*(\Gmu + \epsilon \Imu)^{-1} \varphi_t - \varphi_t}_\mu^2 \\
	& = & \norm{\left(\Gmu(\Gmu + \epsilon \Imu)^{-1} -\Imu \right)\varphi_t}_\mu^2 \\
	& = & \norm{\left((\Gmu + \epsilon \Imu - \epsilon \Imu) (\Gmu + \epsilon \Imu)^{-1} -\Imu \right)\varphi_t}_\mu^2 \\
	& = & \norm{\epsilon (\Gmu + \epsilon \Imu)^{-1} \varphi_t}_\mu^2 \\
	& = & \epsilon^2 \langle \varphi_t,  (\Gmu + \epsilon \Imu)^{-2}  \varphi_t \rangle_\mu
	\end{eqnarray*}
	Summing the last equalities completes the proof.
\end{proof}

\begin{claim}
	We have
	$$
		\langle \varphi_t,  (\Gmu + \epsilon \Imu)^{-1}  \varphi_t \rangle_\mu = \max_{\{\alpha \in L_2(\mu):\, \norm{\alpha}_\mu > 0\}} \frac{|[\Fmu^* \alpha](t) |^2 }{\norm{\Fmu^* \alpha}_T^2 + \epsilon \norm{\alpha}_\mu^2}
	$$
	so the right hand side of \eqref{eq:InnerProdRep} is equal to the right hand side of \eqref{eq:MaxRep}.
\end{claim}

\begin{proof}
	We can reformulate the previous claim as
	:\begin{equation*}
	\begin{aligned}
	\langle \varphi_t,  (\Gmu + \epsilon \Imu)^{-1}  \varphi_t \rangle_\mu  = & {\text{ minimum}}
	& & \norm{\beta}_\mu^2 + \norm{u}^2_T \\
	& \beta \in L_2(\mu) ;
	& & u \in L_2(T)\\
	& \text{subject to:}
	& & \Fmu \beta + \sqrt{\epsilon}u = \varphi_t.
	\end{aligned}
	\end{equation*}
	Let the optimal solution be $\beta^\star$ and $u^\star$. We have $\varphi_t = \Fmu \beta^\star + \sqrt{\epsilon}u^\star$, hence for any $0\neq\alpha\in L_2(\mu)$:
	\begin{eqnarray*}
	|[\Fmu^* \alpha](t) | & = & | \langle \varphi_t, \alpha \rangle_\mu | \\
	& = & | \langle \alpha, \varphi_t \rangle_\mu | \\
	& = & | \langle \alpha, \Fmu \beta^\star + \sqrt{\epsilon}u^\star \rangle_\mu | \\
	& \leq & | \langle \alpha, \Fmu \beta^\star \rangle_\mu | + | \langle \alpha, \sqrt{\epsilon}u^\star \rangle_\mu | \\
	& = & | \langle \Fmu^* \alpha, \beta^\star \rangle_T | + | \langle \alpha, \sqrt{\epsilon}u^\star \rangle_\mu | \\
	& \leq & \norm{(\Fmu^*\alpha)}_T \cdot \norm{\beta^\star}_T  + \sqrt{\epsilon} \norm{\alpha}_\mu \cdot \norm{u^\star}_\mu
	\end{eqnarray*}
	where the last inequality follows from Cauchy-Schwarz inequality. Using Cauchy-Schwarz again:
	\begin{eqnarray*}
		|[\Fmu^* \alpha](t) |^2 & \leq &\left( \norm{(\Fmu^*\alpha)}_T \cdot \norm{\beta^\star}_T  + \sqrt{\epsilon} \norm{\alpha}_\mu \cdot \norm{u^\star}_\mu\right)^2 \\
		& \leq & \left(\norm{\Fmu^*\alpha}^2_T   + \epsilon \norm{\alpha}^2_\mu \right) \cdot \left(\norm{\beta^\star}^2_T  + \norm{u^\star}^2_\mu \right)
	\end{eqnarray*}
	So for every $0\neq\alpha\in L_2(\mu)$:
	$$
	\frac{|[\Fmu^* \alpha](t) |^2 }{\norm{\Fmu^* \alpha}_T^2 + \epsilon \norm{\alpha}_\mu^2} \leq \norm{\beta^\star}^2_T  + \norm{u^\star}^2_\mu = \langle \varphi_t,  (\Gmu + \epsilon \Imu)^{-1}  \varphi_t \rangle_\mu
	$$
	We conclude by showing that the maximum value is attained. Let $\alpha^\star = (\Gmu + \epsilon \mathcal{I}_\epsilon)^{-1} \varphi_t$. We have
	$$
	\norm{\Fmu^* \alpha^\star}_T^2 + \epsilon \norm{\alpha^\star}_\mu^2 = \langle \alpha^\star, (\Gmu + \epsilon \mathcal{I}_\epsilon) \alpha^\star\rangle = \langle \varphi_t, (\Gmu + \epsilon \mathcal{I}_\epsilon)^{-1} \varphi_t \rangle_\mu
	$$ and finally,
	\begin{equation*}
	\frac{|[\Fmu^* \alpha^\star](t) |^2 }{\norm{\Fmu^* \alpha^\star}_T^2 + \epsilon \norm{\alpha}_\mu^2} = \frac{| \langle \varphi_t, \alpha^\star \rangle_\mu |^2}{\langle \varphi_t, (\Gmu + \epsilon \mathcal{I}_\epsilon)^{-1} \varphi_t \rangle_\mu} = \langle \varphi_t,  (\Gmu + \epsilon \Imu)^{-1}  \varphi_t \rangle_\mu.
	\end{equation*}
	
\end{proof}

We now turn to showing that the leverage function integrates to the statistical dimension.
\begin{claim}
	\label{claim:stat-int}
		\begin{align*}
		\int_0^T \tmu(t) dt= \smu.
		\end{align*}
\end{claim}

\begin{proof}
	It follows from Eq.~\eqref{eq:Gint} and Claim~\ref{claim:int-inner-to-trace}  that $\int_0^T \tmu(t) = \tr( (\Gmu + \epsilon \Imu)^{-1}\Gmu)$. The claim follows by noting that ${\cal K}_\mu$ and ${\cal G}_\mu$ have the same eigenvalues (both operators are compact self adjoint operators, so their spectrum consists of only eigenvalues, and it is easy to verify that if $x$ is an eigenvector $\Kmu$ then $\Fmu x$ is eigenvector of $\Gmu$).

\end{proof}
We thus have completed the proof of Theorem \ref{thm:leverageProps}.
\end{proof}

\subsection{Operator Approximation via Leverage Score Sampling}

Analogs of the following concentration result are well known for matrices. Accordingly, the proof is an adaptation of standard proofs for finite matrix approximation by leverage score sampling, where matrix concentration results are replaced with operator concentration results. A similar strategy was employed in~\cite{Bach:2017}.

\begin{lemma}\label{lem:operatorApproximation}
	Consider the conditions of Theorem \ref{thm:baseSampling}, and denote $\widehat{\cal G}_\mu = \bv{F} \bv{F}^*$. Let $\Delta \leq 1/2$ and $\epsilon \le \opnorm{\Gmu}$. If $s \geq \frac{8}{3}\Delta^{-2}\tilde{s}_{\mu,\epsilon}\ln(16\tilde{s}_{\mu,\epsilon}^2/\delta)$, then
	\begin{equation}
	\label{eq:spectral}
	(1-\Delta)({\cal G}_\mu + \epsilon \Imu) \preceq \widehat{\cal G}_\mu + \epsilon\Imu \preceq 	(1+\Delta)({\cal G}_\mu +\epsilon \Imu)
	\end{equation}
	with probability of at least $1-\delta$.
\end{lemma}
\begin{proof}
	%We can assume without loss of generality that $\epsilon \le \opnorm{\Gmu}$.
	
	The condition~\eqref{eq:spectral} is equivalent to
	\begin{equation*}
	{\cal G}_\mu -\Delta({\cal G}_\mu + \epsilon \Imu) \preceq \widehat{\cal G}_\mu \preceq {\cal G}_\mu + \Delta({\cal G}_\mu + \epsilon \Imu)
	\end{equation*}
	By composing with $({\cal G}_\mu + \epsilon \Imu)^{-1/2}$ on the left and right, we find that the condition is equivalent to \ifdraft\Haim{\cite[Proposition 8.26]{HunterNachtergaele:2001}}\fi
	\begin{equation}
	\opnorm{({\cal G}_\mu + \epsilon \Imu)^{-1/2}(\widehat{\cal G}_\mu  -  {\cal G}_\mu)({\cal G}_\mu + \epsilon \Imu)^{-1/2}} \leq \Delta.
	\end{equation}
	Noticing that
	$$
	\bv{F} g = \sum^s_{j=1} w_j g(j) \varphi_{t_j}
	$$
	and that
	$$
	[\bv{F}^* g](j) = w_j \langle \varphi_{t_j}, g \rangle_\mu
	$$
	we understand that $\widehat{{\cal G}}_\mu = \sum^s_{j = 1} w^2_j(\varphi_{t_j}\otimes\varphi_{t_j})$.
	Let
	$${\cal X}_j \eqdef s w^2_j ({\cal G}_\mu + \epsilon \Imu)^{-1/2} (\varphi_{t_j} \otimes \varphi_{t_j})({\cal G}_\mu + \epsilon \Imu)^{-1/2}.$$
	Note that ${\cal X}_j$ is self-adjoint and Hilbert-Schmidt (since it has finite rank). We have
	$$
	({\cal G}_\mu + \Imu)^{-1/2}\widehat{\cal G}_\mu({\cal G}_\mu + \Imu)^{-1/2} = \frac{1}{s}\sum^s_{j=1} {\cal X}_j.
	$$
	Since the time samples are drawn randomly, ${\cal X}_1, \dots, {\cal X}_s$ are random operators.
	We also have, using Claim~\ref{claim:linear-weak-integral},
	\begin{align*}
	\mathbb{E}_{t_j \propto \tilde{\tau}_{\mu,\epsilon}}[{\cal X}_j] &= ({\cal G}_\mu + \epsilon \Imu)^{-1/2} 	\mathbb{E}_{t_j \propto \tilde{\tau}_{\mu,\epsilon}} \left[ s w^2_j (\varphi_{t_j} \otimes \varphi_{t_j}) \right] ({\cal G}_\mu + \epsilon \Imu)^{-1/2}.
	\end{align*}
	Write $w(t) = \sqrt{\frac{\tilde{s}_{\mu,\epsilon}}{T\cdot \ttmu(t)}}$. For every $x,z\in L_2(\mu)$,
	$$
	\int^T_0 \left\langle x, w(t)^2\cdot(\varphi_t \otimes \varphi_t)z\right\rangle_\mu\cdot(\ttmu(t)/\tsmu)dt =
	\frac{1}{T} \int^T_0 \left\langle x, (\varphi_t \otimes \varphi_t)z\right\rangle_\mu dt = \langle x, \Gmu z\rangle_\mu
	$$
	which shows that
	$$
    \E_{t_j \propto \ttmu} \left[ s w^2_j (\varphi_{t_j} \otimes \varphi_{t_j}) \right] = \Gmu
    $$
	so,
	\begin{equation}
	\E_{t_j \propto \ttmu}[{\cal X}_j] = ({\cal G}_\mu + \epsilon \Imu)^{-1/2} {\cal G}_\mu ({\cal G}_\mu + \epsilon \Imu)^{-1/2}.
	\end{equation}
	
	Next, we bound the operator norm of ${\cal X}_j$. The random operator only takes values that are both positive semidefinite and rank one, so the operator norm of ${\cal X}_j$ is equal to the trace of the operator. Thus, we have
	\begin{align*}
\opnorm{{\cal X}_j} &= s w_j^2 \tr\left( ({\cal G}_\mu + \epsilon \Imu)^{-1/2} (\varphi_{t_j} \otimes \varphi_{t_j})({\cal G}_\mu + \epsilon \Imu)^{-1/2} \right)\\
&= \frac{\tsmu}{\ttmu(t_j)} \cdot \frac{1}{T} \tr\left( ({\cal G}_\mu + \epsilon \Imu)^{-1} (\varphi_{t_j} \otimes \varphi_{t_j})  \right)\\
&= \frac{\tsmu}{\ttmu(t_j)} \cdot \tmu(t_j)\tag{via Theorem \ref{thm:leverageProps}, equation \eqref{eq:InnerProdRep}.}\\
&\le \tsmu
	\end{align*}
	where the last line follows since $\ttmu(t_j) \ge \tmu(t_j)$ by assumption.

	The final ingredient for applying Lemma~\ref{lem:op-concentrate} is to bound ${\cal X}^2_j$. Again, using the fact that $\ttmu(t_j) \ge \tmu(t_j)$ we have:
	\begin{align*}
	{\cal X}^2_j =& \frac{\tsmu^2}{T^2 \cdot \ttmu(t_j)^2} ({\cal G}_\mu + \epsilon \Imu)^{-1/2} (\varphi_{t_j} \otimes \varphi_{t_j})({\cal G}_\mu + \epsilon \Imu)^{-1} (\varphi_{t_j} \otimes \varphi_{t_j})({\cal G}_\mu + \epsilon \Imu)^{-1/2} \\
	= & \frac{\tsmu^2 \cdot \langle \varphi_{t_j}, ({\cal G}_\mu + \epsilon\Imu)^{-1}\varphi_{t_j}\rangle_\mu}{T^2 \cdot \ttmu (t_j)^2 }
	({\cal G}_\mu + \epsilon\Imu)^{-1/2} (\varphi_{t_j} \otimes \varphi_{t_j}) ({\cal G}_\mu + \epsilon\Imu)^{-1/2} \\
	= & \frac{\tsmu^2 \cdot \tmu(t_j)}{T \cdot \ttmu (t_j)^2 }
	({\cal G}_\mu + \epsilon\Imu)^{-1/2} (\varphi_{t_j} \otimes \varphi_{t_j}) ({\cal G}_\mu + \epsilon\Imu)^{-1/2} \\
	= & \frac{\tsmu \cdot \tmu(t_j)}{\ttmu (t_j)} {\cal X}_j \preceq \tsmu {\cal X}_j.
	\end{align*}
	So, using Claim~\ref{claim:dominated-expectation2},
	$$
	\E_{t_j \propto \ttmu}[{\cal X}^2_j] \preceq 	\E_{t_j \propto \ttmu}[\tsmu {\cal X}_j] = \tsmu ({\cal G}_\mu + \epsilon\Imu)^{-1/2} {\cal G}_\mu ({\cal G}_\mu + \epsilon\Imu)^{-1/2} \eqdef {\cal M}.
	$$
	Noticing that $\tr({\cal M}) = \tsmu \cdot \smu$ and $\opnorm{{\cal M}} = \frac{\opnorm{\Gmu}}{\opnorm{\Gmu} + \epsilon}\geq 1/2$ by our assumption that $\epsilon \le \opnorm{\Gmu}$, and  Lemma~\ref{lem:op-concentrate} we have:
	\begin{align*}
	\Pr\left(\opnorm{({\cal G}_\mu + \epsilon\Imu)^{-1/2}(\widehat{\cal G}_\mu  -  {\cal G}_\mu)({\cal G}_\mu + \epsilon\Imu)^{-1/2}} \geq \Delta \right)
	\leq& \frac{8\tr({{\cal M}})}{\opnorm{\cal M}}\exp\left(\frac{-s\Delta^2/2}{\opnorm{{\cal M}} + 2 \tsmu\Delta / 3}\right) \\
	\leq& 16\tsmu \cdot \smu\cdot\exp\left(\frac{-s\Delta^2}{ 2 \tsmu(1 +2\Delta / 3)}\right) \\
	\leq& 16\tsmu^2\cdot\exp\left(\frac{-3s\Delta^2}{ 8 \tsmu}\right) \leq \delta.
	\end{align*}
	
%		Otherwise, $\opnorm{{\cal M}} = \opnorm{\Gmu} / (\opnorm{\Gmu} + \epsilon)\geq \opnorm{\Gmu}/(2\epsilon)$, and
%		$\smu \le \tr(\Gmu)/\epsilon \le $.
%		 Applying Lemma~\ref{lem:op-concentrate}, we have
%	\begin{align*}
%	\Pr\left(\opnorm{({\cal G}_\mu + \epsilon\Imu)^{-1/2}(\widehat{\cal G}_\mu  -  {\cal G}_\mu)({\cal G}_\mu + \epsilon\Imu)^{-1/2}} \geq \Delta \right)
%	\leq& \frac{8\tr({{\cal M}})}{\opnorm{\cal M}}\exp\left(\frac{-s\Delta^2/2}{\opnorm{{\cal M}} + 2 \tsmu\Delta / 3}\right) \\
%	\leq& \frac{16 \epsilon}{\opnorm{\Gmu}}\tsmu \cdot \smu\cdot\exp\left(\frac{-s\Delta^2}{ 2 \tsmu(1 +2\Delta / 3)}\right) \\
%	\leq& 16\tsmu^2\cdot\exp\left(\frac{-3s\Delta^2}{ 8 \tsmu}\right) \leq \delta.
%	\end{align*}		
\end{proof}

\subsection{Approximate Discretization via Leverage Score Sampling}

With the operator approximation bound of Lemma \ref{lem:operatorApproximation} in place, we can prove Theorem \ref{thm:baseSampling}, which shows that we can approximately  solve  the regression problem of \eqref{eq:least_squares_setup} (and by Claim \ref{claim:regression_reduction} solve Problem \ref{prob:unformal_interp}) by sampling time domain points via over-approximations to their ridge leverage scores.

\begin{reptheorem}{thm:baseSampling}[Approximate Regression via Leverage Score Sampling]
Assume $\epsilon \le \opnorm{\Kmu}$ and
	consider a measurable $\ttmu(t)$ with $\ttmu(t)  \ge \tmu(t)$ for all $t$ and let $\tsmu  = \int_0^T \ttmu(t) dt$. %Assume that $\tsmu < \infty$.
	Let $s = c \cdot \left(\tsmu \cdot [\log (\tsmu) + 1/\delta]\right)$  for sufficiently large fixed constant $c$ and let $t_1,\ldots,t_s$ be time points selected by drawing each randomly from $[0,T]$  with probability proportional to $\ttmu(t)$. For $j \in 1,\ldots,s$ let $w_j =  \sqrt{\frac{1}{sT} \cdot \frac{\tsmu}{\ttmu(t_j)}} $. Let $\bv{F}: \CC^s \rightarrow L_2(\mu)$ be the operator defined by:
	\begin{align}\label{eq:fdef}
		\left[\bv{F} \,\bv{x}\right](\xi) = \sum_{j=1}^s w_j \cdot \bv{x}(j) \cdot e^{ - 2\pi i \xi t_j}
	\end{align}
	and $\bv{y},\bv{n} \in \RR^s$ be the vector with $\bv{y}(j) = w_j \cdot y(t_j)$ and $\bv{n}(j) = w_j \cdot n(t_j)$.
	With probability $\ge 1- \delta$:
	\begin{itemize}
	\item For any $\beta \geq 0$, if $\tilde{g}\in L_2(\mu)$ satisfies \footnote{We can see that the adjoint $\bv{F}^*: L_2(\mu) \rightarrow \CC^s$ is given by $\left[\bv{F}^* \,g\right](j) = w_j \cdot \int_{\RR} g(\xi) e^{2\pi i \xi t_j}\, d\mu(\xi)$.}
	\begin{align}
		\|\bv{F}^* \tilde{g} - (\bv y + \bv n)\|_2^2 + \epsilon\|\tilde{g}\|_\mu^2 \leq (1+\delta \beta)\cdot\min_{g \in L_2(\mu)}\left[  \|\bv{F}^* g - (\bv y+\bv n)\|_2^2 + \epsilon\|g\|_\mu^2\right]\label{eq:approxkrrRep},
	\end{align}
	then
	\begin{align}
		\|\Fmu^* \tilde{g} - (y+n)\|_T^2 + \epsilon\|\tilde{g}\|_\mu^2 \leq 3(1+2\beta)\cdot\min_{g \in L_2(\mu)}\left[  \|\Fmu^* g - (y+n)\|_T^2 + \epsilon\|g\|_\mu^2\right].
	\end{align}
	\end{itemize}
So $\tilde{g}$ provides an approximate solution to \eqref{eq:least_squares_setup} and by Claim \ref{claim:regression_reduction}, $\tilde{y} = \Fmu^* \tilde{g}$ solves Problem \ref{prob:unformal_interp} with parameter $\Theta(\epsilon)$.
\end{reptheorem}

\begin{proof}
Throughout the proof we will let $\bar y = y + n$ and $\bv{\bar y} = \bv{y} + \bv{n}$.
Let
$$g^\star \eqdef \argmin_{g\in L_2(\mu)}\left[  \|\Fmu^* g - \bar y\|_T^2 + \epsilon\|g\|_\mu^2\right].$$
By Lemma \ref{lem:ridgeMinimizer}, $g^\star = \Fmu (\Kmu + \lambda \mathcal{I}_T)^{-1} \bar y$.
Denote the optimal error as $b^\star \eqdef  \Fmu^* g^\star - \bar y $ and the optimal cost as $B^\star \eqdef  \|\Fmu^* g^\star - \bar y\|_T^2 + \epsilon\| g^\star\|_\mu^2$.

\subsubsection*{Reduction to Affine Embedding}

We prove that, for all $g \in L_2(\mu)$, ridge leverage score sampling lets us approximate the value of the objective function of \eqref{eq:least_squares_setup} when evaluated at $g$. In the randomized linear algebra literature, this is known as an \emph{affine embedding guarantee}. Specifically, we show that, with probability $\ge 1-\delta$, for all $g \in L_2(\mu)$,
\begin{align}\label{eq:affineEmbedding}
\frac{1}{2} \left (\norm{\Fmu^* g - \bar y}_T^2 + \epsilon \norm{g}_\mu^2\right)  \le \norm{\bv{F}^* g - \bv{\bar y}}_2^2 + \epsilon \norm{g}_\mu^2+ \alpha  \le \frac{3}{2} \left (\norm{\Fmu^* g - \bar y}_T^2 + \epsilon\norm{g}_\mu^2 \right)
\end{align}
where $\alpha$ is some fixed value independent of $g$ (but which depends on $\bv{F}$ and $\bv{\bar y}$) with $|\alpha| \le \frac{1}{\delta} \cdot  B^\star$.

It is not hard to see that \eqref{eq:affineEmbedding} gives the theorem. For any $\tilde g \in L_2(\mu)$ satisfying:
	\begin{align}\label{eq:approxDiscrete}
		\|\bv{F}^* \tilde{g} - \bv{\bar y}\|_2^2 + \epsilon\|\tilde{g}\|_\mu^2 \leq (1+\delta C)\cdot\min_{g \in L_2(\mu)}\left[  \|\bv{F}^* g - \bv{\bar y}\|_2^2 + \epsilon\|g\|_\mu^2\right],
	\end{align}
we can apply \eqref{eq:affineEmbedding} to give the main claim of the theorem:
\begin{align*}
\|\Fmu^* \tilde{g} - \bar y\|_T^2 + \epsilon\|\tilde{g}\|_\mu^2 &\le 2 \left (\norm{\bv{F}^* \tilde g - \bv{\bar y}}_2^2 + \epsilon \norm{\tilde g}_\mu^2 + \alpha \right )\tag{applying lower bound of \eqref{eq:affineEmbedding}}\\
& \le 2(1+\delta C) \cdot \min_{g \in L_2(\mu)}\left(  \|\bv{F}^*g - \bv{\bar y}\|_2^2 + \epsilon\|g\|_\mu^2 \right) + 2\alpha \tag{by assumption of \eqref{eq:approxDiscrete}}\\
& \le 2(1+\delta C) \cdot \left (\|\bv{F}^*  g^\star - \bv{\bar y}\|_2^2 + \epsilon\| g^\star\|_\mu^2\right) + 2\alpha\\
&= 2(1+\delta C) \cdot \left (\|\bv{F}^*  g^\star - \bv{\bar y}\|_2^2 + \epsilon\| g^\star\|_\mu^2 + \alpha \right) - 2\delta C\alpha\\
& \le 3(1+\delta C) \cdot \left(\|\Fmu^*  g^\star - \bar y\|_F^2 + \epsilon\| g^\star\|_\mu^2\right) - 2\delta C\alpha \tag{upper bound of \eqref{eq:affineEmbedding}}\\
& \le [3(1+\delta C) + 2C] \cdot \left(\|\Fmu^*  g^\star - \bar y\|_F^2 + \epsilon\| g^\star\|_\mu^2\right)\tag{since $|\alpha| \le \frac{ B^\star}{\delta}$}\\
&\le 3(1+2C) \cdot\min_{g \in L_2(\mu)}\left[  \|\Fmu^* g - \bar y\|_T^2 + \epsilon\|g\|_\mu^2\right]. \tag{definition of $ g^\star$ as optimum}
\end{align*}
Thus, we focus our attention to proving that the affine embedding guarantee of \eqref{eq:affineEmbedding} holds with probability $\ge 1-\delta$.

\subsubsection*{Expression of Error in Terms of $g- g^\star$}

We begin by showing how, for any $g \in L_2(\mu)$, the cost $\norm{\Fmu^* g -\bar  y}_T^2 + \epsilon \norm{g}_\mu^2$ can be written as a function of the deviation from the optimum: $g- g^\star$.
\begin{claim}[Expression for Excess Cost]\label{clm:affineDecomp}
For any $g \in L_2(\mu)$:
\begin{align*}
\norm{\Fmu^*g -\bar y}_T^2 + \epsilon \norm{g}_\mu^2  = \norm{\Fmu^* (g- g^\star)}_T^2 + \epsilon \norm{g -  g^\star}_\mu^2 +  B^\star,
\end{align*}
recalling that $ B^\star \eqdef \norm{\Fmu^*  g^\star - \bar y}_T^2 + \epsilon \norm{ g^\star}_\mu^2$ is the optimum cost of the ridge regression problem.
\end{claim}
\begin{proof} Following Lemma \ref{lem:ridgeMinimizer} we define $\mathcal{T}:L_2(\mu) \to L_2(\mu) \times L_2(T),\quad \mathcal{T}g = (\sqrt{\epsilon} g,\Fmu^*g)$. For any $g$, $\norm{\Fmu^* g - \bar y}_T^2 + \epsilon \norm{g}_\mu^2 = \norm{\mathcal{T} g - (0,\bar y)}_{ L_2(\mu) \times L_2(T)}^2$. Again, as in Lemma \ref{lem:ridgeMinimizer} we know $g^\star$ is the unique minimizer of this function with the property that $(0,y) - \mathcal{T} g^\star \perp \range(\mathcal T)$~\cite[Theorem 6.13]{HunterNachtergaele:2001}.
We can thus decompose:
\begin{align*}
\norm{\Fmu^* g - \bar y}_T^2 + \epsilon \norm{g}_\mu^2 &= \norm{\mathcal{T} g - (0,\bar y)}_{ L_2(\mu) \times L_2(T)}^2\\
&= \norm{\mathcal{T} g^\star - (0,\bar y) + (\mathcal{T} g - \mathcal{T} g^\star)}_{ L_2(\mu) \times L_2(T)}^2\\
&=  \norm{\mathcal{T} g^\star - (0,\bar y)}_{ L_2(\mu) \times L_2(T)}^2 + \norm{\mathcal{T} (g - g^\star)}_{ L_2(\mu) \times L_2(T)}^2\\
& = B^\star + \norm{\Fmu^*(g- g^\star)}_T^2 + \epsilon \norm{g- g^\star}_\mu^2
% \norm{\Fmu^*  g^\star - \bar y + \Fmu^*(g- g^\star)}_T^2 + \epsilon \norm{ g^\star + (g- g^\star)}_\mu^2\\
%&=  B^\star + \norm{\Fmu^*(g- g^\star)}_T^2 + \epsilon \norm{g- g^\star}_\mu^2 + 2 \left (\langle  \Fmu^*(g- g^\star), \Fmu^*  g^\star - \bar y\rangle_T + \epsilon \langle (g- g^\star), g^\star\rangle_\mu\right )\\
%& =  B^\star + \norm{\Fmu^*(g- g^\star)}_T^2 + \epsilon \norm{g- g^\star}_\mu^2 + 2 \langle (g- g^\star), \Gmu  g^\star - \Fmu \bar y + \epsilon  g^\star \rangle_\mu \\
%& =  B^\star + \norm{\Fmu^*(g- g^\star)}_T^2 + \epsilon \norm{g- g^\star}_\mu^2
\end{align*}
which gives the claim.
\end{proof}

\subsubsection*{Bounding The Sampling Error}
We now show that Claim \ref{clm:affineDecomp} holds approximately, even after sampling. This almost immediately  yields the affine embedding bound of \eqref{eq:affineEmbedding}.

Let $\tilde B \eqdef \norm{\bv{F}^*  g^\star - \bv{\bar y}}_2^2 + \epsilon \norm{ g^\star}_\mu^2$ be the error of the optimal solution in our randomly discretized regression problem.
We can write the discretized objective function value for any $g \in L_2(\mu)$ as:
\small
\begin{align}
\norm{\bv F^* g - \bv{\bar y}}_2^2 + \epsilon \norm{g}_\mu^2 &= \norm{\bv F^* (g- g^\star) + \bv F^*  g^\star - \bv{\bar y}}_2^2 + \epsilon \norm{ g^\star + (g- g^\star)}_\mu^2\nonumber\\
& = \tilde B + \norm{\bv F^*(g -  g^\star)}_2^2 + \epsilon \norm{g- g^\star}_\mu^2 + 2 \Re(\langle \bv{F}^*(g -  g^\star), \bv{F}^* g^\star - \bv{\bar y} \rangle) + 2\epsilon \Re(\langle (g- g^\star),  g^\star \rangle_\mu) \label{eq:break1}.
\end{align}
\normalsize

Let $\bFmu:  L_2(T) \times L_2(\mu) \rightarrow L_2(\mu)$ be the operator $\bFmu(f,g) = \Fmu f + \sqrt{\epsilon} \cdot g$. We can see that $\bFmu^*:  L_2(\mu) \rightarrow L_2(T) \times L_2(\mu)$ is given by $\bFmu^* g = (\Fmu^* g, \sqrt{\epsilon} \cdot g)$. Further, we see that $\bFmu \bFmu^* = \Gmu + \epsilon \Imu$. We can write:
\begin{align*}
\bFmu^* = \bFmu^* (\Gmu + \epsilon \Imu)^{-1} (\Gmu + \epsilon \Imu) = \bPmu \bFmu^*
\end{align*}
where $\bPmu = \bFmu^* (\Gmu + \epsilon \Imu)^{-1} \bFmu$. Note that $\bPmu$ is self adjoint. 
Correspondingly, let $\bv{\bar F}: \CC^s \times L_2(\mu) \rightarrow L_2(\mu)$ be given by $\bv{\bar F}(f,g) = \bv{F} f + \sqrt{\epsilon} \cdot g$. We have $\bv{\bar F}^*g = (\bv{F}^* g, \sqrt{\epsilon} \cdot g)$. We can also write $\bv{\bar P} = \bv{\bar F}^* (\Gmu + \epsilon \Imu)^{-1}\bFmu$, and observe that $\bv{\bar F}^* = \bv{\bar P} \bFmu^*$.

With this notation in place we can rewrite the last term of \eqref{eq:break1} as:
\begin{align}\label{eq:break11}
 \langle \bv{F}^*(g- g^\star), \bv{F}^* g^\star - \bv{\bar y} \rangle + \epsilon \langle (g- g^\star),  g^\star \rangle_\mu & = \langle \bv{ \bar F}^*(g- g^\star) , (\bv{ F^*}  g^\star - \bv{\bar y},\sqrt{\epsilon}   g^\star) \rangle_{\CC^s \times L_2(\mu)}\nonumber\\
&= \langle \bv{ \bar P} \bFmu^*(g- g^\star) ,(\bv{ F^*}  g^\star - \bv{\bar y},\sqrt{\epsilon}   g^\star) \rangle_{\CC^s \times L_2(\mu)}\nonumber \\
&=  \langle \bFmu^*(g- g^\star) , \bv{\bar P}^* (\bv{ F^*}  g^\star - \bv{\bar y},\sqrt{\epsilon}   g^\star)  \rangle_{L_2(T) \times L_2(\mu)}.
\end{align}
Using the fact that $\Re(z)\leq |z|$ for all $z\in\CC$, and applying Cauchy-Schwarz to \eqref{eq:break11} and plugging back into \eqref{eq:break1} we have:
\begin{align}
\norm{\bv F^* g - \bv{\bar y}}_2^2 + \epsilon \norm{g}_\mu^2 &\in  \tilde B + \norm{\bv F^*(g- g^\star)}_2^2 + \epsilon \norm{g- g^\star}_\mu^2  \nonumber\\
&\hspace{2em} \pm 2(\norm{\Fmu^*(g- g^\star)}_T + \epsilon \norm{g- g^\star}_\mu )\cdot \norm{\bv{\bar P}^* (\bv{F}^*  g^\star - \bv{\bar y},\sqrt{\epsilon}   g^\star)}_{L_2(T) \times L_2(\mu)}\label{eq:break12}.
\end{align}
We now bound $\norm{\bv{\bar P}^* (\bv{F}^*  g^\star - \bv{\bar y},\sqrt{\epsilon}   g^\star)}_{L_2(T) \times L_2(\mu)}$. \emph{If we had not sampled}, this would equal:
\begin{align}\label{eq:nonSampled}
\norm{\bPmu (\Fmu^*  g^\star - \bar y, \sqrt{\epsilon}   g^\star)}_{L_2(T) \times L_2(\mu)} = \norm{\bFmu^* (\Gmu + \epsilon \Imu)^{-1} \bFmu \left [\bFmu^*  g^\star - (\bar y,0)\right ]}_{L_2(T) \times L_2(\mu)} = 0
\end{align}
since $g^\star$ is the optimum of $\norm{\bFmu^* g - (\bar y,0)}_{L_2(T) \times L_2(\mu)}$ and thus $\bFmu^* g^\star - (\bar y,0)$ is orthogonal to $\range(\bFmu^*)$.
We will show that after sampling, while the norm no longer equals $0$, it is still small. The bound we give is analogous to standard approximate matrix multiplication results for finite dimensional matrices. Specifically, our proof follows that of Lemma 4 in \cite{DrineasKannanMahoney:2006}.
\begin{claim}[Approximate Operator Application]\label{clm:approxMatrixMult}
With probability $\ge 1-\delta$:
\begin{align*}
\norm{\bv{\bar P}^* (\bv{F}^*  g^\star - \bv{\bar y},\sqrt{\epsilon}   g^\star)}_{L_2(T) \times L_2(\mu)} \le \frac{1}{64} \cdot
 B^\star.
\end{align*}
\end{claim}
\begin{proof}
For conciseness let $\mathcal{H}$ denote the space $L_2(T) \times L_2(\mu)$. Let $\varphi_t \in L_2(\mu)$ be given by $\varphi_t(\xi) = e^{-2\pi i t \xi}$.
%Let $\bv{F}_j = w_j \cdot \varphi_{t_j}$ so we have:
%\begin{align}\label{eq:fbreakdown}
%\left[\bv{F} \,g\right](\xi) = \sum_{j=1}^s \cdot g(j) \bv{F}_j.
%\end{align}
% $\bv{\bar F}_j: \CC \times L_2(\mu) \rightarrow L_2(\mu)$ be given  by $\bv{\bar F}_j(f,g) = f \cdot \bv{F}_j + \sqrt{\epsilon} g$ and let $\bv{\bar P}_j = \bv{\bar F}_j^* (\Gmu + \epsilon \Imu)^{-1} \bFmu$.
Let $  b^\star \eqdef \Fmu^*  g^\star -\bar y$ and $\bv{b}^\star \in \CC^s$ be given by $\bv{b}^\star \eqdef \bv{F}^* g^\star - \bv{\bar y}.$ We can see that $\bv{b}^\star(j) = w_j \cdot [\langle \varphi_{t_j}, g^\star\rangle_\mu - \bar y(t_j)]$.
We have: 
\begin{align}
\E \left [\norm{\bv{\bar P}^* (\bv{F}^*  g^\star - \bv{\bar y},\sqrt{\epsilon}   g^\star)}_{\mathcal{H}}^2 \right ] &=\E \left [\norm{\bv{\bar P}^*(\bv{b}^\star,\sqrt{\epsilon}   g^\star)}_{\mathcal{H}}^2 \right ] \nonumber\\
&=  \E \left [\norm{\bv{\bar P}^* (\bv{b}^\star,\sqrt{\epsilon}   g^\star) - \bFmu^* (\Gmu + \epsilon \Imu)^{-1} \bFmu \left [\bFmu^*  g^\star - (\bar y,0)\right ]}_{\mathcal{H}}^2 \right ]\nonumber\tag{Since by \eqref{eq:nonSampled}, $\norm{\bFmu^* (\Gmu + \epsilon \Imu)^{-1} \bFmu \left [\bFmu^*  g^\star - (\bar y,0)\right ]}_{\mathcal{H}} = 0$.}\\
&= \E \left [\norm{\bv{\bar P}^* (\bv{b}^\star,\sqrt{\epsilon}   g^\star) - \bFmu^* (\Gmu + \epsilon \Imu)^{-1} \bFmu \left (b^\star, \sqrt{\epsilon} g^\star \right )}_{\mathcal{H}}^2 \right ]\nonumber\tag{Since $\bFmu^* g^\star = (\Fmu^* g, \sqrt{\epsilon} g)$ and since by definition $b^\star = \Fmu^* g^\star - \bar y$, giving $\left [\bFmu^*  g^\star - (\bar y,0)\right ]= (b^\star,\sqrt{\epsilon}g^\star).$}\\
& = \E \left [[\norm{\bFmu^* (\Gmu+\epsilon \Imu)^{-1} \left ( \bv{\bar F} (\bv{b}^\star,\sqrt{\epsilon}   g^\star) - \bFmu \left (b^\star, \sqrt{\epsilon} g^\star \right )\right )}_{\mathcal{H}}^2 \right  ]\nonumber \tag{Factoring $\bFmu^* (\Gmu+\epsilon \Imu)^{-1}$ out of $\bv{\bar P}^* = \bFmu^* (\Gmu+\epsilon \Imu)^{-1} \bv{\bar F}.$}\\
& = \E \left [[\norm{\bFmu^* (\Gmu+\epsilon \Imu)^{-1} \left ( \bv{ F} \bv{b}^\star- \Fmu b^\star \right )}_{\mathcal{H}}^2 \right  ]\nonumber \tag{Recalling that $\bv{\bar F}(f,g) = \bv{F} f + \sqrt{\epsilon} g$ and similarly $\bFmu(f,g) = \Fmu f + \sqrt{\epsilon} g$.}\\
&= \E \left [\left \| \bFmu^* (\Gmu+\epsilon \Imu)^{-1} \sum_{i=1}^s \left ( \varphi_{t_j} \cdot w_j \cdot \bv{b}^\star(j) - \frac{1}{s} \Fmu b^\star \right ) \right \|_\mathcal{H}^2 \right  ]\label{eq:integralExpandNew},
\end{align}
where the last equality follows since by \eqref{eq:fdef}, for any $\bv{x} \in \CC^s$, $\bv{F}\bv{x} = \sum_{j=1}^s \varphi_{t_j} \cdot w_j \cdot \bv{x}(j)$. To simplify \eqref{eq:integralExpandNew} we first bound, for any $g \in L_2(\mu)$, $\E \left [\langle g,  \varphi_{t_j} \cdot w_j \cdot \bv{b}^\star(j) \rangle_{\mu} \right ]$, recalling that $\bv{b}^\star(j) = w_j \cdot [\langle \varphi_{t_j}, g^\star \rangle_\mu - \bar y (t_j)]$.
Let $p(t) = \frac{\ttmu(t)}{\tsmu}$ be the density  with which we sample our time points $t_1,\ldots,t_s$ and $w(t) = \sqrt{\frac{1}{s T \cdot p(t)}}$ be the reweighting factor we apply if we sample time $t$ (so $w_j = w(t_j)$). 

First we argue that we can apply Fubini's theorem to switch the order of the double integration in $\E \left [\langle g,  \varphi_{t_j} \cdot w_j \cdot \bv{b}^\star(j) \rangle_{\mu} \right ]$ (over random instantiations of $\varphi_{t_j} \cdot w_j \cdot \bv{b}^\star(j)$ and within the inner product). Letting for $z \in L_2(\mu)$, $|z| \in L_2(\mu)$ be given by $|z|(\eta) = |z(\eta)|$ we have:
\begin{align*}
\E \left [\langle |g|,  |\varphi_{t_j} \cdot w_j \cdot \bv{b}^\star(j) | \rangle_{\mu} \right ] \le \norm{g}_\mu \cdot \E \left [ \norm{\varphi_{t_j} w_j \bv{b}^\star(j)}_\mu \right ],
\end{align*}
which, noting that $\norm{\varphi_{t_j}}_\mu = 1$ gives:
\begin{align*}
\E \left [\langle g,  \varphi_{t_j} \cdot w_j \cdot \bv{b}^\star(j) \rangle_{\mu} \right ] &\le \norm{g}_\mu \cdot \E \left [ |w_j \bv{b}^\star(j) |\right ]\\
&= \norm{g}_\mu \cdot \int_{0}^T |\langle \varphi_t, g^\star \rangle_\mu - y(t) | w(t)^2 \cdot p(t)\, dt\\
& = \norm{g}_\mu \cdot \frac{1}{sT} \int_{0}^T |\langle \varphi_t, g^\star \rangle_\mu - y(t)| \, dt\\
& < \infty
\end{align*}
where the last line follows since $g \in L_2(\mu)$ so $\norm{g}_\mu < \infty$ and since $\frac{1}{T} \int_{0}^T |\langle \varphi_t, g^\star \rangle_\mu - y(t)| \, dt \le \frac{1}{T} \int_{0}^T \left (|\langle \varphi_t, g^\star \rangle_\mu - y(t)|^2 + 1\right ) \, dt  = \norm{\Fmu^* g^\star - y}_T^2 + T \le \norm{y}_T^2 < \infty$. Since we have established that $\E \left [\langle |g|,  |\varphi_{t_j} \cdot w_j \cdot \bv{b}^\star(j)| \rangle_{\mu} \right ] $ is finite we can apply Fubini's theorem to compute:
\begin{align}\label{eq:dotExp}
\E \left [\langle g,  \varphi_{t_j} \cdot w_j \cdot \bv{b}^\star(j) \rangle_{\mu} \right ] &= \int_{0}^T [\langle \varphi_t, g^\star \rangle_\mu - y(t)] w(t)^2 \cdot \langle g, \varphi_t \rangle_{\mu} \cdot p(t)\, dt\nonumber \\
%= \frac{1}{sT} \left \langle g, \int_{0}^T b^\star(t) \cdot \varphi_t  \, dt \right \rangle_\mu
& = \frac{1}{sT} \int_0^T \left (b^\star(t)  \cdot \int_{\xi \in \RR} g(\xi)^* e^{-2\pi i \xi t} d\mu(\xi)\right) dt\nonumber\\
& = \frac{1}{s} \int_{\xi \in \RR} \left (g(\xi)^* \cdot \frac{1}{T}\int_0^T e^{-2\pi i \xi t} b^\star(t) dt\right )d\mu(\xi)\nonumber\\
& = \frac{1}{s}\langle g, \Fmu b^\star\rangle_\mu.
\end{align}
%The application of Fubini's theorem in the third line is valid since for any $j$, $w_j = \sqrt{\frac{1}{sT \cdot p(t_j)}} < \infty$ with probability $1$, $\bv{b}^\star(j) blah$ 
This in turn gives that 
\begin{align*}%\label{eqExp}
\E \left [\langle g,  \varphi_{t_j} \cdot w_j \cdot \bv{b}^\star(j) - \frac{1}{s} \Fmu b^\star \rangle_{\mu} \right ] = 0
\end{align*}
and so for any $g \in L_2(\mu)$:
\begin{align}\label{eq:integralExpandNew2}
\E \left [ \left \langle \bFmu^* (\Gmu+\epsilon \Imu)^{-1} g,   \bFmu^* (\Gmu+\epsilon \Imu)^{-1}\left (\varphi_{t_j} \cdot w_j \cdot \bv{b}^\star(j) - \frac{1}{s} \Fmu b^\star\right ) \right \rangle_{\mathcal{H}} \right ] &=\nonumber\\
&\hspace{-25em}\E \left [ \left \langle(\Gmu+\epsilon \Imu)^{-1}\bFmu \bFmu^* (\Gmu+\epsilon \Imu)^{-1} g,  \varphi_{t_j} \cdot w_j \cdot \bv{b}^\star(j) - \frac{1}{s} \Fmu b^\star \right \rangle_{\mu} \right ] = 0.
\end{align}
Further, since $t_1,\ldots,t_s$ are independent, the above gives that for $j \neq k$:
\small
\begin{align}\label{eq:integralExpandNew3}
\E \left [ \left \langle \bFmu^* (\Gmu+\epsilon \Imu)^{-1} \left (\varphi_{t_j} \cdot w_j \bv{b}^\star(j) - \frac{1}{s} \Fmu b^\star\right ) ,   \bFmu^* (\Gmu+\epsilon \Imu)^{-1}\left (\varphi_{t_tk} \cdot w_k \bv{b}^\star(k) - \frac{1}{s} \Fmu b^\star\right ) \right \rangle_{\mathcal{H}} \right ] = 0. 
\end{align}
\normalsize
We can apply \eqref{eq:integralExpandNew2} and \eqref{eq:integralExpandNew3} to expand out \eqref{eq:integralExpandNew}, giving:
\small
\begin{align}
\E \left [\norm{\bv{\bar P}^* (\bv{F}^*  g^\star - \bv{\bar y},\sqrt{\epsilon}   g^\star)}_{\mathcal{H}}^2 \right ]=\nonumber\\
&\hspace{-10em}\sum_{j=1}^s \sum_{k=1}^s \E \left [\left \langle \bFmu^* (\Gmu+\epsilon \Imu)^{-1} \left (\varphi_{t_j} \cdot w_j \cdot \bv{b}^\star(j) - \frac{1}{s}\Fmu b^\star\right), \bFmu^* (\Gmu+\epsilon \Imu)^{-1} \left (\varphi_{t_k} \cdot w_k \cdot \bv{b}^\star(k) - \frac{1}{s}\Fmu b^\star\right)\right \rangle_{\mathcal{H}}\right ]\nonumber\\
&\hspace{-9em}= \sum_{j=1}^s \E \left [ \left \langle \bFmu^* (\Gmu+\epsilon \Imu)^{-1} \left(\varphi_{t_j} \cdot w_j \cdot \bv{b}^\star(j) - \frac{1}{s}\Fmu b^\star\right ), \bFmu^* (\Gmu+\epsilon \Imu)^{-1}\left ( \varphi_{t_j} \cdot w_j \cdot \bv{b}^\star(j) - \frac{1}{s}\Fmu b^\star\right)\right \rangle_{\mathcal{H}}\right]\nonumber\tag{since cross terms are $0$ via \eqref{eq:integralExpandNew3}}\nonumber\\
&\hspace{-9em}= \sum_{j=1}^s \E \left [ \left \langle \bFmu^* (\Gmu+\epsilon \Imu)^{-1} \left(\varphi_{t_j} \cdot w_j \cdot \bv{b}^\star(j) - \frac{1}{s}\Fmu b^\star\right ), \bFmu^* (\Gmu+\epsilon \Imu)^{-1}\left ( \varphi_{t_j} \cdot w_j \cdot \bv{b}^\star(j)\right)\right \rangle_{\mathcal{H}}\right]\tag{applying \eqref{eq:integralExpandNew2} to $-\frac{1}{s} \Fmu b^\star$}\nonumber\\
&\hspace{-9em}=\sum_{i=1}^s \E \big [ \left \|\bFmu^* (\Gmu + \epsilon \Imu)^{-1} \varphi_{t_j} \cdot w_j \cdot \bv{b}^\star(j)\right \|_\mathcal{H}^2 - \frac{1}{s} \left \langle\bFmu^* (\Gmu + \epsilon \Imu)^{-1}  \varphi_{t_j} \cdot w_j \cdot \bv{b}^\star(j),\bFmu^* (\Gmu + \epsilon \Imu)^{-1}  \Fmu b^\star\right \rangle_{\mathcal{H}}  \big]\nonumber\\
&\hspace{-9em}=\sum_{i=1}^s \E \big [ \left \|\bFmu^* (\Gmu + \epsilon \Imu)^{-1} \varphi_{t_j} \cdot w_j \cdot \bv{b}^\star(j)\right \|_\mathcal{H}^2 - \frac{1}{s^2} \norm{\bFmu^* (\Gmu + \epsilon \Imu)^{-1} \Fmu b^\star}_{\mathcal{H}}^2\nonumber\\
&\hspace{-9em}\le \sum_{i=1}^s \E \big [ \left \|\bFmu^* (\Gmu + \epsilon \Imu)^{-1} \varphi_{t_j} \cdot w_j \cdot \bv{b}^\star(j)\right \|_\mathcal{H}^2\label{eq:expBoundNow}
\end{align}
\normalsize
where the second to last line follows from \eqref{eq:dotExp} which gives
\small
\begin{align*}
\E \left [\left \langle\bFmu^* (\Gmu + \epsilon \Imu)^{-1}  \varphi_{t_j} \cdot w_j \cdot \bv{b}^\star(j),\bFmu^* (\Gmu + \epsilon \Imu)^{-1}  \Fmu b^\star\right \rangle_{\mathcal{H}} \right ]\\
&\hspace{-7em}= \E \left [ \left \langle  \varphi_{t_j} \cdot w_j \cdot \bv{b}^\star(j),(\Gmu + \epsilon \Imu)^{-1}\bFmu \bFmu^* (\Gmu + \epsilon \Imu)^{-1}  \Fmu b^\star\right \rangle_{\mu} \right ]\\
&\hspace{-7em}= \frac{1}{s} \left \langle  \Fmu b^\star,(\Gmu + \epsilon \Imu)^{-1}\bFmu \bFmu^* (\Gmu + \epsilon \Imu)^{-1}  \Fmu b^\star\right \rangle_{\mu}\\
&\hspace{-7em}=  \frac{1}{s}\norm{\bFmu^* (\Gmu + \epsilon \Imu)^{-1} \Fmu b^\star}_{\mathcal{H}}^2.
\end{align*}
\normalsize
%
%Combined with  \eqref{eq:dotExp} we can rewrite \eqref{eq:integralExpandNew} as: \Haim{Can you explain how you got the second equality. As a side, since we are dealing with complex numbers then note that when you switch sides you have to conjugate so instead of $2 \cdot \langle \ldots \rangle$ you get $2\Re(\langle \ldots \rangle)$ (already fixed one error like this), so the second equality might actually need to be an inequality).}
%\begin{align}\label{eq:unbiased}
%\E \left [\norm{\bv{\bar P}^* (\bv{F}^*  g^\star - \bv{\bar y},\sqrt{\epsilon}   g^\star)}_{\mathcal{H}}^2 \right ]  &= \sum_{i=1}^s \E \left [ \left \|\bFmu^* (\Gmu + \epsilon \Imu)^{-1}  \bv{F}_j \cdot \bv{b}_j^\star - \frac{1}{s} \Fmu b^\star\right \|_\mathcal{H}^2 \right] \nonumber\\
%&= \sum_{i=1}^s \E \left [ \left \|\bFmu^* (\Gmu + \epsilon \Imu)^{-1} \bv{F}_j \cdot \bv b_j^\star\right \|_\mathcal{H}^2 \right] - \norm{\bFmu^* (\Gmu + \epsilon \Imu)^{-1} \Fmu  b^\star}_\mathcal{H}^2\nonumber\\
%\end{align}
Given the bound of \eqref{eq:expBoundNow} we can now expand out, using the fact that time $t$ is sampled with probability  proportional to $\ttmu(t)$:
\begin{align*}
\E \left [\norm{\bv{\bar P}^* (\bv{F}^*  g^\star - \bv{\bar y},\sqrt{\epsilon}   g^\star)}_{\mathcal{H}}^2 \right ] &\le s \cdot \int_{t = 0}^T \frac{\ttmu(t)}{\tsmu} \cdot \left \|\bFmu^* (\Gmu +\epsilon \Imu)^{-1} \varphi_t \cdot  \frac{ (\langle \varphi_t, g^\star \rangle_\mu - \bar y(t))\cdot \tsmu}{s T \cdot  \ttmu(u)}\right \|_\mathcal{H}^2 dt\\
& = \frac{1}{sT^2} \cdot \int_{t = 0}^T \frac{\tsmu \cdot  b^\star(t)^2}{\ttmu(t) } \cdot \left \|\bFmu^* (\Gmu +\epsilon \Imu)^{-1} \varphi_t\right \|_\mathcal{H}^2 dt\\
& = \frac{1}{sT^2} \cdot \int_{t = 0}^T \frac{\tsmu \cdot  b^\star(t)^2}{\ttmu(t) } \cdot \langle \bFmu \bFmu^* (\Gmu +\epsilon \Imu)^{-1} \varphi_t, (\Gmu +\epsilon \Imu)^{-1}\varphi_t \rangle_\mu^2  dt\\
& = \frac{1}{sT^2} \cdot \int_{t = 0}^T \frac{\tsmu \cdot  b^\star(t)^2}{\ttmu(t) } \cdot \langle \varphi_t, (\Gmu +\epsilon \Imu)^{-1}\varphi_t \rangle_\mu^2  dt\tag{since $\bFmu \bFmu^* = \Gmu +\epsilon \Imu$}\\
& = \frac{1}{sT} \cdot \int_{t = 0}^T \frac{\tsmu \cdot  b^\star(t)^2 \cdot \tmu(t)}{\ttmu(t) } \tag{Theorem \ref{thm:leverageProps}, \eqref{eq:genInnerProd}}
\\&\le \frac{\tsmu\cdot \norm{ b^\star}_T^2}{s}. \tag{since by assumption $\ttmu(t) \ge  \tmu(t)$}
\end{align*}
Since $s = \Omega \left(\frac{\tsmu}{\delta}\right)$ we thus have via Markov's inequality, with probability $\ge 1-\delta$,
\begin{align*}
\norm{\bv{\bar P}^* (\bv{F}^*  g^\star - \bv{\bar y},\sqrt{\epsilon}   g^\star)}_{\mathcal{H}}^2 \le  \frac{1}{64} \cdot \norm{ b^\star}_T^2 \le \frac{1}{64} \cdot  B^\star
\end{align*}
which completes the claim. Note that $64$ is an arbitrarily chosen constant, which can be made as small as we want by increasing the sample size $s$ by a  constant factor.
\end{proof}

Plugging Claim \ref{clm:approxMatrixMult} back into \eqref{eq:break12} gives:
\begin{align*}
\norm{\bv F^* g - \bv{\bar y}}_2^2 + \epsilon \norm{g}_\mu^2 &\in \tilde B + \norm{\bv F^*(g- g^\star)}_2^2 + \epsilon \norm{g- g^\star}_\mu^2 \pm  \frac{1}{4}(\norm{\Fmu^*(g- g^\star)}_T+\epsilon \norm{g- g^\star}_\mu) \cdot \sqrt{ B^\star}\\
&\in \tilde B + \norm{\bv F^*(g- g^\star)}_2^2 + \epsilon \norm{g- g^\star}_\mu^2 \pm  \frac{1}{8}(\norm{\Fmu^*(g- g^\star)}_T+\epsilon \norm{g- g^\star}_\mu)^2 \pm \frac{1}{8}
 B^\star\\
&\in \tilde B + \norm{\bv F^*(g- g^\star)}_2^2 + \epsilon \norm{g- g^\star}_\mu^2 \pm  \frac{1}{4}(\norm{\Fmu^*(g- g^\star)}_T^2+\epsilon \norm{g- g^\star}_\mu^2) \pm \frac{1}{8}  B^\star.
\end{align*}
Applying the operator approximation bound of Lemma \ref{lem:operatorApproximation} with error $\Delta = 1/4$  then gives:
\begin{align*}
\norm{\bv F^* g - \bv{\bar y}}_2^2 + \epsilon \norm{g}_\mu^2 &\in \tilde B + \left (1\pm \frac{1}{2}\right )\left (\norm{\Fmu^* (g- g^\star)}_2^2 + \epsilon \norm{g- g^\star}_\mu^2 \right ) \pm \frac{1}{8}  B^\star.
\end{align*}
Finally, applying Claim \ref{clm:affineDecomp} gives:
\begin{align*}
\norm{\bv F^* g - \bv{\bar y}}_2^2 + \epsilon \norm{g}_\mu^2 &\in (\tilde B-  B^\star) + \norm{\Fmu^* g - \bar y}_T^2 + \epsilon \norm{g}_\mu^2 \pm \frac{1}{2} \left ( \norm{\Fmu^* g - \bar y}_T^2 + \epsilon \norm{g}_\mu^2 \right ).
\end{align*}
Note that $\E[ \tilde B] =  B^\star$. So writing $\alpha =  \tilde B -  B^\star$ we have $|\alpha| \le \frac{1}{\delta} \cdot  B^\star$ with probability $1-\delta$. This completes the theorem.
\end{proof}

%\subsection{Leverage Score Sampling Preserves Statistical Dimension}
%
%\begin{proof}[Proof of Claim \ref{clm:statDimBound}]
%\end{proof}

\subsection{Frequency Subset Selection}
We now prove the frequency subset selection guarantee Theorem \ref{thm:css} used in Section \ref{sec:general1} to bound the leverage scores for general constraints $\mu$, by showing that $\Fmu^*$ can be well approximated by an operator whose columns are spanned by just $O(\smu)$ frequencies.

 \begin{reptheorem}{thm:css}[Frequency Subset Selection]
For some $s \le \lceil 36 \cdot \smu\rceil$ there exists a set of distinct frequencies $\xi_1,\ldots,\xi_s \in \CC$ such that, letting $\bv{C}_s: L_2(T) \rightarrow \CC^{s} $ be defined by:
 \begin{align*}
 [\bv{C}_sg](j) = \frac{1}{T}\int_0^T g(t) e^{-2 \pi i \xi_j t}\, dt,
%F_s \bv{x}(t) = \sum_{i=1}^n x_i e^{2\pi i t \xi_i}
 \end{align*}
 %C^s -> L2(mu)
 and $\bv{Z} = (\bv{C}_s \bv{C}_s^*)^{-1} \bv{C}_s \Fmu^*$,
for $\varphi_t \in L_2(\mu), \bs{\phi}_t \in \CC^s$ with $\varphi_t(\xi) = e^{-2\pi i t \xi}$ and $\bs{\phi}_t(j) = \varphi_t(\xi_j)$:
\begin{align}\label{eq:frobNormBoundRep}
\frac{1}{T} \int_{t \in[0,T]}  \norm{\varphi_t - \bv Z^* \bs{\phi}_t}_\mu^2\, dt \le 4\epsilon \cdot \smu.
\end{align}
\end{reptheorem}
Our proof relies on the following spectral error bound for weighted frequency subset selection:
\begin{lemma}[Frequency Subset Selection -- Direct Spectral Approximation]\label{lem:cssSpectral}
	For some $s \leq \lceil 36 \cdot \smu\rceil$ there exists a set of  distinct frequencies $\xi_1,\ldots,\xi_s \in \RR$ and positive weights $w_1,\ldots,w_s \in \RR$ such that letting $\bv{\bar C}_s: L_2(T) \rightarrow \CC^{s} $ be given by:
 \begin{align*}
 [\bv{\bar C}_sg](j) = \frac{1}{T}\int_0^T g(t) w_j e^{-2 \pi i \xi_j t}\, dt,
%F_s \bv{x}(t) = \sum_{i=1}^n x_i e^{2\pi i t \xi_i}
 \end{align*}
	and letting $\widehat{\cal K}_\mu = \bv{\bar C}_s^* \bv{\bar C}_s$, we have
	\begin{equation}
	\label{eq:spectralK}
	\frac{1}{2}\cdot (\Kmu + \epsilon{\cal I}_{T}) \preceq \widehat{\cal K}_\mu + \epsilon {\cal I}_{T} \preceq 	\frac{3}{2}\cdot ({\cal K}_\mu + \epsilon{\cal I}_{T}).
	\end{equation}
\end{lemma}

\begin{proof}
		We prove a more general statement, in which we are given $0 < \Delta < 1$ and we select $s = \lceil 9 \smu / \Delta^2\rceil$ frequencies $\xi_1,\ldots,\xi_s \in \RR$ and weights $w_1,\ldots,w_s \in \RR$ such that
		\begin{equation*}
		(1-\Delta)({\cal K}_\mu + \epsilon{\cal I}_T) \preceq \widehat{\cal K}_\mu + \epsilon{\cal I}_T \preceq 	(1+\Delta)({\cal K}_\mu +\epsilon{\cal I}_T).
		\end{equation*}
		The claim follows by setting $\Delta = 1/2$. We can assume that $\xi_1,\ldots,\xi_s$ are distinct, since if $\xi_i, \xi_j$ are equal, we can simply remove $\xi_j$ and update $w_i \gets \sqrt{w^2_i + w^2_j}$, leaving $\widehat{\cal K}_\mu$ unchanged and only decreasing $s$.		
		
                The last condition is equivalent to
		\begin{equation*}
		{\cal K}_\mu -\Delta({\cal K}_\mu + \epsilon{\cal I}_T) \preceq \widehat{\cal K}_\mu \preceq {\cal K}_\mu + \Delta({\cal K}_\mu + \epsilon{\cal I}_T).
		\end{equation*}
		Multiplying with $({\cal K}_\mu + \epsilon{\cal I}_T)^{-1/2}$ on the left and right, we find that the condition is equivalent to:
		\begin{equation*}
		- \Delta {\cal I}_T \preceq ({\cal K}_\mu + \epsilon{\cal I}_T)^{-1/2} \widehat{\cal K}_\mu ({\cal K}_\mu + \epsilon{\cal I}_T)^{-1/2} - ({\cal K}_\mu + \epsilon{\cal I}_T)^{-1/2} {\cal K}_\mu ({\cal K}_\mu + \epsilon{\cal I}_T)^{-1/2} \preceq \Delta {\cal I}_T.
		\end{equation*}
		To shorten notation, we write
		${\cal Z} = ({\cal K}_\mu + \epsilon{\cal I}_T)^{-1/2} {\cal K}_\mu ({\cal K}_\mu + \epsilon{\cal I}_T)^{-1/2}$ and
		$\widehat{\cal Z} = ({\cal K}_\mu + \epsilon{\cal I}_T)^{-1/2} \widehat{\cal K}_\mu ({\cal K}_\mu + \epsilon{\cal I}_T)^{-1/2}$.
	 Given $\xi \in \RR$, we define $\vartheta_\xi(t) \eqdef e^{2\pi i t \xi}$ ($\vartheta_\xi \in L_2 (T)$).
	 It is easy to verify that
	 \begin{equation*}
	 {\cal K}_\mu  =  \int_\RR (\vartheta_\xi \otimes \vartheta_\xi) d\mu(\xi)
	 \end{equation*}
	 and
	 \begin{equation*}
	 \widehat{{\cal K}}_\mu  =  \sum^s_{i=1} w^2_i (\vartheta_{\xi_i} \otimes \vartheta_{\xi_i}).
	 \end{equation*}
	 Further define $\bar{\vartheta}_\xi \eqdef ({\cal K}_\mu + \epsilon{\cal I}_T)^{-1/2} \vartheta_\xi$. Since $({\cal K}_\mu + \epsilon{\cal I}_T)^{-1/2}$ is self-adjoint and bounded, we have
	 \begin{equation*}
	 {\cal Z}  =  \int_\RR (\bar{\vartheta}_\xi \otimes \bar{\vartheta}_\xi) d\mu(\xi)
	 \end{equation*}
	 and
	 \begin{equation*}
	 \widehat{\cal Z} =  \sum^s_{i=1} w^2_i(\bar{\vartheta}_{\xi_i} \otimes \bar{\vartheta}_{\xi_i}).
	 \end{equation*}
	
	 We prove the existence of $\xi_1, \dots, \xi_s$ and $w_1, \dots, w_s$ using the deterministic selection process known as ``BSS"~\cite{BSS14}.\footnote{We remark that unlike the process described in~\cite{BSS14}, our existence proof does not trivially translate to an algorithm, since it involves a search over an infinite domain. Nevertheless, for our needs, existence suffices.}
	 In particular, we use a process that in essence is the same as the one described in~\cite[Theorem 5 (arxiv version)]{CohenNelsonWoodruff16}. Indeed, since $\opnorm{{\cal Z}} \leq 1$ and $\tr({\cal Z}) = \smu$ the aforementioned results would  suffice if we were dealing with matrices instead of operators. The rest of the proof extends these results to the operator case. 
	 \ifdraft
	 \Haim{Remark for us: we know that ${\cal Z}$ is Hilbert-Schmidt because ${\cal K}_\mu$ is H-S and we multiply it by a bounded operators.}
    %\todo{Now that we introduced weak integrals, we can shorten some parts of the proof.}
    \fi
	 Let
	 $$ \delta_u \eqdef \Delta / 3 + 2 \Delta^2 / 9,\quad \delta_l \eqdef \Delta / 3 - 2 \Delta^2 / 9$$
	 and for $j = 0, 1, \dots, s$,
	 $${\cal X}^{(j)}_l \eqdef j \delta_l \cdot {\cal Z} - \smu\cdot {\cal I}_T,\quad {\cal X}^{(j)}_u \eqdef j \delta_u \cdot {\cal Z} + \smu \cdot {\cal I}_T.$$
	 The process we shall describe iteratively selects $\xi_1, \xi_2, \dots$ and unscaled weights $\tilde{w}_1, \tilde{w}_2, \ldots$ such that if we define
	 $\widehat{{\cal Z}}^{(j)} \eqdef \sum^j_{i=1} \tilde{w}_i (\bar{\vartheta}_{\xi_i} \otimes \bar{\vartheta}_{\xi_i})$
	 the invariant
	 \begin{equation}
	 \label{eq:invariant_uplo}
	 {\cal X}^{(j)}_l \prec \widehat{{\cal Z}}^{(j)} \prec {\cal X}^{(j)}_u
	 \end{equation}
	 is held. Let us write $s = \lceil 9 \smu/\Delta^2 \rceil$, so $s = C\smu/\Delta^2$ for $C \geq 9$. If indeed we are able to select the frequencies and weights for $s$ steps such that this invariant holds, we shall have
	 $$
	 \frac{C\smu}{3\Delta} \cdot {\cal Z} - (1 + 2C/9)\cdot \smu \cdot {\cal I}_T \preceq \widehat{{\cal Z}}^{(s)}  \preceq \frac{C\smu}{3\Delta} \cdot {\cal Z} + (1 + 2C/9)\cdot \smu \cdot {\cal I}_T
	 $$
	 where we used the fact that ${\cal Z} \preceq {\cal I}_T$. Since $C \geq 9$ we have
	 $$
	 -\Delta \cdot {\cal I}_T \preceq \frac{3\Delta}{C \smu}\widehat{{\cal Z}}^{(s)} - {\cal Z} \preceq \Delta \cdot {\cal I}_T
	 $$
	 so by defining $w_i = \sqrt{\frac{3\Delta}{C \smu} \tilde{w}_i}$ for $i = 1, \dots, s$ we shall then have $\widehat{{\cal Z}} = \frac{3\Delta}{C \smu}\widehat{{\cal Z}}^{(s)}$
	 thereby establishing the desired bound.
	
	 Thus, it suffices to show that we can select frequencies and weights iteratively so that~\eqref{eq:invariant_uplo} is maintained. In fact, the iterative selection process will maintain two additional invariants:
	 \begin{align*}
	 \int_{\RR} \langle \bar{\vartheta}_\xi, ({\cal X}^{(j)}_u - \widehat{{\cal Z}}^{(j)})^{-1} \bar{\vartheta}_\xi\rangle_T d\mu(\xi) &\leq 1\\
	 \int_{\RR} \langle \bar{\vartheta}_\xi, (\widehat{{\cal Z}}^{(j)} - {\cal X}^{(j)}_l)^{-1} \bar{\vartheta}_\xi\rangle_T d\mu(\xi) &\leq 1
	 \end{align*}
	 All the invariants hold for $j=0$. Eq.~\eqref{eq:invariant_uplo} trivially holds for $j=0$. As for the integral,
	 \begin{align*}
	 \int_{\RR} \langle \bar{\vartheta}_\xi, ({\cal X}^{(0)}_u - \widehat{{\cal Z}}^{(0)})^{-1} \bar{\vartheta}_\xi\rangle_T d\mu(\xi)  &=
	 \int_{\RR} \langle \bar{\vartheta}_\xi, \smu^{-1} \bar{\vartheta}_\xi\rangle_T d\mu(\xi) \\
	 & = \smu^{-1} \int_{\RR} \langle ({\cal K}_\mu + \epsilon{\cal I}_{T})^{-1/2} \vartheta_\xi,  ({\cal K}_\mu + \epsilon{\cal I}_{T})^{-1/2} \vartheta_\xi\rangle_T d\mu(\xi) \\
	 & = \smu^{-1} \int_{\RR} \langle  \vartheta_\xi,  ({\cal K}_\mu + \epsilon{\cal I}_{T})^{-1} \vartheta_\xi\rangle_T d\mu(\xi) \\
	 & = \smu^{-1} \tr\left(({\cal K}_\mu + \epsilon{\cal I}_{T})^{-1} {\cal K}_T \right) = 1
	 \end{align*}
	 and similarly for the second invariant. In the above, the last equality is due to Claim~\ref{claim:int-inner-to-trace}.
	
	 Suppose by induction that the invariants for $j$. We prove that it is possible to pick a frequency $\xi$ and weight $w > 0$ such that if we set $\xi_{j+1} = \xi$ and $\tilde{w}_{j+1} = w$ then the invariants will hold for $j+1$.
	
	 Fix $j$. For $t \geq 0$, let us denote
	 \begin{align*}
	 	 M_u(t) = \left({\cal X}^{(j)}_u + t{\cal Z}- \widehat{{\cal Z}}^{(j)}\right)^{-1} \\
	 	 M_l(t) = \left(\widehat{{\cal Z}}^{(j)} - {\cal X}^{(j)}_l - t{\cal Z}\right)^{-1}
	 \end{align*}
	 where $M_u$ is defined for any $t$ (since the inverted operator is strictly positive and bounded, so invertible), and $M_l$ is defined for $t < 1$. We can define $M_l$ for $t < 1$ since $\widehat{{\cal Z}}^{(j)} - {\cal X}^{(j)}_l - t{\cal Z} \succ 0$ for $t < 1$ as we now show. Due to Claim~\ref{claim:int-inner-to-trace}:
	 $$
	 \tr((\widehat{{\cal Z}}^{(j)} - {\cal X}^{(j)}_l)^{-1} {\cal Z}) = \int_{\RR} \langle \bar{\vartheta}_\xi, (\widehat{{\cal Z}}^{(j)} - {\cal X}^{(j)}_l)^{-1} \bar{\vartheta}_\xi\rangle_T d\mu(\xi) \leq 1.
	 $$
	 Since $\widehat{{\cal Z}}^{(j)} - {\cal X}^{(j)}_l \succ 0$ (induction assumption), $(\widehat{{\cal Z}}^{(j)} - {\cal X}^{(j)}_l)^{-1}$ is bounded so according to Claim~\ref{clm:trace-to-ord}, ${\cal Z} \preceq  \widehat{{\cal Z}}^{(j)} - {\cal X}^{(j)}_l$, and then Claim~\ref{claim:stays-positive} implies that $\widehat{{\cal Z}}^{(j)} - {\cal X}^{(j)}_l - t{\cal Z} \succ 0$.
	
	 Consider some fixed $\xi$. We first claim that for $w < 1 / \langle \bar{\vartheta}_\xi, M_u(\delta_u)\bar{\vartheta}_\xi\rangle_T$ we have $M_u(\delta_u)^{-1} - w (\bar{\vartheta}_\xi \otimes \bar{\vartheta}_\xi) \succ 0$. Obviously, the last statement holds for $w=0$, and due to continuity of $w\mapsto \langle x, (M_u(\delta_u)^{-1} - w (\bar{\vartheta}_\xi \otimes \bar{\vartheta}_\xi))x\rangle_T$ with respect to $w$, it will also hold for some interval around $0$. Let $w^\star$ be the maximal value such that for all $w\in[0, w^\star)$ we have $M_u(\delta_u)^{-1} - w (\bar{\vartheta}_\xi \otimes \bar{\vartheta}_\xi) \succ 0$. Our goal is to show that $w^\star \geq  1 / \langle \bar{\vartheta}_\xi, M_u(\delta_u)\bar{\vartheta}_\xi\rangle_T$. Assume by contradiction that
	 $w^\star <  1 / \langle \bar{\vartheta}_\xi, M_u(\delta_u)\bar{\vartheta}_\xi\rangle_T$. For every $w \in [0, w^\star)$, the operator  $M_u(\delta_u)^{-1} - w (\bar{\vartheta}_\xi \otimes \bar{\vartheta}_\xi)$ is invertible, and  we can apply a operator pseudo-inversion lemma due to Deng~\cite[Theorem 2.1]{Deng11} to find that
	 $$
	 (M_u(\delta_u)^{-1} - w (\bar{\vartheta}_\xi \otimes \bar{\vartheta}_T))^{-1} = M_u(\delta_u) + \frac{w}{1 - w\cdot\langle \bar{\vartheta}_\xi, M_u(\delta_u)\bar{\vartheta}_\xi\rangle_T}M_u(\delta_u) (\bar{\vartheta}_\xi \otimes \bar{\vartheta}_\xi) M_u(\delta_u).
	 $$
	 Since we assumed $w^\star < 1 / \langle \bar{\vartheta}_\xi, M_u(\delta_u)\bar{\vartheta}_\xi\rangle_T$, clearly, there exists a $K$ such that for all $w \in [0, w^\star)$ we have:
	 $$
	 (M_u(\delta_u)^{-1} - w (\bar{\vartheta}_\xi \otimes \bar{\vartheta}_T))^{-1} \leq K \cdot {\cal I}_T.
	 $$
	 Note that $M_u(\delta_u)^{-1} - w^\star (\bar{\vartheta}_\xi \otimes \bar{\vartheta}_\xi)$ is not strictly positive for otherwise due to continuity we could have extended the interval, so there exists a $x$ with norm $1$ such that $\langle x, (M_u(\delta_u)^{-1} - w^\star (\bar{\vartheta}_\xi \otimes \bar{\vartheta}_\xi))x\rangle  < 1/2K$. Let $w_1, w_2, \dots$ be a sequence which converges to $w^\star$, and let $y_i = (M_u(\delta_u)^{-1} - w_i (\bar{\vartheta}_\xi \otimes \bar{\vartheta}_\xi)^{1/2})x$. We now have $\langle y_i, y_i\rangle_T = \langle x, (M_u(\delta_u)^{-1} - w_i (\bar{\vartheta}_\xi \otimes \bar{\vartheta}_\xi)x\rangle_T \to \langle x, (M_u(\delta_u)^{-1} - w^\star (\bar{\vartheta}_\xi \otimes \bar{\vartheta}_\xi))x\rangle_T < 1/2K$ as $i \to \infty$.
	 However $\langle y_i, (M_u(\delta_u)^{-1} - w_i (\bar{\vartheta}_\xi \otimes \bar{\vartheta}_\xi))^{-1}y_i\rangle_T = \langle x_i, x_i \rangle_T = 1$ which contradicts the bound on
	 $(M_u(\delta_u)^{-1} - w_i (\bar{\vartheta}_\xi \otimes \bar{\vartheta}_\xi))^{-1}$.
	
	 Thus, if we picked $\xi$ and $w < 1 / \langle \bar{\vartheta}_\xi, M_u(\delta_u)\bar{\vartheta}_\xi\rangle_T$ for the step, we shall have $\widehat{{\cal Z}}^{(j+1)} - {\cal X}^{(j)}_l = M_u(\delta_u)^{-1} - w (\bar{\vartheta}_\xi \otimes \bar{\vartheta}_\xi) \succ 0$ as required, and the upper invariant will translate to
	 \begin{equation*}
	 \int_{\RR} \left\langle \bar{\vartheta}_\eta, \left(M_u(\delta_u)^{-1} - w (\bar{\vartheta}_\xi \otimes \bar{\vartheta}_\xi)\right)^{-1} \bar{\vartheta}_\eta\right\rangle_T d\mu(\eta) \leq 1,
	 \end{equation*}
	 which is equivalent to
	 \begin{equation*}
	 \int_{\RR} \langle \bar{\vartheta}_\eta, M_u(\delta_u)\bar{\vartheta}_\eta\rangle_T d\mu(\eta) + \frac{w \cdot \int_{\RR} \left\langle\bar{\vartheta}_\eta,  M_u(\delta_u) (\bar{\vartheta}_\xi \otimes \bar{\vartheta}_\xi) M_u(\delta_u)\bar{\vartheta}_\eta\right\rangle_T d\mu(\eta)}{1 - w\cdot\langle \bar{\vartheta}_\xi, M_u(\delta_u)\bar{\vartheta}_\xi\rangle_T} \leq 1.
	 \end{equation*}
	 The induction hypothesis is
	 \begin{equation*}
	 \int_{\RR} \langle \bar{\vartheta}_\eta, M_u(0)\bar{\vartheta}_\eta\rangle_T d\mu(\eta) \leq 1.
	 \end{equation*}
	 so the upper invariant is held if
	 	 \begin{equation}
	 	 \label{eq:upper-invariant-c}
	 	 \int_{\RR} \langle \bar{\vartheta}_\eta, M_u(\delta_u)\bar{\vartheta}_\eta\rangle_T d\mu(\eta) - \int_{\RR} \langle \bar{\vartheta}_\eta, M_u(0)\bar{\vartheta}_\eta\rangle_T d\mu(\eta) + \frac{w \cdot \int_{\RR} \left\langle\bar{\vartheta}_\eta,  M_u(\delta_u) (\bar{\vartheta}_\xi \otimes \bar{\vartheta}_\xi) M_u(\delta_u)\bar{\vartheta}_\eta\right\rangle_T d\mu(\eta)}{1 - w\cdot\langle \bar{\vartheta}_\xi, M_u(\delta_u)\bar{\vartheta}_\xi\rangle_T} \leq 0.
	 	 \end{equation}
	 	 Consider any $\eta \in \RR$, and let $f_\eta(y) \eqdef \langle \bar{\vartheta}_\eta, M_u(y)\bar{\vartheta}_\eta\rangle_T$. Using the operator inversion formula, we have for any $t_2 \geq t_1$:
	 	 \begin{equation*}
	 	 M_u(t_2) = M_u(t_1) - (t_2 - t_1)M_u(t_1){\cal Z}^{1/2}\left({\cal I}_T + (t_2 - t_1){\cal Z}^{1/2} M_u(t_1){\cal Z}^{1/2}\right)^{-1} {\cal Z}^{1/2}M_u(t_1).
	 	 \end{equation*}
	 	 From this equation we see that
	 	 $$
	 	 f'_\eta(y) = \langle \bar{\vartheta}_\eta, M_u(y){\cal Z}M_u(y)\bar{\vartheta}_\eta\rangle_T.
	 	 $$
	 	 Furthermore, since for $t_2> t_1$ we have ${\cal I}_T + t_2  {\cal Z}^{1/2} M_u(t_1){\cal Z}^{1/2} \succeq {\cal I}_T + t_1  {\cal Z}^{1/2} M_u(t_1){\cal Z}^{1/2} $ and both operators are strictly positive and bounded, then $({\cal I}_T + t_1 {\cal Z}^{1/2} M_u(t_1){\cal Z}^{1/2})^{-1} \preceq ({\cal I}_T +  t_2  {\cal Z}^{1/2} M_u(t_1){\cal Z}^{1/2})^{-1} $, and we can easily verify that $f_\eta$ is convex. Thus,
	 	 \begin{equation*}
		 f_\eta(\delta_u) - f_\eta(0) \leq -\delta_u  \langle \bar{\vartheta}_\eta, M_u(y){\cal Z}M_u(y)\bar{\vartheta}_\eta\rangle_T.
	 	\end{equation*}
	 	 After integrating on both sides, we have the bound
	 	 \begin{equation*}
	 	  \int_{\RR} \langle \bar{\vartheta}_\eta, M_u(\delta_u)\bar{\vartheta}_\eta\rangle_T d\mu(\eta) - \int_{\RR} \langle \bar{\vartheta}_\eta, M_u(0)\bar{\vartheta}_\eta\rangle_T d\mu(\eta)\leq - \delta_u \int_{\RR} \langle \bar{\vartheta}_\eta, M_u(\delta_u) {\cal Z} M_u(\delta_u) \bar{\vartheta}_\eta\rangle_T d\mu(\eta).
	 	 \end{equation*}
	 	   Using this bound in \eqref{eq:upper-invariant-c} and rearranging, we find that for any $\xi$, the upper invariant is held if we select $w$ such that
		 \begin{equation}
	 	 \label{eq:final-w-upper}
	 	 \frac{1}{w} > \frac{\int_{\RR} \left\langle\bar{\vartheta}_\eta,  M_u(\delta_u) (\bar{\vartheta}_\xi \otimes \bar{\vartheta}_\xi) M_u(\delta_u)\bar{\vartheta}_\eta\right\rangle_T d\mu(\eta)}{\delta_u \int_{\RR} \langle \bar{\vartheta}_\eta, M_u(\delta_u) {\cal Z} M_u(\delta_u) \bar{\vartheta}_\eta\rangle_T d\mu(\eta)} + \langle \bar{\vartheta}_\xi, M_u(\delta_u)\bar{\vartheta}_\xi\rangle_T.
	 	 \end{equation}
	 	 Note that if this is held, we also have $w < 1 / \langle \bar{\vartheta}_\xi, M_u(\delta_u)\bar{\vartheta}_\xi\rangle_T$, as previously required.
	 	
	 	 We now consider the lower invariants. If we picked $\xi$ and $w > 0$ for the step, then $\widehat{{\cal Z}}^{(j+1)} - {\cal X}_l^{(j)} = M_l(\delta_l)^{-1} + w (\bar{\vartheta}_\xi \otimes \bar{\vartheta}_\xi) \succeq M_l(\delta_l)^{-1}  \succ 0$ as long $\delta_l < 1$ which holds for our choice of $\delta_l$. So the left part of \eqref{eq:invariant_uplo} will hold regardless of how we choose $\xi$ and $w > 0$. As for the lower trace bound, it translates to
	 	 \begin{equation*}
	 	 \int_{\RR} \left\langle \bar{\vartheta}_\eta, \left(M_l(\delta_l)^{-1} + w (\bar{\vartheta}_\xi \otimes \bar{\vartheta}_\xi)\right)^{-1} \bar{\vartheta}_\eta\right\rangle_T d\mu(\eta) \leq 1.
	 	 \end{equation*}
	 	 Applying another variant of operator pseudo-inversion lemma~\cite[Theorem 2]{Ogawa:1988}, we find that the last condition is equivalent to
		 \begin{equation*}
		 \int_{\RR} \langle \bar{\vartheta}_\eta, M_l(\delta_l)\bar{\vartheta}_\eta\rangle_T d\mu(\eta) - \frac{w \cdot \int_{\RR} \left\langle\bar{\vartheta}_\eta,  M_l(\delta_l) (\bar{\vartheta}_\xi \otimes \bar{\vartheta}_\xi) M_l(\delta_l)\bar{\vartheta}_\eta\right\rangle_T d\mu(\eta)}{1 + w\cdot\langle \bar{\vartheta}_\xi, M_l(\delta_l)\bar{\vartheta}_\xi\rangle_T} \leq 1.
		 \end{equation*}
		 The induction hypothesis is
		 \begin{equation*}
		 \int_{\RR} \langle \bar{\vartheta}_\eta, M_l(0)\bar{\vartheta}_\eta\rangle_T d\mu(\eta) \leq 1
		 \end{equation*}
		  so the lower invariant is held if
		  \begin{equation}
		  \label{eq:lower-invariant-c}
		  \int_{\RR} \langle \bar{\vartheta}_\eta, M_l(\delta_l)\bar{\vartheta}_\eta\rangle_\mu d\mu(\eta) - \int_{\RR} \langle \bar{\vartheta}_\eta, M_l(0)\bar{\vartheta}_\eta\rangle_T d\mu(\eta) - \frac{w \cdot \int_{\RR} \left\langle\bar{\vartheta}_\eta,  M_l(\delta_l) (\bar{\vartheta}_\xi \otimes \bar{\vartheta}_\xi) M_l(\delta_l)\bar{\vartheta}_\eta\right\rangle_T d\mu(\eta)}{1 + w\cdot\langle \bar{\vartheta}_\xi, M_l(\delta_l)\bar{\vartheta}_\xi\rangle_T} \leq 0.
		  \end{equation}
		  Similarly to before, by using the convexity of each integrand, we can bound
		  \begin{equation*}
		  \int_{\RR} \langle \bar{\vartheta}_\eta, M_l(\delta_l)\bar{\vartheta}_\eta\rangle_T d\mu(\eta) - \int_{\RR} \langle \bar{\vartheta}_\eta, M_l(0)\bar{\vartheta}_\eta\rangle_T d\mu(\eta)\leq  \delta_l \int_{\RR} \langle \bar{\vartheta}_\eta, M_l(\delta_l) {\cal Z} M_u(\delta_l) \bar{\vartheta}_\eta\rangle_T d\mu(\eta).
		  \end{equation*}
		  Using this bound in \eqref{eq:lower-invariant-c} and rearranging, we find that for any $\xi$, the lower invariant is held if we select $w$ such that
		  \begin{equation}
		  \label{eq:final-w-lower}
		  \frac{1}{w} \leq \frac{\int_{\RR} \left\langle\bar{\vartheta}_\eta,  M_l(\delta_l) (\bar{\vartheta}_\xi \otimes \bar{\vartheta}_\xi) M_l(\delta_l)\bar{\vartheta}_\eta\right\rangle_T d\mu(\eta)}{\delta_l \int_{\RR} \langle \bar{\vartheta}_\eta, M_l(\delta_l) {\cal Z} M_u(\delta_l) \bar{\vartheta}_\eta\rangle_T d\mu(\eta)} - \langle \bar{\vartheta}_\xi, M_l(\delta_l)\bar{\vartheta}_\xi\rangle_T.
		  \end{equation}
		
		  Thus, we need to show that there exists a $\xi$ and $w$ such that both~\eqref{eq:final-w-upper} and~\eqref{eq:final-w-lower} hold. However, for a given $\xi$, such a $w$ will surely exist if
		  \begin{align*}
		  &\frac{\int_{\RR} \left\langle\bar{\vartheta}_\eta,  M_u(\delta_u) (\bar{\vartheta}_\xi \otimes \bar{\vartheta}_\xi) M_u(\delta_u)\bar{\vartheta}_\eta\right\rangle_T d\mu(\eta)}{\delta_u \int_{\RR} \langle \bar{\vartheta}_\eta, M_u(\delta_u) {\cal Z} M_u(\delta_u) \bar{\vartheta}_\eta\rangle_T d\mu(\eta)} + \langle \bar{\vartheta}_\xi, M_u(\delta_u)\bar{\vartheta}_\xi\rangle_T \\&< \frac{\int_{\RR} \left\langle\bar{\vartheta}_\eta,  M_l(\delta_l) (\bar{\vartheta}_\xi \otimes \bar{\vartheta}_\xi) M_l(\delta_l)\bar{\vartheta}_\eta\right\rangle_T d\mu(\eta)}{\delta_l \int_{\RR} \langle \bar{\vartheta}_\eta, M_l(\delta_l) {\cal Z} M_u(\delta_l) \bar{\vartheta}_\eta\rangle_T d\mu(\eta)} - \langle \bar{\vartheta}_\xi, M_l(\delta_l)\bar{\vartheta}_\xi\rangle_T.
		  \end{align*}
		  Thus, it it suffices to show that there exists a $\xi$ for which the last inequality holds. To show that such a $\xi$ exists, we will show that the inequality holds for the  integral of both sides with respect to $\mu$ measure. This will guarantee the existence of such a $\xi$ since the Lebesgue integral is strictly positive for non-negative functions.
		  We compute:
		  \begin{align*}
		  &\int_{\RR }\int_{\RR} \left\langle\bar{\vartheta}_\eta,  M_u(\delta_u) (\bar{\vartheta}_\xi \otimes \bar{\vartheta}_\xi) M_u(\delta_u)\bar{\vartheta}_\eta\right\rangle_T d\mu(\eta)d\mu(\xi) \\
		  &= \int_{\RR }\int_{\RR} \left\langle\bar{\vartheta}_\eta,  M_u(\delta_u) (\bar{\vartheta}_\xi \otimes \bar{\vartheta}_\xi) M_u(\delta_u)\bar{\vartheta}_\eta\right\rangle_T d\mu(\xi) d\mu(\eta)\\
		  &= \int_{\RR }\int_{\RR} \left\langle M_u(\delta_u) \bar{\vartheta}_\eta,  (\bar{\vartheta}_\xi \otimes \bar{\vartheta}_\xi) M_u(\delta_u)\bar{\vartheta}_\eta\right\rangle_T d\mu(\xi) d\mu(\eta)\\
		  &= \int_{\RR } \left\langle M_u(\delta_u) \bar{\vartheta}_\eta,  {\cal Z} M_u(\delta_u)\bar{\vartheta}_\eta\right\rangle_T d\mu(\eta)\\
		  &=\int_{\RR} \langle \bar{\vartheta}_\eta, M_u(\delta_u) {\cal Z} M_u(\delta_u) \bar{\vartheta}_\eta\rangle_T d\mu(\eta).
		  \end{align*}
		  Similarly,
		  \begin{equation*}
		  \int_{\RR }\int_{\RR} \left\langle\bar{\vartheta}_\eta,  M_l(\delta_l) (\bar{\vartheta}_\xi \otimes \bar{\vartheta}_\xi) M_l(\delta_l)\bar{\vartheta}_\eta\right\rangle_T d\mu(\eta)d\mu(\xi)
		  = \int_{\RR} \langle \bar{\vartheta}_\eta, M_l(\delta_l) {\cal Z} M_l(\delta_l) \bar{\vartheta}_\eta\rangle_T d\mu(\eta).
		  \end{equation*}
		  ${\cal Z}$ is self-adjoint and positive definite, so the operator pseudo-inversion lemma~\cite[Theorem 2]{Ogawa:1988} implies that $M_u(\delta_u) \preceq M_u(0)$, so by the induction hypothesis
		  \begin{equation*}
		  \int_{\RR }\langle \bar{\vartheta}_\xi, M_u(\delta_u)\bar{\vartheta}_\xi\rangle_T d\mu(\xi) \leq \int_{\RR }\langle \bar{\vartheta}_\xi, M_u(0)\bar{\vartheta}_\xi\rangle_T d\mu(\xi) \leq 1.
		  \end{equation*}
		  We now consider the lower invariant. We already showed that  ${\cal Z} \preceq  \widehat{{\cal Z}}^{(j)} - {\cal X}^{(j)}_l$, so as long as $\delta_l \leq 1/2$ we will have:
		  \begin{equation*}
		  \int_{\RR }\langle \bar{\vartheta}_\xi, M_l(\delta_l)\bar{\vartheta}_\xi\rangle_T d\mu(\xi) =
		  \int_{\RR }\langle \bar{\vartheta}_\xi, (M_l(0)^{-1} - \delta_l {\cal Z})^{-1}\bar{\vartheta}_\xi\rangle_T d\mu(\xi) \leq
		  2\int_{\RR }\langle \bar{\vartheta}_\xi, M_l(0)\bar{\vartheta}_\xi\rangle_T d\mu(\xi) \leq 2
		  \end{equation*}
		  where we used Claim~\ref{clm:invordering}. So there will be a gap in the value of the integrals (as desired), if
		  $$
		  \frac{1}{\delta_u} + 1 < \frac{1}{\delta_l} - 2,
		  $$
		  which is the case for our selection of $\delta_l$ and $\delta_u$.
\end{proof}
From Lemma \ref{lem:cssSpectral} we can prove a stronger spectral error bound for the projection onto the range of $\bv{\bar C}_s$.
\begin{lemma}[Frequency Subset Selection -- Projection Based Spectral Approximation]\label{lem:cssSpectral2}
	For some $s \le \lceil 36 \cdot \smu \rceil$ there exists a set of distinct frequencies $\xi_1,\ldots,\xi_s \in \CC$ such that letting $\bv{C}_s: L_2(T) \rightarrow \CC^{s}$ and $\bv{Z}: L_2(\mu) \rightarrow \CC^s$ be defined as in Theorem \ref{thm:css} and $\widehat{\cal G}_\mu = \bv{Z}^* \bv C_s \bv C_s^* \bv{Z}$,
	\begin{equation}
	\label{eq:spectralK2}
\widehat{\cal G}_\mu \preceq \Gmu \preceq \widehat{\cal G}_\mu + \epsilon \Imu.
	\end{equation}
\end{lemma}
\begin{proof}Let $\xi_1,\ldots,\xi_s \in \CC$ and $w_1,\ldots,w_s \in \RR$ be the frequencies and weights shown to exist in Lemma \ref{lem:cssSpectral} and let $\bv{\bar C}_s$ be as defined in that lemma (note that $\bv{\bar C}_s$ is identical to $\bv{C}_s$ except with its rows weighted by $w_1,\ldots,w_s$.)
First note that for any $g \in L_2(\mu)$,
\begin{align*}
\langle g, \widehat{\cal G}_\mu g \rangle_\mu = \norm{\bv C_s^*\bv{Z}g}_\mu^2 = \norm{\bv C_s^*(\bv{C}_s \bv{C}_s^*)^{-1} \bv{C}_s \Fmu^* g}_\mu^2 \le \norm{ \Fmu^* g}_\mu^2 = \langle g, {\cal G}_\mu g \rangle_\mu
\end{align*}
where the inequality follows from observing that $\bv C_s^*(\bv{C}_s \bv{C}_s^*)^{-1} \bv{C}_s$ is an orthogonal projection.
Thus $ \widehat{\cal G}_\mu \preceq  {\cal G}_\mu$. It remains to show that ${\cal G}_\mu \preceq \widehat{\cal G}_\mu + \epsilon \Imu$. Let $\mathcal{\bar P} = \mathcal{I}_T - \bv{C}_s^* (\bv{C}_s \bv{C}^*)^{-1} \bv{C}_s$ be the projection to the orthogonal complement of $\bv{C}_s^*$'s range and let $\widehat{\cal K}_\mu = \bv{\bar C}_s^* \bv{\bar C}_s$ be as defined in Lemma \ref{lem:cssSpectral}.
Rearranging the guarantee of Lemma \ref{lem:cssSpectral} gives
\begin{align*}
	{\cal K}_\mu \preceq 	2 \cdot \mathcal{\widehat K}_\mu + \epsilon{\cal I}_{T}
\end{align*}
which immediately  gives
\begin{align*}
	\mathcal{\bar P} {\cal K}_\mu \mathcal{\bar P} \preceq 	2\cdot  \mathcal{\bar P} \widehat{\cal K}_\mu \mathcal{\bar P} + \epsilon \mathcal{\bar P} {\cal I}_{T} \mathcal{\bar P} .
\end{align*}
Note that $\bv{\bar C}_s \mathcal{\bar P} = 0$ (since $\mathcal{\bar P}$ is an orthogonal projection onto $\ker(\bv{C}_s)=\ker(\bv{\bar C}_s)$) and so $\mathcal{\bar P}  \widehat{\cal K}_\mu \mathcal{\bar P}  = 0$, giving:
\begin{align}\label{eq:compBound}
\mathcal{\bar P} {\cal K}_\mu \mathcal{\bar P} \preceq \epsilon \mathcal{\bar P} {\cal I}_T \mathcal{\bar P} \preceq \epsilon \mathcal{I}_T.
\end{align}
Note that $\mathcal{\bar P} {\cal K}_\mu \mathcal{\bar P} = \mathcal{\bar P} \Fmu^* \Fmu \mathcal{\bar P}$ and
$$\Gmu - \widehat{\Gmu} = \Fmu \Fmu^* -\bv{Z}^* \bv{C}_s \bv{C}_s^* \bv{Z} =  \Fmu \mathcal{\bar P} \Fmu^*.$$
Thus by  \eqref{eq:compBound} we also have $\Gmu - \widehat{\Gmu} \preceq \epsilon \Imu$ (since the norm of an operator and its adjoint are the same so $\mathcal{\bar P} {\cal K}_\mu \mathcal{\bar P} \preceq \epsilon \mathcal{I}_T \implies \Fmu \mathcal{\bar P} \Fmu^* \preceq \epsilon \Imu$), which completes the lemma.
\end{proof}
Finally, from Lemma \ref{lem:cssSpectral2} we can prove the frequency subset selection guarantee of Theorem \ref{thm:css}.

\begin{proof}[Proof of Theorem \ref{thm:css}]
We consider the same set of frequencies $\xi_1,\ldots,\xi_s$ shown to exist in Lemma \ref{lem:cssSpectral2} and the corresponding operators $\bv{C}_s$, $\bv{Z}$. We show that these frequencies satisfy the guarantee of Theorem \ref{thm:css}.
First, we note that
$$
\bv{K} \eqdef \bv{C}_s\bv{C}^*_s = \frac{1}{T}\int^T_0  (\bs{\phi}_t \otimes  \bs{\phi}_t)dt
$$
(In the above, we abuse notation and use $\bs{\phi}_t$ to denote both the vector defined in the Theorem statement, and the operator $x\in\CC \mapsto x\bs{\phi}_t$).
From Claim~\ref{claim:int-inner-to-trace}:
\begin{align*}
\frac{1}{T} \int_{t  \in [0,T]} \norm{\varphi_t - \bv{Z}^* \bs{\phi}_t}_\mu^2\, dt =& \tr\left(\frac{1}{T} \int_{t  \in [0,T]}(\varphi_t - \bv{Z}^* \bs{\phi}_t)\otimes(\varphi_t - \bv{Z}^* \bs{\phi}_t) \, dt  \right) \\
=&  \tr\left(\frac{1}{T} \int_{t  \in [0,T]}\varphi_t \otimes\varphi_t\, dt \right)  + \tr\left(\frac{1}{T} \int_{t  \in [0,T]}\bv{Z}^*\bs{\phi}_t \otimes\bv{Z}^*\bs{\phi}_t\, dt \right) \\
& - \tr\left(\frac{1}{T} \int_{t  \in [0,T]}\bv{Z}^*\bs{\phi}_t \otimes\varphi_t\, dt \right) - \tr\left(\frac{1}{T} \int_{t  \in [0,T]}\varphi_t \otimes\bv{Z}^*\bs{\phi}_t\, dt \right)
\end{align*}
We have,
$$
\frac{1}{T} \int_{t  \in [0,T]}\varphi_t \otimes\varphi_t\, dt = \Gmu\,,
$$
From Claim~\ref{claim:linear-weak-integral}:
$$
\frac{1}{T} \int_{t  \in [0,T]}\bv{Z}^*\bs{\phi}_t \otimes\bv{Z}^*\bs{\phi}_t\, dt = \bv{Z}^*\left(\frac{1}{T} \int_{t  \in [0,T]}\bs{\phi}_t \otimes\bs{\phi}_t\, dt\right)\bv{Z} = \bv{Z}^*\bv{K}\bv{Z} = \widehat{\Gmu}
$$
Next, consider $\frac{1}{T}\int^T_0 \bs{\phi}_t \otimes \varphi_t\, dt$. For any $\alpha$, 
$$
\frac{1}{T}\left(\int^T_0 \bs{\phi}_t \otimes \varphi_t\, dt\right) \alpha = \frac{1}{T} \int^{T}_0\langle \varphi_t, \alpha \rangle_\mu \bs{\phi}_t\, dt 
$$
where the integral on the left is a weak vector integral. Since for every $g\in L_2(T)$, 
$$
\bv{C}_sg=\frac{1}{T} \int^T_0g(t) \bs{\phi}_t \, dt
$$
and for every $\alpha \in L_2(\mu)$, $[\Fmu^* \alpha ](t) = \langle \varphi_t, \alpha \rangle_\mu$, we have $\frac{1}{T}\int^T_0 \bs{\phi}_t \otimes \varphi_t\, dt = \bv{C}_s \Fmu^*$, so
$$
\frac{1}{T} \int_{t  \in [0,T]}\bv{Z}^*\bs{\phi}_t \otimes\varphi_t\, dt = \bv{Z}^* \left(\frac{1}{T} \int_{t  \in [0,T]}\bs{\phi}_t \otimes\varphi_t\, dt \right) = \bv{Z}^* \bv{C}_s \Fmu^* = \bv{Z}^*\bv{K}\bv{Z} = \widehat{\Gmu}\,.
$$
Combining the previous observations, we find that 
\begin{align*}
\frac{1}{T} \int_{t  \in [0,T]} \norm{\varphi_t - \bv{Z}^* \bs{\phi}_t}_\mu^2\, dt = \tr(\Gmu - \widehat{\Gmu}).
\end{align*}

Let $v_1,\ldots,v_{2\smu} \in L_2(\mu)$ be the eigenfunctions of $\Gmu$ corresponding to its top $2\smu$ eigenvalues. Define $\bv{X}: L_2(\mu) \rightarrow \CC^{2\smu}$ as: for $g \in L_2(\mu)$, $[\bv{X} g](j) = \langle v_j,g\rangle_\mu$. Note that
\begin{align*} \tr(\widehat{\Gmu} - \bv{X}^*\bv{X} \widehat{\Gmu} \bv{X}^* \bv{X}) &= \tr(\bv{Z}^* \bv{C}_s \bv{C}_s \bv{Z} - \bv{X}^* \bv{X} \bv{Z}^* \bv{C}_s \bv{C}_s \bv{Z} \bv{X}^* \bv{X})\\
&= \tr(\bv{C}_s \bv{Z} \bv{Z}^* \bv{C}_s  - \bv{C}_s \bv{Z} \bv{X}^* \bv{X} \bv{Z}^* \bv{C}_s) \ge 0
\end{align*}
since $\bv{C}_s \bv{Z} \bv{Z}^* \bv{C}_s  \succeq \bv{C}_s \bv{Z} \bv{X}^* \bv{X} \bv{Z}^* \bv{C}_s$ ($\bv{X}^* \bv{X}$ is a projection, so $\bv{X}^* \bv{X} \preceq {\cal I}_\mu$).
So we can bound:
\begin{align}\label{blahBound}
\frac{1}{T} \int_{t  \in [0,T]} \norm{\varphi_t - \bv{Z}^* \bs{\phi}_t}_\mu^2\, dt = \tr(\Gmu - \widehat{\Gmu}) &\le \tr(\Gmu - \widehat{\Gmu}) + \tr(\widehat{\Gmu} - \bv{X}^*\bv{X} \widehat{\Gmu} \bv{X}^* \bv{X})\nonumber\\
&= \tr(\Gmu -  \bv{X}^*\bv{X} {\Gmu} \bv{X}^*\bv{X}) + \tr(\bv{X}^* \bv{X} (\Gmu-\widehat{\Gmu})\bv{X}^*\bv{X}).
\end{align}
Let $i_\epsilon$ be the smallest $i$ with $\lambda_i(\Gmu) \le \epsilon$. We have:
\begin{align*}
\smu = \sum_{i=1}^\infty \frac{\lambda_i(\Gmu)}{\lambda_i(\Gmu)+\epsilon} &\ge  \sum_{i=1}^{i+\epsilon} \frac{\lambda_i(\Gmu)}{\lambda_i(\Gmu)+\epsilon} \ge \frac{i_\epsilon}{2}.
\end{align*}
Thus we can bound $\tr(\Gmu -  \bv{X}^*\bv{X} {\Gmu} \bv{X}^*\bv{X})$ as:
\begin{align}\label{blahBound1}
\tr(\Gmu -  \bv{X}^*\bv{X} {\Gmu} \bv{X}^*\bv{X}) = \sum_{i = 2\smu+1}^\infty \lambda_i(\Gmu) \le \sum_{i=i_\epsilon+1}^\infty  \lambda_i(\Gmu) \le 2 \epsilon \smu.
\end{align}
where the last bound follows from the fact that $\smu  \ge  \sum_{i=i_\epsilon+1}^\infty \frac{\lambda_i(\Gmu)}{\lambda_i(\Gmu)+\epsilon} \ge  \sum_{i=i_\epsilon+1}^\infty \frac{\lambda_i(\Gmu)}{2\epsilon}$.

We can also bound $\tr(\bv{X}^* \bv{X} (\Gmu-\widehat{\Gmu})\bv{X}^*\bv{X})$ using Lemma \ref{lem:cssSpectral2}. Since $\Gmu \le \widehat{\Gmu} + \epsilon \Imu$ we have:
\begin{align}\label{blahBound2}
\tr(\bv{X}^* \bv{X} (\Gmu-\widehat{\Gmu})\bv{X}^*\bv{X}) \le \epsilon \tr(\bv{X}^* \bv{X}\bv{X}^*\bv{X}) = \epsilon \smu.
\end{align}
Plugging \eqref{blahBound1} and \eqref{blahBound2} back into \eqref{blahBound} we have:
\begin{align*}
\frac{1}{T} \int_{t  \in [0,T]} \norm{\varphi_t - \bv{Z}^* \bs{\phi}_t}_\mu^2\, dt \le 4 \epsilon \cdot \smu,
\end{align*}
which completes the theorem.
\end{proof}

\section{Tight Statistical Dimension  Bound for Bandlimited Functions}
\label{sec:bandlimited}

In Section \ref{sec:general} we demonstrate, perhaps surprisingly, that a simple function $\ttmu(t)$ (defined in Theorem \ref{thm:fullBound}) exists for \emph{any} $\mu$ that upper bounds $\tmu(t)$ and has $\tsmu = \tilde O(\smu)$. Combined with Theorem \ref{thm:mainAlg} this yields our main algorithmic result Theorem \ref{thm:informal_main}, which shows that we can achieve $O\left (\smu \log^2(\smu \right  )$ sample complexity with just $\tilde O \left (\smu^\omega \right  )$ runtime.

Instantiating Theorem \ref{thm:informal_main} using the approximate ridge leverage function of Theorem \ref{thm:fullBound} requires an upper bound on $\smu$.
In this section we show how to bound $\smu$ when $\mu$ is uniform measure on some interval -- i.e., when our interpolation problem is over bandlimited functions. In Section \ref{app:stat_dimension} we leverage this result to bound $\smu$  for a number of other important priors, including for multiband, Gaussian, and Cauchy-Lorentz.

Beyond letting us upper bound $\smu$ to apply Theorem \ref{thm:informal_main}, our proof for bandlimited functions is constructive, giving a simple upper bound for  $\tmu(t)$ for any $t$. This upper bound can be plugged directly into Algorithm \ref{alg:main} and Theorem \ref{thm:mainAlg} to give a tightening of Theorem \ref{thm:informal_main} by a logarithmic factor in the bandlimited case.
%
%Before proving this result in full generality, we give a simpler (and tighter) proof for the case when $\mu$ is uniform on some interval  -- i.e. we are interpolating a bandlimited function.
Like our general result, the proof is based on the definition of leverage scores given in \eqref{eq:leverage_def}. This definition makes it clear that, to upper bound $\tmu(t)$, it suffices to show that a function with Fourier support controlled by $\mu$ cannot ``spike'' too extremely at time $t$.

For bandlimited functions, we obtain a smoothness bound by introducing and applying a Bernstein type smoothness bound for low-degree polynomials and relying on the fact that any bandlimited function is well approximated by a low-degree polynomial. This approach mirrors the general proof in Section \ref{sec:general}, which uses a more sophisticated smoothness bound for sparse Fourier functions.

Our result for bandlimited function is as follows:

\begin{theorem}
	\label{thm:bandlimited_leverage_scores} Let $\mu$ be the uniform measure on $[-F,F]$. Let $q = \lceil 16\pi e FT + 2\log(1/\epsilon) +11 \rceil$. For all $t \in [0,T]$, let the approximate ridge leverage function $\ttmu$ equal:
	\begin{align*}
		\ttmu(t) =  \frac{1}{T}\left(4 + \frac{q}{\sqrt{\min(t,T-t)/T}}\right).
	\end{align*}
	For any $\epsilon \le 1,F,T$, $\ttmu(t)$ satisfies:
	\begin{enumerate}
		\item $\ttmu(t) \geq \tmu(t)$.
		\item $\int_{0}^T\ttmu(t) dt \eqdef \tsmu = O\left(FT + \log(1/\epsilon)\right).$
	\end{enumerate}
	Thus we have $ \smu \le \tilde  \smu = O\left(FT + \log(1/\epsilon)\right)$.
\end{theorem}
\ifdraft
\Haim{Maybe we should use $\tilde{\tau}_{F,T,\epsilon}$ instead of $\tilde{\tau}_{\mu,\epsilon}$? Reason: $\mu$ is uniform on $[-F,F]$ so $F$ is the parameter that defines the distribution.  $\tilde{\tau}_{\mu,\epsilon}$ also depend on $T$ that does not appear in the subscript.}\Cam{Going to leave the same for now since this matches the general notation. I agree that its not ideal that $\tau_{\mu,\epsilon}$ depends on $T$ but this isn't clear from notation.}
\fi

Combined with Theorem \ref{thm:mainAlg}, Theorem \ref{thm:bandlimited_leverage_scores} immediately gives: %a strengthening of Theorem \ref{thm:informal_main} for bandlimited signals:
\begin{corollary}
Let $\mu$ be the uniform measure  on $[-F,F]$.
Using  $\ttmu$ as defined in Theorem \ref{thm:bandlimited_leverage_scores}, Algorithm \ref{alg:main} returns $t_1,\ldots,t_s \in [0,T]$ and $\bv{z} \in \CC^s$ such that $\tilde y(t) = \sum_{i=1}^s \bv{z}(i) \cdot k_\mu(t_i,t)$ satisfies with probability $\ge 1-\delta$:
\begin{align*}
\norm{\tilde y - y}_T^2  \le 6\epsilon \norm{x}_\mu^2 + 7 \norm{n}_T^2.
\end{align*}
The algorithm queries $y+n$ at $s$ points and runs in $O(s^\omega)$ time where $s = O([FT + \log(1/\epsilon)]\cdot [\log(FT + \log(1/\epsilon)) + 1/\delta])$. The output $\tilde y(t)$ can be evaluated using Algorithm \ref{alg:main2} in $O(s)$ time.
\end{corollary}
\begin{proof}
The corollary follows immediately from Theorem \ref{thm:mainAlg} after noting that
\begin{itemize}
\item $Z = O(1)$ since, as shown in Appendix \ref{app:kernel_computation}, $k_\mu(t_1,t_2) = \frac{\sin(2\pi F(t_1-t_2))}{2\pi F(t_1-t_2)}$ and so can be computed in $O(1)$ arithmetic operations.
\item $W = O(1)$ since to sample points proportional to $\ttmu(t)$, we must sample a mixture of the uniform distribution and the distribution with density proportional to $\frac{1}{\sqrt{\min(t,T-t)/T}}$. It suffices to show that we can sample from the later in $O(1)$ time, and in fact that we can sample $t \in [0,1/2]$ with probability proportional to $\frac{1}{\sqrt{t}}$ in $O(1)$ time, since we can then symmetrize and scale this distribution. We can accomplish this with inverse transform sampling. Our density  is $p(t) = \frac{1}{2\sqrt{2t}}$ and so its cumulative distribution function is   $C(t) = \sqrt{t/2}$. Thus we can sample $z$ uniformly in $[0,1]$ and return $C^{-1}(z) = 2z^2$, which will be a sample from the desired distribution. This can be done with $O(1)$ arithmetic computations.
\end{itemize}
\end{proof}

\subsection{Smoothness bounds for polynomials}
\label{subsec:poly_facts}
As mentioned our main techniques tool is a Bernstein type smoothness bounds for low-degree polynomials.
In general, low-degree polynomials are smoother than high-degree polynomials, and thus cannot spike as sharply.
There are a number of ways to formalize this statement. The well known Markov brother's inequality and Bernstein inequality bound the maximum derivative of a polynomial by a function of the polynomial's degree and it's maximum value on an interval.

To bound leverage scores, we are interested in a slightly different metric of smoothness. In particular, we need to bound the maximum squared value of a polynomial by its average squared value on $[0,T]$. We can use standard properties of the Legendre polynomials to prove:
s
\begin{claim}\label{claim:smoothness_for_middle}
	For any degree $d$ polynomial $p(\cdot)$ with complex coefficients and $t\in [0,T]$, let $r = \frac{\min\left(t,T-t\right)}{T}$. Then:
	\begin{align*}
	|p(t)|^2 \leq \frac{d+1}{\sqrt{r}}\cdot \frac{1}{T}\int_{0}^T |p(t)|^2 \, dt.
	\end{align*}
\end{claim}
\ifdraft
\Haim{A somewhat similar result appears as Theorem 7.71.1 in Szego's book ``orthogonal polynomials'', although without a full proof (but they do say that the two results we used to prove our bound are used to prove their bound). Theorem 7.71.2 is an interesting and potentially useful generalization.}
\fi

This bound is tighter for points near the center of the interval $[0,T]$ and goes to infinity near the edges. Using the Markov brother's inequality, it's possible to obtain a fixed up bound of $O(d^2)$, which is tighter for small values of $r$. However,  this won't be necessary for our purposes. We note that, when $t = T/2$, the upper bound on $p(t)^2$ improves to $O(d)$ times the average squared value of $p$, quadratically better than an $O(d^2)$ bound. This improvement is nearly optimal: the upper bound of Claim \ref{claim:smoothness_for_middle} is matched up to a logarithmic factor by an appropriately scaled and shifted Chebyshev polynomial of the first kind applied to $[T/2-t]^2$ (see e.g. \cite{FrostigMuscoMusco:2016} for a construction).
\begin{proof}[Proof of Claim \ref{claim:smoothness_for_middle}]
The claim follows from properties of the standard orthogonal Legendre polynomials, which are denoted $P_0, P_1, \ldots$ and defined via the recurrence relation:
\begin{align*}
	P_0(x) &= 1 \\
	P_1(x) &= x \\
	&\,\,\,\vdots \\
	P_k(x) &= \frac{2k-1}{k}x\cdot P_{k-1}(x) - \frac{n-1}{n}\cdot P_{k-2}(x).  \\
\end{align*}
The Legendre polynomials are orthogonal over the interval $[-1,1]$ with respect to the constant weight function. In particular, they satisfy
\begin{align}
\label{eq:ortho1}
	\int_{-1}^1 P_j(x) P_k(x)\,dx = \frac{2}{2j+1} \delta_{j,k},
\end{align}
where $\delta_{m,n}$ is the Kronecker delta function. Additionally, for $x \in [-1,1]$, $|P_j(x)| \leq 1$ for all $j$.

Using these facts we can show that for any degree $d$ polynomial $p(\cdot)$, interval $[a,b]$, and $x\in [a,b]$:
	\begin{align*}
	|p(x)|^2 \leq \frac{d+1}{\sqrt{r}}\cdot \frac{\int_{a}^b |p(t)|^2 \, dt}{(b-a)},
	\end{align*}
	where $r = \frac{\min\left(|a - x|,|b-x|\right)}{(b-a)}$. Setting $a = 0$ and $b = T$ gives the claim.
	
	We begin by noting that, without loss of generality, we can assume that $a = -1$ and $b = 1$. In particular, shift and stretch $p(x)$ by defining $g(x) = p\left(\frac{2(x-a)}{b-a} - 1\right)$. $g$ has degree $d$ and the maximum of $|g(x)|^2$ for $x\in [-1,1]$ is the same as the maximum of $|p(x)|^2$ for $x\in [a,b]$. Additionally, $\frac{\int_{-1}^1 |g(t)|^2 \, dt}{2} = \frac{\int_{a}^b |p(t)|^2 \, dt}{(b-a)}$. Accordingly, to prove the claim it suffices to prove that, for any degree $d$ polynomial $g$,
		\begin{align}
		\label{eq:bern_to_prove}
			\max_{x\in[-1,1]} |g(x)|^2 \leq \frac{d+1}{\sqrt{r}} \cdot \frac{\int_{-1}^1 |g(t)|^2 \, dt}{2}.
		\end{align}

	Our proof depends on a Bernstein type inequality for Legendre polynomials, which can be found in \cite{Lorch:1983}. Specifically, for all $j = 0,1,2,\ldots$ and any $x\in [-1,1]$ it holds that:
	\begin{align}
	\label{eq:legendre_bernstein}
		P_j(x)^2 \leq \frac{2}{\pi(j + 1/2)} \frac{1}{\sqrt{1-x^2}}.
	\end{align}
	Writing $g$ in the Legendre basis:
	\begin{align*}
		g(x) = \sum_{j=0}^d c_j P_j(x),
	\end{align*}
	we have from \eqref{eq:legendre_bernstein} that
	\begin{align*}
		|g(x)| \leq \sum_{j=0}^d |c_j| \left(\frac{2}{\pi(j + 1/2)} \frac{1}{\sqrt{1-x^2}}\right)^{1/2}
	\end{align*}
	and thus
	\begin{align}
	\label{eq:first_bound_with_sqrt}
		|g(x)|^2 &\leq (d+1)\sum_{j=0}^d |c_j|^2 \frac{2}{\pi(j + 1/2)} \frac{1}{\sqrt{1-x^2}} \nonumber\\
		 			&= \frac{2}{\pi}\frac{(d+1)}{\sqrt{1-x^2}}\sum_{j=0}^d |c_j|^2 \frac{2}{2j + 1}\nonumber\\
		 			&=  \frac{2}{\pi}\frac{(d+1)}{\sqrt{1-x^2}} \int_{-1}^1 |g(t)|^2 \, dt.
	\end{align}
	The last equality step follows from \eqref{eq:ortho1}.
	Finally, let $q = \min\left(|-1-x|,|1-x|\right)$ and note that
	\begin{align*}
		\frac{1}{\sqrt{1-x^2}} = \frac{1}{\sqrt{1-(1-q)^2}} \leq \frac{1}{\sqrt{q}}.
	\end{align*}
	As defined, $r = q/2$ Plugging into \eqref{eq:first_bound_with_sqrt} we have a final bound of
	\begin{align*}
		|g(x)|^2 \leq \frac{4}{\pi}\frac{(d+1)}{\sqrt{2r}} \frac{\int_{-1}^1 |g(t)|^2 \, dt}{2} < \frac{(d+1)}{\sqrt{r}} \frac{\int_{-1}^1 |g(t)|^2 \, dt}{2},
	\end{align*}
	which establishes \eqref{eq:bern_to_prove} and thus the claim.
\end{proof}

\subsection{Smoothness bounds for bandlimited functions}
With Claim \ref{claim:smoothness_for_middle} in place, we are now ready to prove our main result for bandlimited functions.
\begin{proof}[Proof of Theorem \ref{thm:bandlimited_leverage_scores}]
	Following Definition \ref{def:ridgeScores}, our goal is to choose $\ttmu$ to satisfy:
	\begin{align}
	\label{eq:bandlimited_goal}
			\ttmu(t) \geq \frac{1}{T} \cdot \frac{|[\Fmu \alpha](t) |^2 }{\norm{\Fmu \alpha}_T^2 + \epsilon \norm{\alpha}_\mu^2}.
	\end{align}
	for any $\alpha$. Let $z = \Fmu\alpha$.
	Expanding $e^{-2 i \pi \xi t}$ using its Maclaurin series and letting $d$ be some degree parameter that we will fix later, we write $z$ as the sum of two functions, $a$ and $b$:
	\begin{align}\label{eq:zab}
	z(t)  &= \frac{1}{2F}\int_{-F}^F \alpha(\xi) e^{-2 i \pi \xi t} \, d\xi \nonumber\\
	& = \sum_{j=0}^\infty \frac{1}{2F}\int_{-F}^F \alpha(\xi) \frac{(-2\pi i \xi)^j}{j!}t^j \, d\xi \nonumber\\
	& = \sum_{j=0}^d \left(\frac{1}{2F}\int_{-F}^F \alpha(\xi) \frac{(-2\pi i \xi)^j}{j!} \, d\xi\right) t^j + \sum_{j={d+1}}^\infty \frac{1}{2F}\int_{-F}^F \alpha(\xi) \frac{(-2\pi i\xi)^j}{j!}t^j \, d\xi \nonumber \\
	&\eqdef a(t) + b(t).
	\end{align}
	Note that $a$ is a degree $d$ polynomial with complex coefficients. So by Claim \ref{claim:smoothness_for_middle},
	\begin{align}
	\label{eq:a_bound}
	|a(t)|^2 \leq \frac{d+1}{{\sqrt{\min(t,T-t)/T}}} \cdot\|a\|_T^2.
	\end{align}
	Turning our attention to $b$, we see that:
	\begin{align}
		|b(t)| &= \left|\sum_{j={d+1}}^\infty \frac{1}{2F}\int_{-F}^F \alpha(\xi) \frac{(-2\pi i\xi)^j}{j!}t^j \, d\xi\right| \leq \sum_{j={d+1}}^\infty \frac{(2\pi FT)^j}{j!}  \left|\frac{1}{2F} \int_{-F}^F \alpha(\xi)  \, d\xi \right| \nonumber\\
		&\leq \sum_{j={d+1}}^\infty \frac{(2\pi FT)^j}{j!} \sqrt{\frac{1}{2F} \int_{-F}^F 1\, d\xi }\sqrt{\|\alpha\|_\mu^2} = \sum_{j={d+1}}^\infty \frac{(2\pi FT)^j}{j!} \cdot \|\alpha\|_\mu .
		\label{eq:first_b_bound}
	\end{align}
	The second to last step uses Cauchy-Schwarz inequality.
	Finally using that for all $j$, $j! \ge (j/e)^j$, for any $d \geq 4\pi e FT$:
	\begin{align}
		\sum_{j={d+1}}^\infty \frac{(2\pi FT)^j}{j!}  &\leq \sum_{j={d+1}}^\infty \left(\frac{2\pi e FT}{j}\right)^j \nonumber\\
		&\leq \sum_{j={d+1}}^\infty \left(\frac{2\pi e FT}{d+1}\right)^j\nonumber\\
		&\le \sum_{j={d+1}}^\infty \left(\frac{1}{2}\right)^j = \frac{1}{2^{d}}.
		\label{eq:stirling}
	\end{align}
So, if we take $d  = \lceil 4\pi e FT + \log(1/\epsilon)/2  + 1 \rceil $, it follows from \eqref{eq:first_b_bound} and \eqref{eq:stirling} that 
\begin{align*}
|b(t)| \le \frac{1}{2^{d}} \cdot \norm{\alpha}_{\mu} \le \frac{1}{2^{\lceil \log(1/\epsilon)/2 \rceil+1 }} \cdot \norm{\alpha}_{\mu} \le \frac{\sqrt{\epsilon}}{2} \cdot \norm{\alpha}_{\mu}.% \le \frac{\epsilon}{4} \cdot \norm{\alpha}_{\mu}^2.
\end{align*}
It follows that $\|b\|_T \leq \frac{\sqrt{\epsilon}}{2}\|\alpha\|_\mu$. Using the decomposition of \eqref{eq:zab} and the fact that for any  real nonnegative $c,d$, $c^2 + d^2 \le (c+d)^2$, and for any complex $e,f$, $|e + f|^2 \le 2|e|^2 + 2|f|^2$:
\begin{align*}
\frac{|z(t)|^2}{\|z\|_T^2 + \epsilon\|\alpha\|_\mu^2} &\leq \frac{|a(t) + b(t)|^2}{(\norm{a}_T  - \|b\|_T)^2 + \epsilon\|\alpha\|_\mu^2}\\
& \leq \frac{2|a(t)|^2 + 2|b(t)|^2}{\frac{1}{2}(\norm{a}_T - \norm{b}_T + \sqrt{\epsilon} \norm{\alpha}_\mu)^2}\\
 &\leq \frac{4|a(t)|^2 + 4|b(t)|^2}{(\norm{a}_T + \frac{\sqrt{\epsilon}}{2} \norm{\alpha}_\mu)^2} \\
&\le \frac{4|a(t)|^2 + \epsilon \norm{\alpha}_\mu^2}{\norm{a}_T^2 + \frac{\epsilon}{4} \norm{\alpha}_\mu^2}.
\end{align*}
It follows from \eqref{eq:a_bound} that:
\begin{align*}
\frac{|z(t)|^2}{\|z\|_T^2 + \epsilon\|\alpha\|_\mu^2} &\le \max \left (\frac{4 |a(t)|^2}{\norm{a}_T^2}, 4 \right )\\
&\leq   \frac{4(d+1)}{\sqrt{\min(t,T-t)/T}} +4.
\end{align*}
In Theorem \ref{thm:bandlimited_leverage_scores} we set $q = \lceil 16\pi e FT + 2\log(1/\epsilon) +11 \rceil$. We have $q \geq 4\cdot \lceil 4\pi e FT + \log(1/\epsilon)/2 +2 \rceil = 4(d+1)$ since, for any $x$,  $\lceil 4x + 3\rceil \ge 4 \lceil x \rceil$.
Recalling that $z = \Fmu \alpha$, it follows $\ttmu$ defined in that theorem satisfies \eqref{eq:bandlimited_goal} for any $\alpha$.
It remains to bound the total measure of our approximate ridge leverage function, $\tsmu$. To do so, note that:
\begin{align*}
\tsmu= \frac{2}{T}\int_{0}^{T/2} \frac{q}{\sqrt{t/T}} +4 \, dt.
\end{align*}
%where $q =  2\cdot \lceil 4\pi e FT + \log(1/\epsilon) \rceil + 4$.
We can compute:
\begin{align*}
\frac{2}{T}\int_{0}^{T/2} \frac{q}{\sqrt{t/T}} + 4 \, dt = 2\int_{0}^{1/2} \frac{q}{\sqrt{t}} + 4 \, dt = 2\sqrt{2} q + 4 = O(FT + \log(1/\epsilon)).
\end{align*}
	This bound establishes the theorem.
\end{proof}

\section{Statistical dimension for common Fourier constraints}
\label{app:stat_dimension}

In this section we leverage Theorem \ref{thm:bandlimited_leverage_scores}  to give upper bounds on the statistical dimensions of a number common priors $\mu$ used for Fourier constrained interpolation, including multiband, Gaussian, and Cauchy-Lorentz priors. We start by giving two simple lemmas that we use to translate our bound for bandlimited functions to these more general priors.

\begin{lemma}[Statistical Dimension of Sum of Measures] \label{disjoint-measure-levscore}
Let $\mu_1, \mu_2, \cdots \mu_s$ be finite measures on $\RR$. Let $\mu$ be a probability measure defined by $\mu = \mu_1+ \mu_2+ \cdots+ \mu_s$.
\begin{align*}
\smu \le \sum_{i=1}^s s_{\mu_i,\epsilon}.
\end{align*}
\end{lemma}
\begin{proof}
We can see from Definition \ref{def:statDim} that for $\mu = \mu_1+\ldots + \mu_s$ the kernel operator $\Kmu$ satisfies $\Kmu = \sum_{i=1}^s \mathcal{K}_{\mu_i}$. 
We can thus bound:
\begin{align*}
\smu = \tr(\Kmu(\Kmu+\epsilon \mathcal{I}_T)^{-1}) &= \sum_{i=1}^s \tr(\mathcal{K}_{\mu_i}(\Kmu+\epsilon \mathcal{I}_T)^{-1}) \\
&\le \sum_{i=1}^s \tr(\mathcal{K}_{\mu_i}(\mathcal{K}_{\mu_i}+\epsilon \mathcal{I}_T)^{-1})\\
&=\sum_{i=1}^s s_{\mu_i,\epsilon}.
\end{align*}
The second to last inequality follows since $0 \preceq\mathcal{K}_{\mu_i} \preceq \Kmu$, so $0 \prec \mathcal{K}_{\mu_i} + \epsilon \mathcal{I}_T \preceq \Kmu + \epsilon \mathcal{I}_T $ and $(\Kmu+\epsilon \mathcal{I}_T)^{-1} \preceq (\mathcal{K}_{\mu_i}+\epsilon \mathcal{I}_T)^{-1}$  by  Claim~\ref{clm:invordering}. Letting $e_1,e_2$ be an orthonormal basis for $L_2(T)$, we thus have:
\begin{align*}
\tr(\mathcal{K}_{\mu_i}(\mathcal{K}_{\mu_i}+\epsilon \mathcal{I}_T)^{-1}) &= \tr(\mathcal{K}_{\mu_i}^{1/2}(\mathcal{K}_{\mu_i}+\epsilon \mathcal{I}_T)^{-1}\mathcal{K}_{\mu_i}^{1/2})\tag{By cyclic property  of the trace, Claim \ref{clm:trace-sqrt}.}\\
&= \sum_{i=1}^\infty \langle \mathcal{K}_{\mu_i}^{1/2} e_i ,(\mathcal{K}_{\mu_i}+\epsilon \mathcal{I}_T)^{-1}\mathcal{K}_{\mu_i}^{1/2} e_i\rangle_T\\
&\ge \sum_{i=1}^\infty \langle \mathcal{K}_{\mu_i}^{1/2} e_i ,(\Kmu+\epsilon \mathcal{I}_T)^{-1}\mathcal{K}_{\mu_i}^{1/2} e_i\rangle_T\\
& = \tr(\mathcal{K}_{\mu_i}^{1/2}(\Kmu+\epsilon \mathcal{I}_T)^{-1}\mathcal{K}_{\mu_i}^{1/2})\\
& = \tr(\mathcal{K}_{\mu_i}(\Kmu+\epsilon \mathcal{I}_T)^{-1}).
\end{align*}
This completes the lemma.
\end{proof}

\begin{lemma}[Statistical Dimension of Scaled Measures]\label{lem:scaling-measure-statdim}
Let $\mu$ be a measure on $\RR$. For any parameter $\gamma >0$, we have:
$$\smu = s_{(\mu / \gamma ),(\epsilon/\gamma )}.$$
\end{lemma}
\begin{proof}
From Definition \ref{def:statDim}, we can see that $\mathcal{K}_{(\mu/\gamma)} = \frac{1}{\gamma}\Kmu$ and thus  has eigenvalues equal to $\lambda_1(\Kmu)/\gamma, \lambda_2(\Kmu)/\gamma, \ldots$ 
We can thus compute:
\begin{align*}
s_{(\mu/\gamma),(\epsilon/\gamma)} &= \sum_{i=1}^\infty \frac{\lambda_i(\mathcal{K}_{(\mu/\gamma)})}{\lambda_i(\mathcal{K}_{(\mu/\gamma)})+\epsilon/\gamma}\\
&=\sum_{i=1}^\infty \frac{\lambda_i(\Kmu)/\gamma}{\lambda_i(\Kmu)/\gamma+\epsilon/\gamma}\\
&=\sum_{i=1}^\infty \frac{\lambda_i(\Kmu)}{\lambda_i(\Kmu)+\epsilon}\\
&= \smu.
\end{align*}
\end{proof}

We now use Lemmas \ref{disjoint-measure-levscore} and \ref{lem:scaling-measure-statdim} to prove our statistical dimension bounds.
We first start with multiband Fourier constraints, showing that the statistical dimension is roughly proportional to the total length of all the frequency bands times the time domain window size, intuitively matching the Landau rate for asymptotic recovery of multiband functions \cite{Landau:1967a}.

\begin{theorem}[Multiband Statistical Dimension]\label{theorem-multiband-statdim}
Consider a set of $s$ disjoint frequency bands, $I_1, I_2, \cdots, I_s$, and suppose that the length of the band $I_i$ is denoted by $F_i$. Let $\mu$ be the measure which induces a uniform probability density on $I_1 \cup I_2 \cup \cdots \cup I_s$. We have:
$$ \smu = O\left(\sum_{i=1}^s F_i T + s\log(1/\epsilon)\right).$$
\end{theorem}
\begin{proof}
For every $i$, let $\mu_i$ be the measure defined by $\mu_i(A) = \mu(A \cap I_i)$. Note that we have $\mu = \sum_i \mu_i$ and so can invoke Lemma \ref{disjoint-measure-levscore}, giving:
\begin{equation}\label{eq:statdim-subaditiv}
\smu \le \sum_{i=1}^s s_{\mu_i,\epsilon}.
\end{equation}

If $\mu_i$ gave a uniform probability measure on frequency band $I_i$ (i.e., if we had $\mu_i(\RR) = 1$), we could use the result of Theorem \ref{thm:bandlimited_leverage_scores} to bound $s_{\mu_i,\epsilon} = O(F_iT + \log (1/\epsilon))$. This is not the case, but we can instead let $\gamma_i \eqdef \mu_i(\RR) \le 1$. By Lemma \ref{lem:scaling-measure-statdim},
$$s_{\mu_i,\epsilon} = s_{(\mu_i / \gamma_i),(\epsilon/\gamma_i)}.$$
Now $\mu_i / \gamma_i$ is a uniform probability measure on $I_i$, so we can invoke Theorem \ref{thm:bandlimited_leverage_scores} giving:
\begin{align*}
s_{\mu_i,\epsilon} &= s_{(\mu_i/ \gamma_i),(\epsilon/\gamma_i)} = O\left( F_i T + \log(\gamma_i/\epsilon) \right).
\end{align*}
Plugging this bound in \eqref{eq:statdim-subaditiv} and using that $\gamma_i \le 1$ we obtain:
\begin{align*}
\smu &= O\left( \sum_{i=1}^s  F_i T + \log(\gamma_i/\epsilon)\right) = O\left(  \sum_{i=1}^s F_i T  + s\log(1/\epsilon) \right),
\end{align*}
completing the theorem.\end{proof}

\ifdraft
\Haim{As a side, seems like the technique we use in the next theorem can be used to bound the statistical dimension of any prior based only on how it decays. e.g. if the corresponding density decays exponentially, then the stat dim is essentially like the Gaussian.}\Cam{Agreed. Basically  you just need to be able to 'cover the distribution' with rectangles whose  area is falling off.}
\fi
We next bound the statistical dimension of Gaussian measure.
\begin{theorem}[Gaussian Statistical Dimension]\label{thm:guassianStatDim}
Let $\mu$ induce the Gaussian probability distribution with standard deviation $F$ defined by $d \mu(\xi) = \frac{1}{\sqrt{2\pi F^2}} e^{-\xi^2 / 2F^2} d\xi$.
We have:
$$ \smu = O\left( F T \sqrt{\log(1/\epsilon)} + \log(1/\epsilon) \right).$$
\end{theorem}
\begin{proof}
Let $I_h$ be the interval defined by $I_h = \{\xi \in \RR : |\xi| \le F\sqrt{\log(1/\epsilon)}\}$.
We decompose $\mu$ into two measures $\mu_h$ and $\mu_t$ as follows:
$$\mu_h(A) = \mu(A \cap I_h)\text{ and }\mu_t(A) = \mu(A - A\cap I_h).$$
We can see that $\mu = \mu_h+\mu_t$ and so by Lemma \ref{disjoint-measure-levscore},
$\smu \le s_{\mu_t,\epsilon} + s_{\mu_h,\epsilon}.$
For $\mu_t$ we have:
\begin{align*}
 \tr(\mathcal{K}_{\mu_t}) = \mu_t(\RR) &= \frac{1}{\sqrt{2\pi F^2}}\int_{|\xi| > F \sqrt{\log(1/\epsilon)}}  e^{-\xi^2 / 2F^2} \, d\xi\\
 &= 1-\erf(\sqrt{\log(1/\epsilon)}) \le 2\epsilon,
\end{align*}
where the last bound follows from a Chernoff bound, giving $1-\erf(x) \le 2e^{-x^2}$ \cite{wainwright2015high}. This lets us crudely bound:
\begin{equation} \label{statdim-tail}
s_{\mu_t,\epsilon} = \tr(\mathcal{K}_{\mu_t}(\mathcal{K}_{\mu_t}+\epsilon \mathcal{I}_T)^{-1}) \le \tr(\mathcal{K}_{\mu_t}) / \epsilon \le 2,
\end{equation}
where the first ineguality is because $\opnorm{(\mathcal{K}_{\mu_t}+\epsilon \mathcal{I}_T)^{-1}}\leq 1/\epsilon$.

We next bound the statistical dimension of $\mu_h$. Let $\tilde{\mu}_h$ be a uniform measure on $I_h$, with $d\mu(\xi) = \frac{1}{\sqrt{2\pi F^2}} d \xi$ for all $\xi \in I_h$. Note that $d \tilde{\mu}_h(\xi) \ge d {\mu}_h(\xi)$ for all  $\xi \in I_h$ which gives that $K_{\mu_h} \preceq K_{{\tilde \mu}_h}$ and so $s_{\mu_h,\epsilon} \leq s_{\tilde{\mu}_h,\epsilon}.$

Let $\gamma \eqdef \tilde{\mu}_h(\RR) = \sqrt{\frac{2\log(1/\epsilon)}{\pi}}$. By Lemma \ref{lem:scaling-measure-statdim}, $s_{\tilde{\mu},\epsilon} = s_{(\tilde{\mu} / \gamma ),(\epsilon/\gamma )}.$
Since $\tilde{\mu} / \gamma$ is a uniform probability measure on $I_h$, we can invoke Theorem \ref{thm:bandlimited_leverage_scores} to give:
\begin{align}
s_{{\mu}_h,\epsilon} \le s_{\tilde{\mu}_h,\epsilon} = s_{(\tilde{\mu}_h / \gamma ),(\epsilon/\gamma )} \nonumber &= O\left( F T \sqrt{\log(1/\epsilon)} + \log(\gamma/\epsilon) \right) \nonumber\\
&= O\left( F T \sqrt{\log(1/\epsilon)} + \log(1/\epsilon) \right), \label{statdim-head}
\end{align}
where the last equality follows from the fact that $\gamma = O(\sqrt{\log(1/\epsilon)})$. Combining \eqref{statdim-tail} and \eqref{statdim-head} and applying Lemma \ref{disjoint-measure-levscore} we have:
\begin{align*}
\smu &\le s_{\mu_t,\epsilon} + s_{\mu_h,\epsilon}\\
&= 2 + O\left( F T \sqrt{\log(1/\epsilon)} + \log(1/\epsilon) \right)\\
&= O\left( F T \sqrt{\log(1/\epsilon)} + \log(1/\epsilon) \right),
\end{align*}
which completes the theorem.
\end{proof}

Finally, we bound the statistical dimension of the Cauchy-Lorentz measure.

\begin{theorem}\label{thm:lorentzStatDim}
Let $\mu$ induce the Cauchy-Lorentz probability distribution with scale parameter $F$ defined by $d\mu(\xi) = \frac{1}{\pi F \left[ 1+\left(\frac{\xi}{F}\right)^2 \right]} d\xi$. We have:
	$$ \smu = O\left( \frac{F T }{\sqrt{\epsilon}} + \frac{1}{\sqrt{\epsilon}} \right).$$
\end{theorem}
\begin{proof}
As in the proof of Theorem \ref{thm:guassianStatDim} we define $I_h$ to be the interval $I_h = \{ \xi \in \RR: |\xi| \le F/\sqrt{\epsilon}\}$. We decompose $\mu$ into two measures $\mu_h$ and $\mu_t$ as follows:
	$$\mu_h(A) = \mu(A \cap I_h)\text{ and }\mu_t(A) = \mu(A - A\cap I_h).$$
	We have $\mu = \mu_h + \mu_t$ by Lemma \ref{disjoint-measure-levscore},
$\smu \le s_{\mu_t,\epsilon} + s_{\mu_h,\epsilon}.$
For $\mu_t$ we have:
\begin{align*}
 \tr(\mathcal{K}_{\mu_t}) = \mu_t(\RR) &= \frac{1}{\pi F}\int_{|\xi| > F /\sqrt{\epsilon}}  \frac{1}{1+\left(\frac{\xi}{F}\right)^2} \, d\xi\\
 &= \frac{2}{\pi} \int_{1/\sqrt{\epsilon}}^{\infty}  \frac{1}{1+\xi^2} \, d\xi\\
 &\le \frac{2}{\pi}  \int_{1/\sqrt{\epsilon}}^{\infty}  \frac{1}{\xi^2} \, d\xi = \frac{2 \sqrt{\epsilon}}{\pi}.
 \end{align*}
As in \eqref{statdim-tail} we can thus bound:
	\begin{equation} \label{cauchy-statdim-tail}
	s_{\mu_t,\epsilon} \le \tr(\mathcal{K}_{\mu_t}) / \epsilon = O(1/\sqrt{\epsilon}).
	\end{equation}
	We next bound the statistical dimension of $\mu_h$. Let $\tilde{\mu}_h$ be a uniform measure on $I_h$ with $d\mu(\xi) = \frac{1}{\pi F}$ for all $\xi \in I_h$. As in the proof of Theorem \ref{thm:guassianStatDim}, $d \tilde{\mu}_h(\xi) \ge d {\mu}_h(\xi)$ for all  $\xi \in I_h$ which gives that $K_{\mu_h} \preceq K_{{\tilde \mu}_h}$ and so $s_{\mu_h,\epsilon} < s_{\tilde{\mu}_h,\epsilon}.$

Let $\gamma \eqdef \tilde{\mu}_h(\RR) = \frac{2}{\pi \sqrt{\epsilon}}$. By Lemma \ref{lem:scaling-measure-statdim}, $s_{\tilde{\mu},\epsilon} = s_{(\tilde{\mu} / \gamma ),(\epsilon/\gamma )}.$
Since $\tilde{\mu} / \gamma$ is a uniform probability measure on $I_h$, we can invoke Theorem \ref{thm:bandlimited_leverage_scores} to give:
\begin{align}
s_{{\mu}_h,\epsilon} \le s_{\tilde{\mu}_h,\epsilon} = s_{(\tilde{\mu}_h / \gamma ),(\epsilon/\gamma )} \nonumber &= O\left( \frac{F T}{\sqrt{\epsilon}} + \log(\gamma/\epsilon) \right) \nonumber\\
&= O\left( \frac{F T}{\sqrt{\epsilon}} + \log(1/\epsilon) \right), \label{cauchy-statdim-head}
\end{align}
where the last equality follows from the fact that $\gamma = O(1/\sqrt{\epsilon})$. Combining \eqref{cauchy-statdim-tail} and \eqref{cauchy-statdim-head} and applying Lemma \ref{disjoint-measure-levscore} we have:
\begin{align*}
\smu &\le s_{\mu_t,\epsilon} + s_{\mu_h,\epsilon} = O\left(\frac{1}{\sqrt{\epsilon}}+ \frac{F T}{\sqrt{\epsilon}} + \log(1/\epsilon) \right) =O\left(\frac{F T}{\sqrt{\epsilon}} + \frac{1}{\sqrt{\epsilon}} \right),
\end{align*}
which completes the theorem.
\end{proof}

\section{Kernel computation for common Fourier constraints}
\label{app:kernel_computation}

Algorithm \ref{alg:main} and the corresponding Theorem \ref{thm:informal_main} assumes the ability to compute the kernel function $k_\mu(t_1,t_2) = \int_{\xi \in \RR}  e^{-2\pi i (t_1 - t_2)} d\mu(\xi)$. In this section we give close forms for the kernel functions of popular measures $\mu$, including all those whose statistical dimension we bound in Appendix \ref{app:stat_dimension}.

\paragraph{Bandlimited Fourier Constraint:}
When $\mu$ is the uniform measure on frequencies in $[-F, F]$,
\begin{align*}
k_\mu(t_1,t_2) &= \int_{\xi \in \RR}  e^{-2\pi i (t_1 - t_2) \xi} d\mu(\xi)\\
&= \frac{1}{2F} \int_{-F}^F  e^{-2\pi i (t_1 - t_2) \xi} d\xi\\
&= \frac{\sin(2\pi F(t_1-t_2))}{2\pi F(t_1-t_2)}.
\end{align*}
So, $k_\mu$ is the {\em sinc kernel}.

\paragraph{Multiband Fourier Constraint:}
Consider a set of $s$ disjoint frequency bands, $I_1, I_2, \cdots, I_s$, where $I_j = [c_j - F_j, c_j + F_j]$. Let $\mu$ be the uniform measure on $I_1 \cup I_2 \cup \cdots \cup I_s$. Then we have:
\begin{align*}
k_\mu(t_1,t_2) &= \int_{\xi \in \RR}  e^{-2\pi i (t_1 - t_2)\xi} d\mu(\xi)\\
&= \frac{1}{2 \sum_j F_j}\cdot \sum_{j} {e^{-2\pi i c_j (t_1 - t_2)\xi}} \int_{-F_j}^{F_j}  e^{-2\pi i (t_1 - t_2)\xi} d\xi\\
&= \frac{1}{2 \pi \sum_j F_j (t_1-t_2)}\sum_j e^{-2\pi i c_j (t_1 - t_2)} \cdot \sin(2\pi F_j(t_1-t_2)).
\end{align*}

\paragraph{Gaussian Fourier Constraint:}
When $\mu$ induces the Gaussian probability distribution with standard deviation $F$ defined by $d \mu(\xi) = \frac{1}{\sqrt{2\pi F^2}} e^{-\xi^2 / 2F^2} d\xi$, then
\begin{align*}
k_\mu(t_1,t_2) &= \int_{\xi \in \RR}  e^{-2\pi i (t_1 - t_2)\xi} d\mu(\xi)\\
&= \frac{1}{\sqrt{2\pi F^2}} \cdot \int_{\xi \in \RR}  e^{-2\pi i (t_1 - t_2)\xi}  e^{-\xi^2 / 2F^2} d\xi\\
&= e^{-2\pi^2 F^2 (t_1-t_2)^2}.
\end{align*}
So, $k_\mu$ is the {\em Gaussian kernel}.

%\paragraph{Gaussian Mixture Fourier Constraint:}
%Consider a set of $s$ means $c_1,\ldots,c_s$ and $s$ standard deviations $F_1,\ldots,F_s$. Let $\mu$ be

\paragraph{Cauchy-Lorentz Fourier Constraint:}
When $\mu$ induces the Cauchy-Lorentz probability density with scale parameter $F$ defined by $d\mu(\xi) = \frac{1}{\pi F \left[ 1+\left(\frac{\xi}{F}\right)^2 \right]} \cdot d\xi$, we have:
\begin{align*}
k_\mu(t_1,t_2) &= \int_{\xi \in \RR}  e^{-2\pi i (t_1 - t_2)} d\mu(\xi)\\
&= \int_{\xi \in\RR}  e^{-2\pi i (t_1 - t_2)}  \frac{1}{\pi F \left[ 1+\left(\frac{\xi}{F}\right)^2 \right]}d\xi\\
&= e^{-2\pi F |t_1-t_2|}.
\end{align*}
In the machine learning literature, $k_\mu$ is known as the  {\em Laplacian kernel}.

%\section{Facts about polynomials}
%\label{app:poly_facts}

\section{Signal Reconstruction as Bayesian Estimation}\label{app:bayes}

In this section, we show how,
as an alternative to Problem \ref{prob:unformal_interp}, we can formulate signal fitting as a Bayesian estimation problem, where the signal $y$ is a stationary stochastic process and  the measure $\mu$ (which we assume to be symmetric about $0$ throughout this section so that $k_\mu(t_1,t_2)$ is real valued) corresponds to a prior on $y$'s power spectral density.
%
%We prove that, in the limit as the number of time domain samples goes to infinity, both the maximum a posteriori (MAP) and Bayes minimum mean squared error (MMSE) estimates are given by  solving the same least squares problem of \eqref{eq:least_squares_setup}. Via Theorem \ref{thm:informal_main}, this optimization problem can be solved approximately with $\tilde O(\smu)$ samples
 %using Algorithm \ref{alg:main}  and the universal sampling distribution of Theorem \ref{thm:fullBound}. Via Claim \ref{claim:regression_reduction} one can see that the lower bound of Section \ref{sec:lb}, Theorem \ref{thm:mainLB} extends to solving \eqref{eq:least_squares_setup}, even approximately, and thus our sample complexity is nearly optimal.
%
% \todo{In the end, we may want to formalize the problem differently. In the GP literature its well known that the MAP estimate corresponds to the Bayes MMSE estimate, which also corresponds to estimating $y$ by the mean of the posterior. So maybe we could formalize the problem as approximately minimizing expected mean squared error from $y$ over the join distribution of $y$ coming from the prior and a randomly chosen $t \in [0,T]$? This would be quite natural. Basically, whichever natural Bayesian estimation  problem works, I'm happy with. Not tied to MAP specifically.}
%
 %We begin by defining the prior distribution on our signal $y$ that we will consider when trying to learn $y$ from time domain samples.
 This form of prior is commonly used in statistical signal processing, kriging, and machine learning applications \cite{HandcockStein:1993,Ripley:2005,RasmussenWilliams06}:
 We first define a stationary Gaussian process:
  \begin{defn}[Stationary Gaussian Process \cite{RasmussenWilliams06}]\label{def:gaussProcess1}
  A stochastic process $y: \RR \rightarrow \RR$ is a \emph{Gaussian process} if for any finite collection $t_1,\ldots t_s \in \RR$, $y(t_1),\ldots, y(t_s)$ is distributed as a multivariate Gaussian. $y$ is a \emph{stationary Gaussian process} if the mean $\E[y(t)]$ is independent of $t$ and the autocorrelation $\E[  y(t_1)  \cdot y(t_2)]$ depends only on $t_1-t_2$. 
   \end{defn}
 
We now define the specific Gaussian process prior we consider: 
 \begin{defn}[Gaussian Process Prior]\label{def:gaussProcess}
 Consider a symmetric probability density function $p_\mu: \RR \rightarrow \RR^+$ and the associated measure $\mu$ corresponding to $p_\mu$. We say that a stochastic process $y: \RR\rightarrow \RR$ is distributed according to $\mathcal{D}_\mu$ if $y$ is distributed as a stationary Gaussian process (Definition \ref{def:gaussProcess1}) with with mean $\E[ y(t)] = 0$ and autocorrelation function $\E[  y(t_1)  \cdot y(t_2)] = k_\mu(t_1,t_2)$ for any $t_1,t_2$, where $k_\mu$ is defined in \eqref{eq:kernel_function}.
 \end{defn}

 As discussed, the prior of Definition \ref{def:gaussProcess} amounts to a prior on the power spectral density of $y$, with the expected power spectral density given by $p_\mu$. Formally:
%\todo{I think Claim \ref{clm:gp} should be formally correct, but it needs to be verified. I am following the argument around page 281 of \texttt{https://ee.stanford.edu/~gray/sp.pdf} which shows how to compute the output distribution of a linear filter (e.g. the Fourier transform) applied to a Gaussian process.}
 \begin{claim}[Equivalent Power Spectral Density Prior]\label{clm:gp}
 Consider $y$ distributed as in Definition \ref{def:gaussProcess}. Suppose that $p_\mu$ is bounded. For every $T>0$, let $\hat y_T: \RR \rightarrow \RR$ be the truncated Fourier transform of $y$, a.k.a. the amplitude spectral density of $y$:
 $$\hat {y}_T(\xi) \eqdef \frac{1}{\sqrt{T}}\int_{-T/2}^{T/2} y(t) e^{-2\pi i {t}{\xi}}\,d{t}.$$

 For every $T$, $\hat y_T$ is a Gaussian process with $\E[\hat y(t)] = 0$. Also as $T$ goes to infinity, the covariance of $\hat y_T$ converges to a diagonal covariance given by $p_\mu$. That is, for any $\xi_1,\ldots,\xi_s \in \RR$, $\lim\limits_{T \rightarrow \infty} [\hat y_T (\xi_1),\ldots,\hat y_T(\xi_s)] \sim \mathcal{N}(0,\bv{P})$ where $\bv{P}$ is a diagonal matrix with $\bv{P}_{i,i} = p_\mu(\xi_i)$.
 \end{claim}
\begin{proof}
$\hat y$ is a Gaussian process since it is a linear transformation of a Gaussian process, $y$~\cite{rasmussen2004gaussian}.
We first check that for every $T$, the mean of this random process is zero at every point $\xi$.

\begin{align*}
\E[\hat y_T(\xi)] &= \E \left[ \frac{1}{\sqrt{T}}\int_{-T/2}^{T/2} y(t) e^{-2\pi i {t}{\xi}}\,d{t} \right]\\
&= \frac{1}{\sqrt{T}}\int_{-T/2}^{T/2} \E [ y(t) ] e^{-2\pi i {t}{\xi}}\,d{t} \\
&= 0,
\end{align*}
where the application of Fubini's theorem in second line above is valid because $k_\mu(0) = 1$ and hence for every $t\in \RR$, $\E [ |y(t)| ] < \infty$.
Now in order to show that the covariance of $\hat y_T$ converges to being diagonal we check the covariance of $\hat y_T$ at two arbitrary points $\xi_1 \neq \xi_2$, as $T$ goes to infinity,
$$ \lim\limits_{T \rightarrow \infty} \E[ \hat y_T(\xi_1)^* \hat y_T(\xi_2) ] = \lim\limits_{T \rightarrow \infty} \E \left[ \frac{1}{{T}}\int_{-T/2}^{T/2} \int_{-T/2}^{T/2} y(t_1) e^{2\pi i {t_1}{\xi_1}} y(t_2) e^{-2\pi i {t_2}{\xi_2}}\,d{t_1}d{t_2} \right]$$
Now note that for every fixed $T$,
\begin{align*}
\E \left[ \frac{1}{{T}}\int_{-T/2}^{T/2} \int_{-T/2}^{T/2} y(t_1) e^{2\pi i {t_1}{\xi_1}} y(t_2) e^{-2\pi i {t_2}{\xi_2}}\,d{t_1}d{t_2} \right] &=  \frac{1}{{T}}\int_{-T/2}^{T/2} \int_{-T/2}^{T/2} \E [ y(t_1)y(t_2) ] e^{2\pi i {t_1}{\xi_1}} e^{-2\pi i {t_2}{\xi_2}}\,d{t_1}d{t_2}\\
&= \frac{1}{{T}}\int_{-T/2}^{T/2} \int_{-T/2}^{T/2} k_\mu(t_1,t_2) e^{2\pi i {t_1}{\xi_1}} e^{-2\pi i {t_2}{\xi_2}}\,d{t_1}d{t_2}
\end{align*}
Therefore,
\begin{align}
\lim\limits_{T \rightarrow \infty} \E[ \hat y_T(\xi_1)^* \hat y_T(\xi_2) ] &=  \lim\limits_{T \rightarrow \infty} \frac{1}{{T}}\int_{-T/2}^{T/2} \int_{-T/2}^{T/2} k_\mu(t_1,t_2) e^{2\pi i {t_1}{\xi_1}} e^{-2\pi i {t_2}{\xi_2}}\,d{t_1}d{t_2} \nonumber\\
&= \lim\limits_{T \rightarrow \infty} \frac{1}{{T}}\int_{-T/2}^{T/2} \int_{-T/2-t}^{T/2-t} k_\mu(\tau) e^{2\pi i {t}{\xi_1}} e^{-2\pi i (t+\tau){\xi_2}}\,d{\tau}d{t}\nonumber\\
&= \lim\limits_{T \rightarrow \infty} \frac{1}{{T}}\int_{-T/2}^{T/2} e^{2\pi i {t}(\xi_1 - \xi_2)} \int_{-T/2-t}^{T/2-t} k_\mu(\tau)  e^{-2\pi i \tau{\xi_2}}\,d{\tau}d{t}\nonumber\\
&=\lim\limits_{T \rightarrow \infty} \frac{1}{{T}}\int_{-T/2}^{T/2} e^{2\pi i {t}(\xi_1 - \xi_2)} \int_{-\infty}^{\infty} \wh k_\mu(\xi + \xi_2) \left(i \frac{e^{-2\pi i (T/2+t)\xi} - e^{2\pi i(T/2-t)\xi }}{2\pi \xi}
\right) d{\xi} d{t}, \nonumber
\end{align}
where the last equality above used Plancherel theorem.
Now we switch the order of two integrals,
\begin{align*}
\E[ \hat y_T(\xi_1)^* \hat y_T(\xi_2) ] &=  \frac{1}{{T}}\int_{-T/2}^{T/2} e^{2\pi i {t}(\xi_1 - \xi_2)} \int_{-\infty}^{\infty} p_\mu(\xi + \xi_2) \left(i \frac{e^{-2\pi i (T/2+t)\xi} - e^{2\pi i(T/2-t)\xi }}{2\pi \xi}
\right) d{\xi}dt\\
&= i \int_{-\infty}^{\infty} \frac{p_\mu(\xi + \xi_2)}{{2\pi \xi}} \left( {e^{-2\pi i (T/2)\xi} - e^{2\pi i(T/2)\xi }}
\right)\int_{-T/2}^{T/2} \frac{e^{2\pi i {t}(-\xi + \xi_1 - \xi_2)}}{{T}} d{t} d{\xi}\\
&= \int_{-\infty}^{\infty} p_\mu(\xi + \xi_2) \frac{\sin{(\pi T\xi )}}{{\pi \xi}} \cdot \frac{\sin{(\pi T (\xi - \xi_1 + \xi_2))}}{\pi T (\xi - \xi_1 + \xi_2)} d{\xi}\\
&= \int_{-\infty}^{\infty} T \sinc{(T (\xi - \xi_1 + \xi_2))}  {\sinc{(T\xi )}}\, p_\mu(\xi + \xi_2) d{\xi}
\end{align*}
Let $\tilde\mu$ be the measure which induces the probability density $p_\mu(\cdot + \xi_2)$, hence it satisfies $d{\tilde{\mu}}(\xi) = p_\mu(\xi + \xi_2) d{\xi}$. 
Now if we take the limit of covariance as $T \rightarrow \infty$ we get that,
$$\lim\limits_{T \rightarrow \infty} \E[ \hat y_T(\xi_1)^* \hat y_T(\xi_2) ] = \lim\limits_{T \rightarrow \infty} \int_{-\infty}^{\infty} T \sinc{(T (\xi - \xi_1 + \xi_2))} \cdot {\sinc{(T\xi )}}\, d{\tilde\mu}(\xi)$$
For ease of notation we call the integrand in above $f_T(\xi) = T \sinc{(T (\xi - \xi_1 + \xi_2))}  {\sinc{(T\xi )}}$. Remember, the assumption is that $\xi_1 \ne \xi_2$. The sequence $\{f_T(\xi)\}$ converges pointwise to zero for all $\xi \in \RR\setminus\{\xi_1,\xi_2\}$. On points $\xi_1,\xi_2$ it is also bounded by $\frac{1}{|\xi_2-\xi_1|}$. Therefore, the sequence $\{f_T(\xi)\}$ converges pointwise to zero $\tilde{\mu}$-almost everywhere.
Also the seguence $\{f_T\}$ is $\tilde{\mu}$-almost dominated by an integrable function g in the sense that for all $T \ge 1$,
$$|f_T(\xi)| \le g(\xi)$$
$g$ exists since $|f_T(\xi)| \le T\cdot  \frac{2}{T|\xi|+1} \cdot \frac{2}{T|\xi - \xi_1 + \xi_2|+1} \eqdef h_T(\xi)$ for every $\xi \in \RR$ and $h_T(\xi)$ is monotonely converging to zero for $\tilde{\mu}$-almost every $\xi$ and $h_T(\xi)$ is integrable for every value of $T$.
Therefore by Lebesgue's dominated convergence theorem we have,
\begin{align*}
\lim\limits_{T \rightarrow \infty} \E[ \hat y_T(\xi_1)^* \hat y_T(\xi_2) ] &= \lim\limits_{T \rightarrow \infty} \int_{-\infty}^{\infty} f_T(\xi)\, d{\tilde\mu}\\
&= \int_{-\infty}^{\infty} \lim\limits_{T \rightarrow \infty} f_T(\xi)\, d{\tilde\mu}\\
&=0.
\end{align*}
Finally note that the limit of the diagonal entries of the covariance, $ \lim\limits_{T \rightarrow \infty} \E[ |\hat y_T(\xi)|^2] = p_\mu(\xi)$ for every $\xi \in \RR$ by the Wiener-Khintchine-Einstein
Theorem~\cite{miller2012probability}.
\end{proof}

It is well known \cite{RasmussenWilliams06} that for $y$ distributed as in Definition \ref{def:gaussProcess}, the posterior distribution of $y$ given  samples $t_1,\ldots t_s \in [0,T]$ is also a Gaussian process. Its mean (the Bayes MMSE estimator) and its mode (the MAP estimator) coincide and are given by:
\begin{theorem}[Gaussian Process Prior Signal Estimation -- Finite Samples]\label{prob:bayes}
Consider $y$ distributed as in Definition \ref{def:gaussProcess} and noise $n$ distributed as a Gaussian process covariance $\epsilon \cdot \bv{I}$. %Let $ o_T: [0,T]\rightarrow \RR$ be the observable function defined on $[0,T]$ with $ o_T(t) = y(t) + n(t)$.
Given $t_1,\ldots,t_s \in [0,T]$, let $\bv{y},\bv{n} \in \RR^{s}$ be given by $\bv{y}(i) = y(t_i)$ and $\bv{n}(t) =  n(t_i)$. Let $\bv{F}: \CC^s \rightarrow L_2(\mu)$ be the operator defined by $[\bv F g](\xi) = \sum_{j=1}^s g(j) e^{-2\pi i \xi t_j}$.
Both the MAP and MMSE estimates for $y$ are given by $\tilde y  = \Fmu^* \tilde g$ where:
	\begin{align*}
		\tilde g = \argmin_{g\in L_2(\mu)}\left[ \frac{1}{s} \| \bv{F}^* g - (\bv{y+n})\|_2^2 + \epsilon\|g\|_\mu^2\right].
	\end{align*}
\end{theorem}
\begin{proof}
Letting $\bv{K}  = \bv{F}^* \bv{F}$, so $\bv{K}(i,j) = k_\mu(t_i,t_j)$, it is well known \cite{RasmussenWilliams06} that the posterior distribution of $y$ given $t_1,\ldots,t_s$ is a Gaussian process with mean $\tilde y(t)$ given by:
\begin{align*}
\tilde y(t) = \bv{k}_t^*(\bv{K}+\epsilon \bv{I})^{-1} (\bv{y}+\bv{n}).
\end{align*}
where $\bv{k}_t \in \RR^n$ is given by $\bv k_t(i) = k_\mu(t_i,t)$. It can be shown, analogously to the proof of Theorem \ref{thm:mainAlg}, that $\tilde y = \Fmu \tilde  g$ where
\begin{align*}
\tilde g = \argmin_{g\in L_2(\mu)}\left[  \frac{1}{s}\| \bv{F}^* g - (\bv{y+n})\|_2^2 + \epsilon\|g\|_\mu^2\right].
\end{align*}
Further, since $\tilde y$ is the mean of the posterior distribution, it gives the Bayes MMSE estimator, and since this posterior distribution is a Gaussian process, also gives the MAP estimator.
\end{proof}

We can see that the least squares problem~\eqref{eq:least_squares_setup} roughly corresponds to a limit of the finite sample optimization problem of Theorem \ref{prob:bayes} as the number of samples goes to infinity. Via Theorem \ref{thm:informal_main}, this optimization problem can be solved approximately with $\tilde O(\smu)$ samples
 using Algorithm \ref{alg:main}  and the universal sampling distribution of Theorem \ref{thm:fullBound}. Via Claim \ref{claim:regression_reduction} one can see that the lower bound of Section \ref{sec:lb} (Theorem \ref{thm:mainLB}) extends to solving \eqref{eq:least_squares_setup}, even approximately, and thus our sample complexity is nearly optimal.
\end{document}